\documentclass[11pt,a4paper]{article}
\usepackage[T1]{fontenc}
\usepackage{lmodern}
\usepackage[textwidth=165mm,textheight=245mm,centering]{geometry}
\usepackage{amsmath,amssymb,amsthm}
\usepackage{graphicx,booktabs,array,tabularx}
\usepackage{enumitem,needspace,placeins}
\usepackage[authoryear,round]{natbib}
\usepackage[font=small,labelfont=bf]{caption}
\usepackage[hidelinks]{hyperref}

\graphicspath{{figures/}}
\setlist[enumerate]{nosep,leftmargin=1.7em}
\theoremstyle{plain}
\newtheorem{theorem}{Theorem}
\newtheorem{lemma}{Lemma}
\newtheorem{corollary}{Corollary}
\newtheorem{proposition}{Proposition}
\theoremstyle{remark}
\newtheorem{remark}{Remark}

\newenvironment{keywords}{\par\smallskip\noindent\textbf{Keywords:}\ }
  {\par\medskip}
\renewcommand{\paragraph}[1]{\par\addvspace{6.5pt}\noindent
  \textit{#1}\quad\ignorespaces}
\makeatletter
\newcommand{\proofanchor}[1]{\phantomsection\begingroup
  \protected@edef\@currentlabel{\thesection}\label{#1}\endgroup}
\makeatother

\hypersetup{pdftitle={Adaptive Cauchy inference from invariant information},
  pdfauthor={Yiming Liu and Shaochen Wang}}

\newcommand{\E}{\mathbb E}
\newcommand{\PP}{\mathbb P}

\newcommand{\normal}{\mathsf N}
\newcommand{\Cauchy}{\operatorname{Cauchy}}

\title{Adaptive Cauchy inference from invariant information}
\author{YIMING LIU\\[2pt]
  {\small School of Economics, Jinan University,}\\
  {\small Guangzhou 510632, China}\\
  {\small \href{mailto:liuyiming@jnu.edu.cn}{\texttt{liuyiming@jnu.edu.cn}}}
  \and
  SHAOCHEN WANG\\[2pt]
  {\small School of Mathematics,}\\
  {\small South China University of Technology,}\\
  {\small Guangzhou 510640, China}\\
  {\small \href{mailto:mascwang@scut.edu.cn}{\texttt{mascwang@scut.edu.cn}}}}
\date{}
\begin{document}
\maketitle
\begin{abstract}
Selecting weights and projection directions from the observations used
for testing can invalidate a reference distribution. We develop a
framework for exact conditional inference under rotation symmetry.
A maximal invariant permits adaptation while leaving a uniform common
phase for calibration. This yields explicit Cauchy references for
invariantly selected weights and projections. Under separate block
rotations, the same principle gives exact normal-score combinations
with coefficients selected from all blockwise invariants. The calibration
results require neither densities nor moments. To illustrate the framework, we study direction learning and power
in spherical Gaussian location models, identifying a regime in
which consistent detection is possible without consistent direction
recovery. For heterogeneous functional data, a matched mean--residual
construction produces the required null symmetry under a conditional
Gaussian model. Simulations compare Radius and Blend, two
representative procedures, with sample splitting, and an electrocardiogram
analysis illustrates their application.\end{abstract}


\section{Introduction}

\citet{PillaiMeng2016} proved that, if $\mathbf X=(X_j)_{j=1}^m,\mathbf Y=(Y_j)_{j=1}^m$ are
independent copies of
$\normal_m(\mathbf 0,\boldsymbol\Sigma)$, where
$\boldsymbol\Sigma$ is positive semidefinite with strictly positive
diagonal entries, then
\begin{equation}\label{eq:PM-intro}
  \sum_{j=1}^m w_j\frac{X_j}{Y_j}\sim\Cauchy(0,1),
  \qquad
  w_j\geq0,\quad \sum_{j=1}^m w_j=1,
\end{equation}
provided that the weights $w_j$ are independent of
$(\mathbf X,\mathbf Y)$. Notably, no independence assumption is imposed
on the coordinate ratios $X_j/Y_j$.
\citet{Pillai2016} obtained related non-Gaussian identities for a class
of distributions invariant under a common rotation. The Cauchy
distribution has also been used to aggregate dependent evidence.
\citet{LiuXie2020}, for example, proposed a Cauchy transformation for
combining dependent $p$-values, with a tractable tail approximation,
while \citet{LiuMengPillai2025} considered confidence regions based on
heavy-tailed combining rules.

The independence requirement on the weights in \eqref{eq:PM-intro},
however, separates selection from inference by construction. This can be
costly when the informative coordinates or projection directions are
unknown before the data are observed. One could use an independent
sample for selection, but doing so sacrifices observations that could
otherwise contribute to inference. Using the same observations is more
attractive, yet arbitrary data-dependent weighting need not preserve the
Cauchy law. The question is therefore not whether the weights must be
fixed, but what information in the observations may be used for
adaptation while leaving an exact reference distribution for inference.

There are several ways to protect inference from data-dependent
selection. External covariates can guide hypothesis weights
\citep{IgnatiadisEtAl2016}, and sample splitting can reserve independent
observations for selection and testing \citep{LiuYuLi2022}. Invariance
provides a route that can instead reuse the same observations.
\citet{Lauter1998,LauterGlimmKropf1999}, for example,
constructed exact multivariate tests with data-dependent linear scores
under left-spherical symmetry. These results establish that same-observation selection can be
compatible with exact inference under suitable symmetry. Here we
exploit the common-rotation structure of paired observations to
obtain an explicit conditional Cauchy reference for data-dependent
weights and projections.

The key is to distinguish dependence on the observations from
dependence on the phase used for calibration. Writing
$\mathbf G=\mathbf Y+i\mathbf X$, the matrix
$\mathbf M=\mathbf G\mathbf G^*$ records the paired energies and
relative coordinate angles. Under common-rotation symmetry, a
uniform phase remains after conditioning on $\mathbf M$.
Selection based on $\mathbf M$ can therefore use the observations
without selecting a favorable realization of that phase.

Our contributions are threefold.
\begin{enumerate}
\item \emph{Exact adaptation using the same observations.}
Under rotation symmetry, supports, projections and admissible weights
can be selected from invariant information in the observations used
for testing. We obtain explicit conditional Cauchy references and,
under separate block rotations, exact normal-score combinations with
jointly selected coefficients. Neither densities nor moments are required.

\item \emph{From invariant information to test power.}
Spherical Gaussian location models provide a tractable illustration.
Explicit population gaps, the complete-invariant likelihood and power
formulas connect direction learning to the final test. The analysis
separates weighting efficiency from direction alignment and identifies
a regime of consistent detection without consistent direction recovery,
complementing the general orbit-recovery perspective \citep{Bandeira2023}.

\item \emph{Rotation symmetry by construction.}
For heterogeneous functional data, matched mean--residual summaries
produce the required null symmetry under a conditional Gaussian model.
This provides a route from the original observations to exact adaptive
inference while allowing different covariance structures across batches.
\end{enumerate}

Radius and Blend serve as representative implementations. Comparisons
with sample splitting examine the tradeoff between learning from all
invariant summaries and using unrestricted training observations while
retaining fewer blocks for testing. Simulations explore this tradeoff
across dimensions, covariance structures and noise laws, and an ECG
analysis illustrates the procedures.

\textit{Notations}:
For $\mathbf g=(g_j)_{j=1}^m\ne\mathbf0$, define its active support and the corresponding simplex by
\begin{equation*}
  \mathcal A(\mathbf g)=\{j:|g_j|>0\},\qquad
  \Delta(\mathbf g)=\left\{\mathbf w\in[0,1]^m:
    \sum_{j=1}^m w_j=1,\quad w_j=0\text{ for }j\notin\mathcal A(\mathbf g)
  \right\},
\end{equation*}
which is used to handle structural zeros.
Moreover, define $q(z)=\operatorname{Im}(z)/\operatorname{Re}(z)$ when
$\operatorname{Re}(z)\ne0$, and set $q(z)=0$ otherwise.

\section{Exact adaptation under rotation symmetry}
\label{sec:calibration}

\subsection{Conditional Cauchy calibration}
\label{subsec:singlecalibration}

We first identify the information that can be used to select weights
and projections within a single pair. We then extend the construction
to joint selection and evidence combination across pairs.

Let $\mathbf X,\mathbf Y\in\mathbb R^m$ and write
$\mathbf G=\mathbf Y+i\mathbf X$. For $\phi\in[0,2\pi)$, define
\[
  R_\phi(\mathbf X,\mathbf Y)
  =
  \bigl(
    \mathbf X\cos\phi+\mathbf Y\sin\phi,\,
    \mathbf Y\cos\phi-\mathbf X\sin\phi
  \bigr).
\]
This transformation is equivalent to
$\mathbf G\mapsto e^{i\phi}\mathbf G$ and leaves
\begin{equation}\label{eq:gram}
  \mathbf M=\mathbf G\mathbf G^*
  =\mathbf X\mathbf X^\top+\mathbf Y\mathbf Y^\top
   +i(\mathbf X\mathbf Y^\top-\mathbf Y\mathbf X^\top)
\end{equation}
unchanged, where $*$ denotes conjugate transpose. For
$\mathbf G\ne\mathbf0$, \(\mathbf M\) is a maximal invariant in the sense that two complex vectors have the same matrix \(\mathbf M\) if and only if they differ only by a common phase factor.
Consequently, $\mathbf M$ determines the coordinate magnitudes and
relative phases, but not the common phase.
We assume that the joint distribution of $(\mathbf X,\mathbf Y)$
is invariant under these rotations, i.e.,
\begin{equation}\label{eq:commonrotation}
  R_\phi(\mathbf X,\mathbf Y)
  \overset d=(\mathbf X,\mathbf Y)
  \quad\text{for every }\phi\in[0,2\pi).
\end{equation}
We refer to \eqref{eq:commonrotation} as
\emph{common-rotation symmetry}.
Table~\ref{tab:rotation-examples} below lists several classes
of distributions satisfying common-rotation symmetry,
including joint elliptical pairs with a common scatter
structure, shared covariance mixtures, and circular
coordinate constructions.
Common-rotation symmetry therefore requires neither
Gaussianity nor a spherical covariance structure.
\begin{table}[!htbp]
\centering
\caption{Examples satisfying common-rotation symmetry.}
\label{tab:rotation-examples}
\small
\begin{tabularx}{\textwidth}{@{}p{0.24\textwidth}X@{}}
\toprule
Model & Construction \\
\midrule

Joint elliptical pairs
&
$\mathbf X=R\boldsymbol\Sigma^{1/2}\mathbf U_1$ and
$\mathbf Y=R\boldsymbol\Sigma^{1/2}\mathbf U_2$,
where
$\mathbf U=(\mathbf U_1^\top,\mathbf U_2^\top)^\top$
is uniform on the unit sphere in $\mathbb R^{2m}$,
$R\geq0$ is independent of $\mathbf U$, and
$\boldsymbol\Sigma$ is positive semidefinite.

\\[5pt]

Shared covariance mixtures
&
Conditional on a random positive semidefinite matrix
$\boldsymbol\Omega$, the vectors $\mathbf X$ and $\mathbf Y$
are independent $\mathcal N_m(\mathbf0,\boldsymbol\Omega)$.
Taking $\boldsymbol\Omega=S^2\boldsymbol\Sigma$ gives
shared scale mixtures.
\\[5pt]

Circular coordinate pairs
&
$X_j=R_j\sin\Theta_j$ and $Y_j=R_j\cos\Theta_j$,
$j=1,\ldots,m$, where the angles are independent
$\operatorname{Unif}[0,2\pi)$ variables and are independent
of the nonnegative radius vector
$\mathbf R=(R_1,\ldots,R_m)^\top$.
The radii may be dependent, random, or fixed.
\\

\bottomrule
\end{tabularx}
\end{table}

Given the common-rotation symmetry,  the following theorem does not require the weights to be independent of the observations.
\begin{theorem}
\label{thm:cauchy}
Suppose $\mathbb P(\mathbf G=\mathbf0)=0$,
$\mathbf W\in\Delta(\mathbf G)$ almost surely, and
\begin{equation}\label{eq:jointrotation}
  (e^{i\phi}\mathbf G,V,\mathbf W)
  \overset d=(\mathbf G,V,\mathbf W)
  \quad\text{for every }\phi\in[0,2\pi),
\end{equation}where $V$ is an auxiliary random variable.
Then, for $T=\sum_{j=1}^m W_jq(G_j)$,
\begin{equation}\label{eq:conditionalcauchy}
  \mathcal L(T\mid\mathbf M,V,\mathbf W)
  =\operatorname{Cauchy}(0,1)
  \quad\text{almost surely}.
\end{equation}
Consequently, $T$ is independent of $(\mathbf M,V,\mathbf W)$.
In particular, suppose
$(e^{i\phi}\mathbf G,V)\overset d=(\mathbf G,V)$ for every $\phi$.
If the weights $\mathbf W$ are any measurable function of
$(\mathbf M,V)$ taking values in $\Delta(\mathbf G)$, then
\[
  \mathcal L(T\mid\mathbf M,V)
  =\operatorname{Cauchy}(0,1)
  \quad\text{almost surely}.
\]
\end{theorem}

\begin{remark}
    A weight-selection rule may also use an auxiliary input, such as an
independent random seed. We denote this input by $V$. When no
additional input is used, $V$ can be taken as constant. To see why the theorem holds, write the nonzero coordinates as
\(
  G_j=r_je^{i(\Theta+u_j)},
\)
where the radii $r_j$ and relative phases $u_j$ are determined by
$\mathbf M$. Conditional on $(\mathbf M,V,\mathbf W)$, only the common phase
$\Theta$ remains random, and it is uniform on $[0,2\pi)$. Hence
\[
  T=\sum_{j:r_j>0}W_j\tan(\Theta+u_j)
  \quad\text{almost surely}.
\]
The classical shifted-tangent identity  \citep{PitmanWilliams1967,Williams1969} gives
\eqref{eq:conditionalcauchy}. The complete conditioning argument is
provided in Supplementary Section~\ref{proof:cal-single}.

\end{remark}
\begin{remark}
    A simple invariant weighting rule uses paired coordinate energies. Define
\[
  E_j=X_j^2+Y_j^2=M_{jj},\qquad
  W_j=\frac{E_j}{\sum_{\ell=1}^m E_\ell},
  \qquad j=1,\ldots,m.
\]
The denominator is positive whenever $\mathbf G\ne\mathbf0$.
This rule assigns more weight to coordinates with larger paired
squared magnitudes. The weights are computed from the same
observations as the ratios, but depend only on the diagonal entries
of $\mathbf M$. Theorem~\ref{thm:cauchy} therefore gives
\[
  \mathcal L\left(
    \sum_{j=1}^m
    \frac{E_j}{\sum_{\ell=1}^m E_\ell}q(G_j)
    \,\middle|\,\mathbf M
  \right)
  =\operatorname{Cauchy}(0,1),
\]
Invariant supports may likewise be selected before assigning the active weights.
\end{remark}

The same principle permits direction selection from the invariant information. In particular, projecting both members of the pair onto an invariantly selected direction preserves the conditional Cauchy law.
\begin{corollary}
\label{cor:projections}
Suppose $\mathbb P(\mathbf G\ne\mathbf0)=1$ and
\[
  (e^{i\phi}\mathbf G,V)
  \overset d=(\mathbf G,V)
  \quad\text{for every }\phi\in[0,2\pi).
\]
Let $\widehat{\mathbf a}$ be a real unit vector that is a
measurable function of $(\mathbf M,V)$, and define
$\widetilde X=\widehat{\mathbf a}^{\top}\mathbf X$ and
$\widetilde Y=\widehat{\mathbf a}^{\top}\mathbf Y$.
On the event
\[
  \widetilde X^2+\widetilde Y^2
  =
  \widehat{\mathbf a}^{\top}
  \operatorname{Re}(\mathbf M)
  \widehat{\mathbf a}>0,
\]
we have
\begin{equation}\label{eq:selectedprojection}
  \mathcal L\{q(\widetilde Y+i\widetilde X)\mid\mathbf M,V\}
  =\operatorname{Cauchy}(0,1)
  \quad\text{almost surely}.
\end{equation}
\end{corollary}

This corollary provides the exact calibration for the adaptive projections studied in Section~\ref{sec:learning}.
Extensions to weighted combinations of finitely many selected projections and to Hilbert spaces are given in Supplementary Propositions~\ref{proof:cal-finite-projections} and~\ref{proof:cal-hilbert}.

\subsection{Joint adaptation and evidence combination}
\label{subsec:jointadaptation}

With several paired observations, a common direction may be learned
from all pairs rather than from each pair separately. We therefore
need a joint reference distribution that remains valid after such
selection. Let $\mathbf X_b,\mathbf Y_b\in\mathbb R^{m_b}$ denote
the observations in block $b$, and write
$\mathbf G_b=\mathbf Y_b+i\mathbf X_b$, $b=1,\ldots,B$.
We assume that each pair can be rotated by its own angle without
changing the joint distribution
\begin{equation}\label{eq:product}
  (e^{i\phi_1}\mathbf G_1,\ldots,e^{i\phi_B}\mathbf G_B,V)
  \overset d=(\mathbf G_1,\ldots,\mathbf G_B,V)
  \quad\text{for every }(\phi_1,\ldots,\phi_B)\in[0,2\pi)^B,
\end{equation}
where $V$ is the auxiliary random variable introduced in Section~\ref{subsec:singlecalibration} and may
be taken as constant. The angles need not be equal.
For each block, set $\mathbf M_b=\mathbf G_b\mathbf G_b^*$ and let
\[
  \mathcal I_B=\sigma(\mathbf M_1,\ldots,\mathbf M_B,V).
\]
This notation collects all information in the blockwise matrices
and the auxiliary input. A rule measurable with respect to
$\mathcal I_B$ may use invariant information from every block.
Under \eqref{eq:product}, the common phases of the nonzero pairs are independent and uniform conditional on $\mathcal I_B$, which supplies the joint reference
needed below. Specifically, write $\mathcal B_* = \{b:\mathbf G_b\ne\mathbf0\}$
for the blocks that contain a nonzero pair, and let $\Phi$ denote
the standard normal distribution function.

\begin{theorem}\label{thm:joint}
Suppose \eqref{eq:product} holds. For each $b\in\mathcal B_*$,
let $\mathbf W_b=(W_{b1},\ldots,W_{b,m_b})$ be $\mathcal I_B$-measurable and belong to
$\Delta(\mathbf G_b)$. Define
\begin{equation*}
  T_b=\sum_{j=1}^{m_b}W_{bj}q(G_{bj}),\qquad
  p_b=1-\frac2\pi\arctan|T_b|.
\end{equation*}
Conditional on $\mathcal I_B$, the $T_b$ are independent standard
Cauchy variables, and  consequently the $p_b$ are independent
$\operatorname{Unif}(0,1)$ variables across blocks.

Let $\widehat{\mathbf a}_b\in\mathbb R^{m_b}$ be
$\mathcal I_B$-measurable and satisfy
\[
  \widehat{\mathbf a}_b^{\top}
  \operatorname{Re}(\mathbf M_b)\widehat{\mathbf a}_b>0,
  \qquad b\in\mathcal B_*.
\]
Then the signed scores
\begin{equation*}
  Z_b^{\mathrm{loc}}
  =\operatorname{sgn}(\widehat{\mathbf a}_b^{\top}\mathbf X_b)
    \Phi^{-1}(1-p_b/2)
\end{equation*}
are conditionally independent standard normal variables. Consequently,
for real $\mathcal I_B$-measurable coefficients $\lambda_b$, on
the event $\sum_{b\in\mathcal B_*}\lambda_b^2>0$,
\begin{equation}\label{eq:combination}
  N_B=
  \frac{\sum_{b\in\mathcal B_*}\lambda_bZ_b^{\mathrm{loc}}}
       {\{\sum_{b\in\mathcal B_*}\lambda_b^2\}^{1/2}},
  \qquad
  \mathcal L(N_B\mid\mathcal I_B)=\mathcal N(0,1).
\end{equation}
The same conclusions apply after forming finitely many linear
projections in each block using directions selected from
$\mathcal I_B$, with zero projected pairs omitted and conditioning
still on the original $\mathcal I_B$.
\end{theorem}

\begin{remark}
    The weights
$W_{bj}$ combine ratios within block $b$, whereas $\lambda_b$ combine evidence across blocks. Both may depend on invariant information from all blocks.
For location alternatives, retaining the projected sign allows evidence in a common direction to accumulate across blocks.
Under the null hypothesis, Theorem~\ref{thm:joint} gives
\(
  N_B\mid\mathcal I_B\sim\mathcal N(0,1)
\)
whenever at least one nonzero projected pair has a nonzero coefficient. Hence the two-sided $p$-value is
\[
  p_{\mathrm{comb}}
  =2\{1-\Phi(|N_B|)\},
\]
where $\Phi$ denotes the standard normal distribution function. If no such contribution remains, we set
$p_{\mathrm{comb}}=1$, so that the test does not reject. The proof of Theorem~\ref{thm:joint}, including the
conditional sign symmetry used in the signed-normal construction, is given in Supplementary Section~\ref{proof:cal-joint}.

\end{remark}

Theorem~\ref{thm:joint} establishes exact calibration for invariantly selected directions and weights.
The next section studies how these choices affect power through a representative construction using
one projection per block.

\section{Representative procedures and power}
\label{sec:learning}

\subsection{Representative adaptive procedures}
\label{subsec:procedures}

Section~\ref{sec:calibration} permits both projection directions and block coefficients to be selected from invariant information. We now study how these two choices affect power.
For illustration, we use one common projection direction across blocks and compare two coefficient rules, named {\it Radius} and {\it Blend}.  The Gaussian models introduced in Section~\ref{subsec:learnedpower} provide explicit power formulas but are not required for exact null calibration.

Suppose that all $B$ pairs lie in the same space $\mathbb R^m$, where $m\ge2$, and let $\mathcal I_B=\sigma(\mathbf M_1,\ldots,\mathbf M_B,V)$, with $V$ constant or independent parameter-free randomization. We use $\widehat{\mathbf a}_B\in\mathbb R^m$ for a single unit direction learned jointly from the $B$ blocks.  In contrast, $\lambda_b\in\mathbb R$ is a scalar coefficient assigned to block $b$. Both procedures below use the same direction and differ only in these coefficients.

We first construct the common direction. Define the total paired
energy by
\begin{equation}\label{eq:pairedenergy}
  E_b=\|\mathbf X_b\|^2+\|\mathbf Y_b\|^2
     =\operatorname{tr}(\mathbf M_b),
\end{equation}
and write $\mathcal B_*=\{b:E_b>0\}$ for the nonzero blocks, as in Section~\ref{subsec:jointadaptation}. If this set is empty, set the final $p$-value to one. Otherwise, form
\begin{equation}\label{eq:direction}
  \mathbf C_b
  =\mathbf X_b\mathbf X_b^\top+\mathbf Y_b\mathbf Y_b^\top
  =\operatorname{Re}(\mathbf M_b),
  \qquad
  \mathbf S_b=\frac{\mathbf C_b}{E_b},
  \quad b\in\mathcal B_*,
\end{equation}
with $\mathbf S_b=\mathbf0$ for $b\notin\mathcal B_*$, and let
\begin{equation}\label{eq:directionaverage}
  \widehat{\mathbf A}_B
  =\frac1B\sum_{b=1}^B\mathbf S_b.
\end{equation}
Choose $\widehat{\mathbf a}_B$ to be a leading unit eigenvector of $\widehat{\mathbf A}_B$, with fixed measurable sign and tie conventions.
Equivalently,
\begin{equation*}
  \widehat{\mathbf a}_B
  \in\arg\max_{\|\mathbf a\|=1}
  \frac1B\sum_{b\in\mathcal B_*}
  \frac{(\mathbf a^\top\mathbf X_b)^2+
        (\mathbf a^\top\mathbf Y_b)^2}{E_b}.
\end{equation*}
The direction therefore maximizes the average fraction of paired
energy captured across blocks. Dividing by $E_b$ prevents overall
block scale alone from determining its contribution to this
criterion. The selected direction is $\mathcal I_B$-measurable.

Next, project every block onto this same direction and define
\begin{equation}\label{projected_XY}
  \xi_b=\widehat{\mathbf a}_B^\top\mathbf X_b,
  \qquad
  \eta_b=\widehat{\mathbf a}_B^\top\mathbf Y_b,
  \qquad
  r_b=(\xi_b^2+\eta_b^2)^{1/2}.
\end{equation}
For $r_b>0$, use the ratio $T_b=\xi_b/\eta_b$ and the sign
$\operatorname{sgn}(\xi_b)$ in Theorem~\ref{thm:joint}.
The signed-normal score is then
\begin{equation}\label{eq:angular}
  Z_b^{\mathrm{loc}}(\widehat{\mathbf a}_B)
  =\Phi^{-1}\left\{
    \frac12+\frac1\pi\arcsin\left(\frac{\xi_b}{r_b}\right)
  \right\}.
\end{equation}
The direction determines which component of each vector is tested.
The coefficients determine how the resulting block scores are
combined. In the present notation, \eqref{eq:combination} becomes
\[
  N_B
  =\frac{\sum_{b:r_b>0}\lambda_b
                 Z_b^{\mathrm{loc}}(\widehat{\mathbf a}_B)}
         {\{\sum_{b:r_b>0}\lambda_b^2\}^{1/2}}.
\]
Here the $\lambda_b$ are $\mathcal I_B$-measurable coefficients
and need not sum to one. The denominator normalizes the statistic,
and multiplying all coefficients by a common positive constant
leaves $N_B$ unchanged. A zero denominator gives a final $p$-value
of one.

{\it Radius} assigns the projected radius to each block,
\begin{equation}\label{eq:radstat}
  \lambda_b^{\mathrm{rad}}=r_b,
  \qquad
  N_B^{\mathrm{rad}}
  =\frac{\sum_{b:r_b>0}r_b
                 Z_b^{\mathrm{loc}}(\widehat{\mathbf a}_B)}
         {\{\sum_{b:r_b>0}r_b^2\}^{1/2}}.
\end{equation}
For a fixed projection with homogeneous Gaussian noise, this rule maximizes local sensitivity among nonnegative functions of the projected radius.  Proposition~\ref{proof:fixedpower} in Supplementary states this restricted comparison, which motivates the coefficient rule but does not establish optimality over adaptive tests.

On the other hand, {\it Blend} adjusts the radius coefficients for differences
in noise variance across blocks. Specifically, under a spherical Gaussian null, block $b$ has noise covariance
$s_b^2\mathbf I_m$, whose inverse noise variance
$s_b^{-2}$, referred to as block precision, is estimated by
\[
  \widehat\pi_b=\frac{2m-2}{E_b}.
\]
Under a spherical Gaussian null with covariance $s_b^2\mathbf I_m$,
$E_b/s_b^2\sim\chi^2_{2m}$ and
$\mathbb E_0\widehat\pi_b=s_b^{-2}$.

The estimate $\widehat\pi_b$ uses only block $b$ and
can be noisy when $m$ is small.
We therefore combine it with a common reference obtained
by averaging the estimated noise variances across the
nonzero blocks and then taking the reciprocal.
Specifically, define
\begin{equation}\label{eq:blendestimates}
  \widehat J_{+,B}
  =\frac1{|\mathcal B_*|}\sum_{b\in\mathcal B_*}\frac{E_b}{2m},
  \qquad
  \widehat J_{-,B}
  =\frac1{|\mathcal B_*|}\sum_{b\in\mathcal B_*}\widehat\pi_b,
  \qquad
  \widehat P_B=\widehat J_{+,B}^{-1}.
\end{equation}
Here $\widehat P_B$ is the reciprocal of the estimated average noise variance, whereas $\widehat J_{-,B}$ is the average of the individual inverse-variance estimates. Blend combines the common reference $\widehat P_B$ with the block-specific estimate $\widehat\pi_b$, using the estimated scale heterogeneity to determine their relative contributions.
The population product of average variance and average inverse variance is at least one, with equality exactly when all block scales agree. This motivates the heterogeneity measure
\[
  \widehat{\mathcal H}_B
  =\max\{1,\widehat J_{+,B}\widehat J_{-,B}\}.
\]
The truncation enforces the population lower bound, which the empirical product need not satisfy.
The remaining choice is how much to rely on the individual inverse-variance estimates.
To make this choice explicit, consider the family
\[
  \lambda_b(t)
  =r_b\{(1-t)\widehat P_B+t\widehat\pi_b\},
  \qquad 0\le t\le1.
\]
At $t=0$, every radius receives the same multiplier, so the normalized statistic coincides with Radius.
At $t=1$, each radius is multiplied by its individual inverse-variance estimate. Intermediate values combine the common reference with the observed differences between blocks.

We select the mixing proportion using a Gaussian local power criterion, keeping the common direction
rule unchanged.
The population mixing proportion in Supplementary equation
\eqref{proof:blend-mixing-optimum} maximizes the leading local power
coefficient within this family.
Replacing the population heterogeneity by
$\widehat{\mathcal H}_B$ gives
\begin{equation}\label{eq:blendmixing}
  \widehat t_B
  =\frac{\widehat{\mathcal H}_B-1}
         {(2-\tau)\widehat{\mathcal H}_B-1},
  \qquad
  \tau=\frac{m-1}{m},
\end{equation}
where the factor $\tau$ arises from the Gaussian energy moments in this calculation.
The selected proportion is zero when $\widehat{\mathcal H}_B=1$ and increases as the
estimated heterogeneity increases, while remaining
strictly below one. Hence,  the resulting Blend coefficients are
\begin{equation}\label{eq:blendweights}
  \lambda_b^{\mathrm{blend}}
  =\lambda_b(\widehat t_B)
  =r_b\{(1-\widehat t_B)\widehat P_B
       +\widehat t_B\widehat\pi_b\},
  \qquad b\in\mathcal B_*.
\end{equation}
Substituting these coefficients into \eqref{eq:combination} gives $N_B^{\mathrm{blend}}$.

\begin{corollary}\label{cor:representative-calibration}
Under the product-rotation null symmetry in \eqref{eq:product},
both Radius and Blend have the conditional $\mathcal N(0,1)$
reference in \eqref{eq:combination} whenever their denominator
is positive. Zero projected pairs are omitted. If the denominator
is zero, the final $p$-value is one.
\end{corollary}
Both procedures select the direction and block coefficients using only $\mathcal I_B$, so their exact null calibration follows from Theorem~\ref{thm:joint}. This guarantee holds under the stated rotation symmetry, without requiring Gaussian observations, spherical covariance, or accurate variance estimates. Other choices based only on $\mathcal I_B$ have the same guarantee, although their power may differ. The next section studies power under Gaussian location models.

\subsection{Direction information and power}
\label{subsec:learnedpower}

We study how the common direction and the block coefficients defined
in Section~\ref{subsec:procedures} affect power. We first quantify
recovery of the normalized-scatter direction, then connect direction
error to the power of Radius. We keep the unit-scale analysis together
before examining the effect of unequal scales through Blend.

Suppose the pairs follow the Gaussian location model
\begin{equation}\label{eq:spherical}
  \mathbf X_b=\rho\mathbf v+s_b\boldsymbol\varepsilon_b,
  \qquad
  \mathbf Y_b=s_b\boldsymbol\zeta_b,
  \qquad b=1,\ldots,B,
\end{equation}
where $\|\mathbf v\|=1$, $m\ge2$, and all noise vectors
$\boldsymbol\varepsilon_b,\boldsymbol\zeta_b$ are mutually independent
$\mathcal N_m(\mathbf0,\mathbf I_m)$ variables.
The deterministic scales $s_b>0$ specify the noise standard deviations,
$\mathbf v$ is the unknown signal direction, and $\rho\ge0$ is the
signal magnitude. Under $\rho=0$, the model satisfies the
product-rotation symmetry in \eqref{eq:product}.
The Gaussian specification is used for the power calculations below,
while Corollary~\ref{cor:representative-calibration} continues to apply
under the broader rotation assumptions.
Write
\begin{equation*}
  L_B=\ell(\widehat{\mathbf a}_B,\mathbf v)
     =1-(\widehat{\mathbf a}_B^\top\mathbf v)^2
\end{equation*}
to measure the loss for estimating the signal line $\operatorname{span}(\mathbf v)$.
The recovery and detection results below concern the normalized-scatter
direction in \eqref{eq:directionaverage}. The power expansions also
allow this direction to be replaced by any common
$\mathcal I_B$-measurable unit direction, as specified at those results. The next theorem gives the population eigenvalue gap and bounds the expected loss of the normalized-scatter
estimator.

\begin{theorem}
\label{thm:recovery}
Under \eqref{eq:spherical}, let $\rho>0$ and use the
normalized-scatter direction from Section~\ref{subsec:procedures}. Define
\[
  g_m(t)=\mathbb E\left\{\frac{N_t}{m+N_t}\right\},
  \qquad N_t\sim\operatorname{Poisson}(t^2/2),
\]
where $N_t$ is the mixing variable for the noncentral chi-square
energy along $\mathbf v$. Set
\[
  t_b=\frac{\rho}{s_b},
  \qquad \bar g_B=\frac1B\sum_{b=1}^B g_m(t_b).
\]
Then
\begin{equation}\label{eq:populationdirection}
  \mathbb E\mathbf S_b
  =
  \frac{1-g_m(t_b)}{m}\mathbf I_m
  +
  g_m(t_b)\mathbf v\mathbf v^\top,
\end{equation}
and
\[
  \mathbb E\widehat{\mathbf A}_B
  =
  \frac{1-\bar g_B}{m}\mathbf I_m
  +
  \bar g_B\mathbf v\mathbf v^\top.
\]
The leading population line is $\operatorname{span}(\mathbf v)$,
with eigenvalue gap $\bar g_B$.
Moreover, if
\[
  \mathcal V_B
  =
  \frac1{B^2}
  \sum_{b=1}^B
  \mathbb E
  \left\|
    \mathbf S_b-\mathbb E\mathbf S_b
  \right\|_F^2,
\]
then
\begin{equation}\label{eq:directionrisk}
  \sup_{\|\mathbf v\|=1}
  \mathbb E_{\rho,\mathbf v}
  \ell(\widehat{\mathbf a}_B,\mathbf v)
  \le
  \min\left\{
    1,\,
    \frac{4\mathcal V_B}{\bar g_B^2}
  \right\}
  \le
  \min\left\{
    1,\,
    \frac{4}{B\bar g_B^2}
  \right\}.
\end{equation}
\end{theorem}
\begin{remark}\label{rem:gapexp}
Theorem~\ref{thm:recovery} shows that trace normalization preserves
the population signal line under spherical noise. For weak signals (i.e., $t_b=\rho/s_b\to0$),
$\mathbb P(N_t=1)=t^2/2+O(t^4)$ and
$\mathbb P(N_t\ge2)=O(t^4)$, so the Poisson representation gives
\begin{equation}\label{eq:gapexp}
  g_m(t)=\frac{t^2}{2(m+1)}+O(t^4),
  \qquad t\to0.
\end{equation}
For fixed $m$, unit scales, and $\rho=\rho_B\to0$,
$\bar g_B=g_m(\rho_B)\asymp\rho_B^2$. Hence
\eqref{eq:directionrisk} implies
\[
  \mathbb E_{\rho_B,\mathbf v}
  \ell(\widehat{\mathbf a}_B,\mathbf v)
  =O\!\left(\frac1{B\rho_B^4}\right),
\]
  so $B\rho_B^4\to\infty$ is sufficient for consistent recovery.
Supplementary Theorem~\ref{proof:info-risk} establishes its necessity for uniformly consistent recovery from the blockwise invariant information in this fixed-dimensional, unit-scale model.
The fourth-order information scale agrees with the invariant-moment perspective of \citet{Bandeira2023}.
\end{remark}

We next retain the Radius coefficients $\lambda_b=r_b$ and study how the selected direction affects power.
For illustration, we take $s_b=1$, fixed $m\ge2$, and $\rho=\rho_B\to0$.
The common direction $\widehat{\mathbf a}_B$ may now be any
$\mathcal I_B$-measurable unit vector, including the normalized-scatter
choice. In particular, it need not be independent of the observations
used to form the statistic.
For fixed $0<\alpha<1$, write
$z_\alpha=\Phi^{-1}(1-\alpha/2)$.
For a normal variable with mean $t$ and variance one, the two-sided rejection probability is
\[
  H_\alpha(t)=1-\Phi(z_\alpha-t)+\Phi(-z_\alpha-t).
\]
The next theorem connects this normal power function to the realized
alignment $1-L_B$ of a direction learned from the same observations.

\begin{theorem}
\label{thm:bridge}
Uniformly over the preceding $\mathcal I_B$-measurable unit-direction
rules,
\begin{align}
 & \mu_B
  =
  \mathbb E_\rho
  \left(
    N_B^{\mathrm{rad}}
    \mid\mathcal I_B
  \right)=
  \kappa_R
  \sqrt B\,\rho_B\,
  \widehat{\mathbf a}_B^\top\mathbf v
  +
  O_P(\sqrt B\,\rho_B^3+\rho_B),
  \label{eq:bridgemean}\\
  &\operatorname{Var}_\rho
  (N_B^{\mathrm{rad}}\mid\mathcal I_B)
  =
  1+O_P(\rho_B^2),
  \label{eq:bridgevar}\\
  &\sup_{x\in\mathbb R}
  \left|
    \mathbb P_\rho
    (N_B^{\mathrm{rad}}\le x\mid\mathcal I_B)
    -
    \Phi(x-\mu_B)
  \right|
  =
  O_P(B^{-1/2}+\rho_B^2).
  \label{eq:bridgecdf}
\end{align}
If, in addition,
$\sqrt B\,\rho_B^3\to0$, then
\begin{equation}\label{eq:powerbridge}
  \mathbb P_\rho
  \left(
    |N_B^{\mathrm{rad}}|>z_\alpha
  \right)
  =
  \mathbb E_\rho
  H_\alpha
  \left(
    \kappa_R
    \sqrt B\,\rho_B
    \sqrt{1-L_B}
  \right)
  +o(1),
\end{equation}
where
\begin{equation*}
  \beta=-\int_0^1\Phi^{-1}(u)\cos(\pi u)\,du\simeq0.671,
  \qquad \kappa_R=\sqrt2\,\beta.
\end{equation*}
\end{theorem}

\begin{remark}
The leading conditional mean is proportional to
$\widehat{\mathbf a}_B^\top\mathbf v$. Since $H_\alpha$ is even,
the two-sided power depends on its magnitude,
$|\widehat{\mathbf a}_B^\top\mathbf v|=\sqrt{1-L_B}$.
Equation~\eqref{eq:powerbridge} therefore averages the normal power
at shift magnitude $\kappa_R\sqrt B\rho_B\sqrt{1-L_B}$ over the
random alignment. Better alignment increases this leading power,
but perfect direction recovery is not required for detection.
\end{remark}

For $B\ge m$, the normalized-scatter direction has a Beta
null squared-alignment distribution, as shown in Supplementary
Lemma~\ref{proof:det-null-line}.
The next corollary uses this law to obtain explicit local power
and separates the detection and direction-recovery regimes.

\begin{corollary}
\label{cor:localpower}
\label{cor:detection}
Under \eqref{eq:spherical}, let $m\ge2$ be fixed,
$s_b=1$, and $\rho_B\to0$.
Use the normalized-scatter direction
$\widehat{\mathbf a}_B$ and the Radius coefficients.
For fixed $0<\alpha<1$, the following conclusions hold.

\begin{enumerate}
\item If $B\rho_B^2\to0$, then
\[
  \mathbb P_{\rho_B,\mathbf v}
  \bigl(|N_B^{\mathrm{rad}}|>z_\alpha\bigr)
  \longrightarrow\alpha.
\]

\item If $\rho_B=h/\sqrt B$ for fixed $h>0$, then
\begin{equation*}
  \mathbb P_{\rho_B,\mathbf v}
  \bigl(|N_B^{\mathrm{rad}}|>z_\alpha\bigr)
  \longrightarrow
  \mathbb E H_\alpha(\kappa_R h\sqrt U),
  \qquad
  U\sim\operatorname{Beta}
  \left(\frac12,\frac{m-1}{2}\right).
\end{equation*}

\item If $B\rho_B^2\to\infty$, then
\[
  \mathbb P_{\rho_B,\mathbf v}
  \bigl(|N_B^{\mathrm{rad}}|>z_\alpha\bigr)
  \longrightarrow1.
\]
If additionally $B\rho_B^4\to0$, then
\[
  \mathbb E_{\rho_B,\mathbf v}L_B
  \longrightarrow1-\frac1m.
\]
\end{enumerate}
\end{corollary}

\begin{remark}
The last conclusion shows that consistent detection is possible even when the selected direction retains the same limiting expected loss as under the null.
\end{remark}

We now compare Radius and Blend under unequal noise scales, using the same common $\mathcal I_B$-measurable unit direction. Define the population counterparts of the scale summaries in Section~\ref{subsec:procedures} by
\begin{equation}\label{eq:blendpopulation}
  J_{+,B}=\frac1B\sum_b s_b^2,
  \qquad
  J_{-,B}=\frac1B\sum_b s_b^{-2},
  \qquad
  \mathcal H_B=J_{+,B}J_{-,B}.
\end{equation}
With $\tau=(m-1)/m$, the contribution of the Blend
coefficients to local power is described by
\begin{equation*}
  J_B^{\mathrm{blend}}
  =J_{-,B}\left\{
    \tau+\frac{(1-\tau)^2}{\mathcal H_B-\tau}
  \right\}.
\end{equation*}
The next theorem separates this contribution from the effect of direction alignment.

\begin{theorem}\label{thm:blendpower}
Under \eqref{eq:spherical}, suppose $m\ge3$ is fixed,
$0<s_-\le s_b\le s_+<\infty$, $\rho_B\to0$ and
$\sqrt B\rho_B^3\to0$. For any $\mathcal I_B$-measurable unit-direction
rule, the feasible Blend statistic satisfies
\begin{equation}\label{eq:blendpower}
 \mathbb P_{\rho_B}(|N_B^{\mathrm{blend}}|>z_\alpha)
 =\mathbb E_{\rho_B}H_\alpha\left(
   \kappa_R\sqrt B\rho_B\sqrt{J_B^{\mathrm{blend}}}
                  \sqrt{1-L_B}\right)+o(1).
\end{equation}
With the same direction rule, the Radius information coefficient is
$J_B^{\mathrm{rad}}=J_{+,B}^{-1}$, and
\begin{equation}\label{eq:blendcomparison}
 J_B^{\mathrm{rad}}\le J_B^{\mathrm{blend}}\le J_{-,B},
 \qquad
 J_B^{\mathrm{blend}}-J_B^{\mathrm{rad}}
 =J_{-,B}\frac{\tau(\mathcal H_B-1)^2}
                  {\mathcal H_B(\mathcal H_B-\tau)}.
\end{equation}
The upper coefficient $J_{-,B}$ is attained by radius weights divided by the known block variances.
\end{theorem}
\begin{remark}\label{rem:blendlocalpower}
Under the conditions of Theorem~\ref{thm:blendpower}, use the normalized-scatter direction and let
$\rho_B=h/\sqrt B$ for fixed $h\ge0$. Suppose
$J_B^{\mathrm{rad}}\to j_{\mathrm{rad}}$ and
$J_B^{\mathrm{blend}}\to j_{\mathrm{blend}}$.
Then
\[
\begin{aligned}
  \mathbb P_{\rho_B,\mathbf v}
  (|N_B^{\mathrm{rad}}|>z_\alpha)
  &\longrightarrow
  \mathbb E H_\alpha
  (\kappa_Rh\sqrt{j_{\mathrm{rad}}U}),\\
  \mathbb P_{\rho_B,\mathbf v}
  (|N_B^{\mathrm{blend}}|>z_\alpha)
  &\longrightarrow
  \mathbb E H_\alpha
  (\kappa_Rh\sqrt{j_{\mathrm{blend}}U}),
\end{aligned}
\qquad
U\sim\operatorname{Beta}
\left(\frac12,\frac{m-1}{2}\right).
\]
Both methods incur the same random alignment loss, represented by $U$.
Their block coefficients affect power through
$j_{\mathrm{rad}}$ and $j_{\mathrm{blend}}$.
Since $j_{\mathrm{blend}}\ge j_{\mathrm{rad}}$, Blend has at least the same limiting power.
The inequality is strict when $h>0$ and $j_{\mathrm{blend}}>j_{\mathrm{rad}}$.

\end{remark}

For nonspherical covariance, total paired energy need not estimate
noise variance along the selected direction. The population direction
and power formulas above are therefore specific to the stated
Gaussian models. Exact null calibration continues to follow from
the rotation symmetry, and other invariant direction and coefficient
rules remain admissible.

\section{Application to heterogeneous functional data}
\label{sec:functional}

We now construct the paired observations of
Section~\ref{sec:calibration} from functional data collected in
heterogeneous batches. A sample mean retains the location contrast
of interest, while a residual contrast from the same observations
provides a matched noise summary. Under a conditional Gaussian model,
these summaries yield the required null rotation symmetry and permit
the adaptive procedures of Section~\ref{sec:learning}.

Let $\mathcal H$ be a real separable Hilbert space. Suppose observations
are indexed by site $s=1,\ldots,S$, stratum $h=1,\ldots,H_s$,
batch $b=1,\ldots,B$, and replicate
$i=1,\ldots,n_{shb}$. Each batch will form one block for inference.
We assume
\begin{equation}\label{eq:functionalmodel}
  \mathbf Z_{shbi}
  =
  \boldsymbol\mu_{sh}
  +
  \boldsymbol\varepsilon_{shbi},
\end{equation}
where, conditional on a latent random element $\Lambda$, the errors
are independent across all indices and
\begin{equation}\label{eq:functionalerror}
  \boldsymbol\varepsilon_{shbi}\mid\Lambda
  \sim
  \mathcal N_{\mathcal H}
  \left\{
    \mathbf0,\mathbf K_{shb}(\Lambda)
  \right\}.
\end{equation}
The covariance operators $\mathbf K_{shb}(\Lambda)$ are almost surely trace class and may be singular. Within each $(s,h,b)$, all replicates have the same conditional covariance. Covariances may otherwise differ across sites, strata and batches. Dependence on the common $\Lambda$ allows shared random scales or covariance variation and hence marginal dependence across batches. No unconditional moments are required, so $\boldsymbol\mu_{sh}$ denotes
a location parameter even when the marginal mean does not exist.

Fix coefficients $c_s$ and $q_{sh}$ and define the target
functional location
\[
  \boldsymbol\Delta
  =
  \sum_{s=1}^S
  c_s
  \sum_{h=1}^{H_s}
  q_{sh}\boldsymbol\mu_{sh}.
\]
Let
\[
  \boldsymbol\delta
  =
  \boldsymbol\Phi
  (\boldsymbol\Delta-\boldsymbol\Delta_0),
\]
where
$\boldsymbol\Phi:\mathcal H\to\mathbb R^m$
is a fixed bounded linear map specifying the functional features used
for inference. For example, for square-integrable curves, one may take
$\boldsymbol\Phi\mathbf z
=(\langle\mathbf z,\boldsymbol\psi_1\rangle,\ldots,
  \langle\mathbf z,\boldsymbol\psi_m\rangle)^\top$
for prespecified functions
$\boldsymbol\psi_1,\ldots,\boldsymbol\psi_m\in\mathcal H$.
Our target hypothesis is
\[
  H_0^{(\Phi)}:\boldsymbol\delta=\mathbf0.
\]
This tests the selected features of the functional contrast, which is equivalent to $\boldsymbol\Delta=\boldsymbol\Delta_0$ when $\boldsymbol\Phi$ is injective on a prespecified parameter class containing both $\boldsymbol\Delta$ and $\boldsymbol\Delta_0$.

Terms with $q_{sh}=0$ are omitted from all sums below.
For each $(s,h,b)$ with $q_{sh}\ne0$, assume $n_{shb}\ge2$ and choose a fixed vector
\(
  \mathbf u_{shb}
  =
  (u_{shb1},\ldots,u_{shbn_{shb}})^\top
\)
satisfying
\begin{equation}\label{eq:residualcoeff}
  \sum_{i=1}^{n_{shb}}u_{shbi}=0,
  \qquad
  \sum_{i=1}^{n_{shb}}u_{shbi}^2=1.
\end{equation}
Write
\[
  \overline{\mathbf Z}_{shb}
  =
  \frac1{n_{shb}}
  \sum_{i=1}^{n_{shb}}
  \mathbf Z_{shbi},
\]
and define
\begin{equation}\label{eq:functionalsummaries}
  \mathbf D_{sb}
  =
  \sum_h q_{sh}\overline{\mathbf Z}_{shb},
  \qquad
  \mathbf R_{sb}
  =
  \sum_h
  \frac{q_{sh}}{\sqrt{n_{shb}}}
  \sum_{i=1}^{n_{shb}}
  u_{shbi}\mathbf Z_{shbi}.
\end{equation}
The mean summary $\mathbf D_{sb}$ retains the location information.
The zero-sum condition removes location from $\mathbf R_{sb}$ and
makes its conditional cross-covariance with $\mathbf D_{sb}$ zero.
The unit norm of $\mathbf u_{shb}$, together with the factor
$n_{shb}^{-1/2}$, gives both summaries the covariance contribution
$q_{sh}^2\mathbf K_{shb}(\Lambda)/n_{shb}$ from stratum $h$.
Their joint conditional Gaussian law then gives conditional
independence.
For each batch, define
\begin{equation}\label{eq:functionalpairs}
  \mathbf X_b
  =
  \sum_{s=1}^S
  c_s\boldsymbol\Phi\mathbf D_{sb}
  -
  \boldsymbol\Phi\boldsymbol\Delta_0,
  \qquad
  \mathbf Y_b
  =
  \sum_{s=1}^S
  c_s\boldsymbol\Phi\mathbf R_{sb},
\end{equation}
and, as in Section~\ref{sec:calibration},
\(
\mathbf G_b=\mathbf Y_b+i\mathbf X_b,~
  \mathbf M_b=\mathbf G_b\mathbf G_b^*.
\)
Conditional on $\Lambda$,
\begin{equation}\label{eq:functionalmatching}
  \mathbf X_b\mid\Lambda
  \sim
  \mathcal N_m
  \{\boldsymbol\delta,\boldsymbol\Sigma_b(\Lambda)\},\qquad
  \mathbf Y_b\mid\Lambda
  \sim
  \mathcal N_m
  \{\mathbf0,\boldsymbol\Sigma_b(\Lambda)\},\qquad
  \mathbf X_b\perp\mathbf Y_b\mid\Lambda,
\end{equation}
where
\begin{equation}\label{eq:functionalcov}
  \boldsymbol\Sigma_b(\Lambda)
  =
  \boldsymbol\Phi
  \left\{
    \sum_{s,h}
    \frac{c_s^2q_{sh}^2}{n_{shb}}
    \mathbf K_{shb}(\Lambda)
  \right\}
  \boldsymbol\Phi^*.
\end{equation}
Here $\boldsymbol\Phi^*$ is the adjoint of $\boldsymbol\Phi$.
The covariance matrices need not be estimated to form the pairs.

Under $H_0^{(\Phi)}$ and conditional on $\Lambda$, each constructed pair therefore consists of independent, centered Gaussian vectors with the same covariance. Rotation symmetry follows from this matching, rather than being inferred from the observed pairs. Conditional independence across batches gives the product-rotation symmetry required in Section~\ref{subsec:jointadaptation}, which remains valid after integrating over $\Lambda$. Thus the construction supplies the required symmetry under the stated model without imposing it directly on the original, uncentered functional observations.

\begin{proposition}
\label{prop:functionalcalibration}
Under \eqref{eq:functionalmodel}--\eqref{eq:functionalcov} and
$H_0^{(\Phi)}$, take $V$ to be constant or independent of the joint
collection of observations and $\Lambda$. Then
$(\mathbf G_1,\ldots,\mathbf G_B,V)$ satisfies
\eqref{eq:product}, both conditional on $\Lambda$ and after
integrating over it. For any $\mathcal I_B$-measurable projection
directions and block coefficients satisfying Theorem~\ref{thm:joint},
the combined statistic has the exact conditional
$\mathcal N(0,1)$ reference given $\mathcal I_B$ whenever its
denominator is positive. Zero projected pairs are omitted. Setting
the final $p$-value to one when no nonzero weighted contribution
remains gives a conditionally valid test.
\end{proposition}

\begin{remark}
For $m\ge2$, Radius and Blend give two direct implementations. Construct $\widehat{\mathbf a}_B$ from all invariant matrices using \eqref{eq:direction}--\eqref{eq:directionaverage}, resolving signs and tied eigenvalues by fixed measurable rules based on these matrices alone. Then form $\xi_b,\eta_b,r_b$ in \eqref{projected_XY} and combine the scores in \eqref{eq:angular} using the coefficients in \eqref{eq:radstat} or \eqref{eq:blendweights}.
The feature map $\boldsymbol\Phi$ and residual contrasts are fixed in advance, whereas the direction $\widehat{\mathbf a}_B$ and coefficients $\lambda_b$ are learned from the paired invariants.
Proposition~\ref{prop:functionalcalibration} gives exact null calibration for both procedures even with nonspherical or singular batch covariances. The same guarantee applies to other invariant direction and coefficient rules satisfying Theorem~\ref{thm:joint}.

When the constructed pairs satisfy the spherical Gaussian model in \eqref{eq:spherical}, the power results in Section~\ref{subsec:learnedpower} also apply under their stated conditions.
\end{remark}
\section{Simulation and real data analysis}
\label{sec:simulation-real-data}
\subsection{Simulation studies}
\label{subsec:functionalsimulation}

We compare Radius and Blend with their corresponding sample-splitting procedures across different dimensions, numbers of blocks, noise distributions, covariance structures and signal strengths. Each setting is repeated independently 2,000 times, and all tests use the nominal level $\alpha=0.05$.

We first generate pairs that are independent across blocks
\[
  \mathbf X_b
  =\frac{h}{\sqrt B}\mathbf v
   +s_bW_b^{1/2}\boldsymbol\Sigma^{1/2}
      \boldsymbol\varepsilon_b,
  \qquad
  \mathbf Y_b
  =s_bW_b^{1/2}\boldsymbol\Sigma^{1/2}
      \boldsymbol\zeta_b,
  \qquad b=1,\ldots,B,
\]
where $\mathbf v=\mathbf1_m/\sqrt m$ and all
$\boldsymbol\varepsilon_b,\boldsymbol\zeta_b$ are independent
$\mathcal N_m(\mathbf0,\mathbf I_m)$ vectors.
The nonnegative variables $W_b$ are independent across blocks
and independent of the Gaussian vectors. The two observations
within a block share the same $W_b$. Consequently, under $h=0$,
the pairs satisfy the product-rotation symmetry required for
exact calibration.

We consider
\[
  m\in\{4,8,16,32\},
  \qquad B\in\{4,16,32\},
  \qquad h\in\{0,2,4,6,8\}.
\]
The value $h=0$ is used to assess null rejection probabilities,
whereas positive values assess power. The normalization
$h/\sqrt B$ gives
$\sum_{b=1}^B\|\mathbb E\mathbf X_b\|^2=h^2$.
Thus comparisons across $B$ keep the total squared signal
magnitude fixed while varying the number of blocks.
Three noise distributions are considered. Gaussian noise uses
$W_b=1$. Standardized Student-$t_5$ noise uses
$W_b=3/Q_b$, where $Q_b\sim\chi^2_5$.
For contaminated Gaussian noise, $W_b$ equals $1/2$ with probability
$0.9$ and $11/2$ with probability $0.1$.
All three choices satisfy $\mathbb EW_b=1$, so their marginal
noise covariance matrices agree for fixed $s_b$ and
$\boldsymbol\Sigma$.

We use three covariance structures
\[
  \boldsymbol\Sigma\in
  \left\{
    \mathbf I_m,\,
    \operatorname{diag}
      \left(\tfrac12\mathbf1_{m/2},
            \tfrac32\mathbf1_{m/2}\right),\,
    \mathbf T_m
  \right\}.
\]
Here $\mathbf T_m$ has unit diagonal, entries $.3$ immediately
above and below the diagonal, and zeros elsewhere.
Homogeneous blocks use $s_b^2=1$. Heterogeneous blocks alternate between $s_b^2=1/4$ for odd $b$ and $s_b^2=7/4$ for even $b$, also giving average block variance one. The covariance structure controls variation across coordinates, whereas $s_b$ controls overall noise differences across blocks.

 Radius and Blend use all $B$ blocks to learn the common normalized-scatter direction in \eqref{eq:direction}--\eqref{eq:directionaverage} and differ only in their block coefficients. Their split counterparts use the first $B/2$ blocks to estimate
\[
  \widehat{\mathbf a}_{\mathrm{split}}
  =\frac{\sum_{b=1}^{B/2}\mathbf X_b}
         {\left\|\sum_{b=1}^{B/2}\mathbf X_b\right\|},
\]
with a fixed unit vector when the numerator vanishes. They apply the corresponding Radius or Blend rule to the remaining blocks, estimating all coefficients from those testing blocks alone. Both halves have the same block-scale composition. All four methods use the same generated data and signal magnitude $h/\sqrt B$, so the comparison reflects both direction learning and the number of blocks available for testing.

Null rejection frequencies remain close to the nominal level across the prescribed models. Table~\ref{tab:factorial-null} summarizes the 216 null settings in the main grid.
Figure~\ref{fig:factorial-baseline} shows the homogeneous Gaussian baseline with identity covariance. Radius and Blend have similar power and substantially outperform their splits when $B=4$.
At $B=32$, splitting can be more powerful, especially at larger dimensions, consistent with more accurate direction learning by the training mean.
Figure~\ref{fig:factorial-heterogeneous} shows Student-$t$ noise with band covariance and heterogeneous block scales. Blend improves on Radius at intermediate signal strengths, and both full-data methods often outperform their splits. The complete grid suggests that Blend benefits mainly from block-scale variation rather than coordinate covariance alone. Marginal noise variances are matched across noise laws, but their block-precision distributions differ.

\begin{table}[!htbp]
\centering
\caption{Null rejection frequencies across the 216 settings in the
main factorial grid. The range retains
the smallest and largest observed frequencies.}
\label{tab:factorial-null}
\small
\begin{tabular}{lcccc}
\toprule
 & Radius & Blend & Split-Radius & Split-Blend\\
\midrule
Mean & 0.0507 & 0.0505 & 0.0503 & 0.0499\\
Range & 0.0380--0.0635 & 0.0355--0.0620 & 0.0380--0.0625 & 0.0365--0.0650\\
\bottomrule
\end{tabular}
\end{table}

\Needspace{10\baselineskip}
We also examine the functional construction in
Section~\ref{sec:functional}, using $B=16$ blocks and $n=20$
curves per group in each block. For $g=1,2$ and $t\in[0,1]$, we generate
\[
  Z_{gbi}(t)=\mu_g(t)+s_b\sqrt{W_b}
    \sum_{j=1}^{16}e_{gbi,j}\psi_j(t),
  \qquad \psi_j(t)=\sqrt2\sin(\pi jt),
\]
where the coefficient vectors are independent, centred Gaussian vectors,
independent of the block scales, with covariance $\mathbf I_{16}$
for group 1 and $\mathbf T_{16}$ for group 2. Here $\mathbf T_{16}$
is the band matrix defined above, and $W_b$ is shared by all curves
in block $b$.
With $m=8$ and $\mathbf v=\mathbf1_m/\sqrt m$, set
\[
  \mu_1(t)=-\mu_2(t)
  =\frac{h}{\sqrt{2nB}}\sum_{j=1}^{m}v_j\psi_j(t).
\]
Let $\mathbf C_{gbi}$ contain the first $m$ sine coefficients
of $Z_{gbi}$, and write
$\overline{\mathbf C}_{gb}=n^{-1}\sum_{i=1}^n\mathbf C_{gbi}$.
We construct
\[
  \mathbf X_b=\sqrt{\frac n2}
    (\overline{\mathbf C}_{1b}-\overline{\mathbf C}_{2b}),
  \qquad
  \mathbf Y_b=\frac1{\sqrt2}\sum_{i=1}^n
    u_i(\mathbf C_{1bi}-\mathbf C_{2bi}),
\]
where $u_i=1/\sqrt n$ for the first half of the observations
and $u_i=-1/\sqrt n$ otherwise.
The resulting $\mathbf X_b$ and $\mathbf Y_b$ have means
$h\mathbf v/\sqrt B$ and $\mathbf0$,
respectively, and are independent conditional on $W_b$, with common covariance
$s_b^2W_b(\mathbf I_m+\mathbf T_m)/2$.
We cross Gaussian and Student-$t$ scales with the homogeneous
and heterogeneous block scales defined above and
$h\in\{0,2,4,6,8\}$, applying all four procedures to each setting.

Figure~\ref{fig:factorial-functional} shows that Radius and Blend
outperform their corresponding splits at intermediate signal
strengths. Their power is similar under homogeneous Gaussian noise,
whereas Blend benefits from block-scale variation, including the
shared Student-$t$ scales. The differences narrow as power approaches
one. These results illustrate the benefits of data reuse and adaptive
weighting when the pairs are constructed from individual curves.
Supplementary Section~\ref{supp:factorial-study} provides the remaining
power comparisons and an additional signal-location experiment.

\begin{figure}[p]
\centering
\includegraphics[width=\linewidth,height=.70\textheight,keepaspectratio]{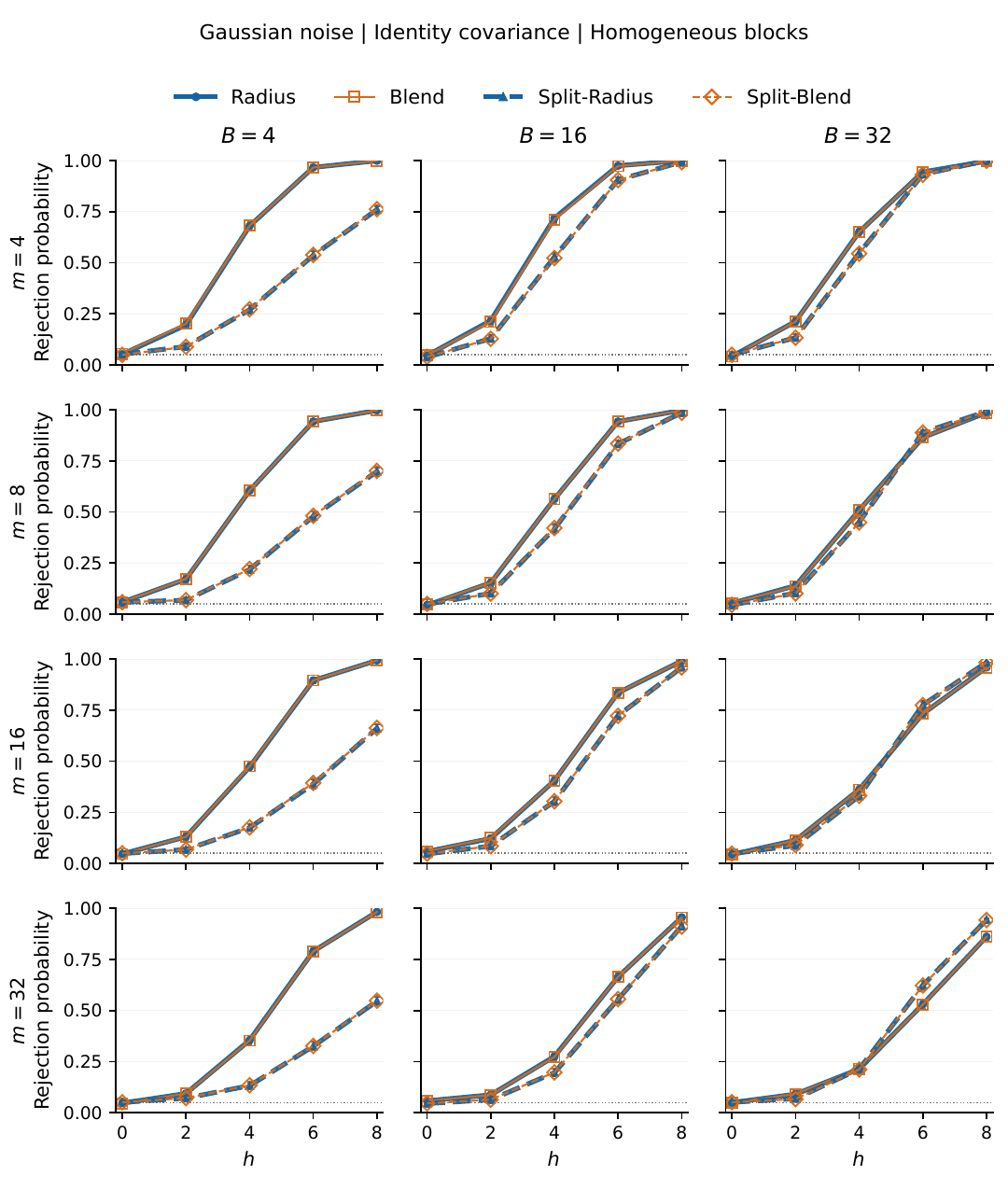}
\caption{Gaussian noise with identity covariance and homogeneous block
scales. Rows correspond to $m=4,8,16,32$ and columns to $B=4,16,32$.
Each point uses 2,000 independent datasets. Colours distinguish Radius
and Blend, and solid and dashed lines distinguish full-data and split
procedures. Bars are pointwise 95\% Wilson intervals. The dotted
horizontal line is $.05$. At $h=0$ the curves show null rejection
frequencies, and at $h>0$ they show power. Radius and Blend nearly coincide here. Wider blue lines
and filled markers distinguish Radius from the thinner orange lines
and open markers of Blend.}
\label{fig:factorial-baseline}
\end{figure}
\begin{figure}[p]
\centering
\includegraphics[width=\linewidth,height=.70\textheight,keepaspectratio]{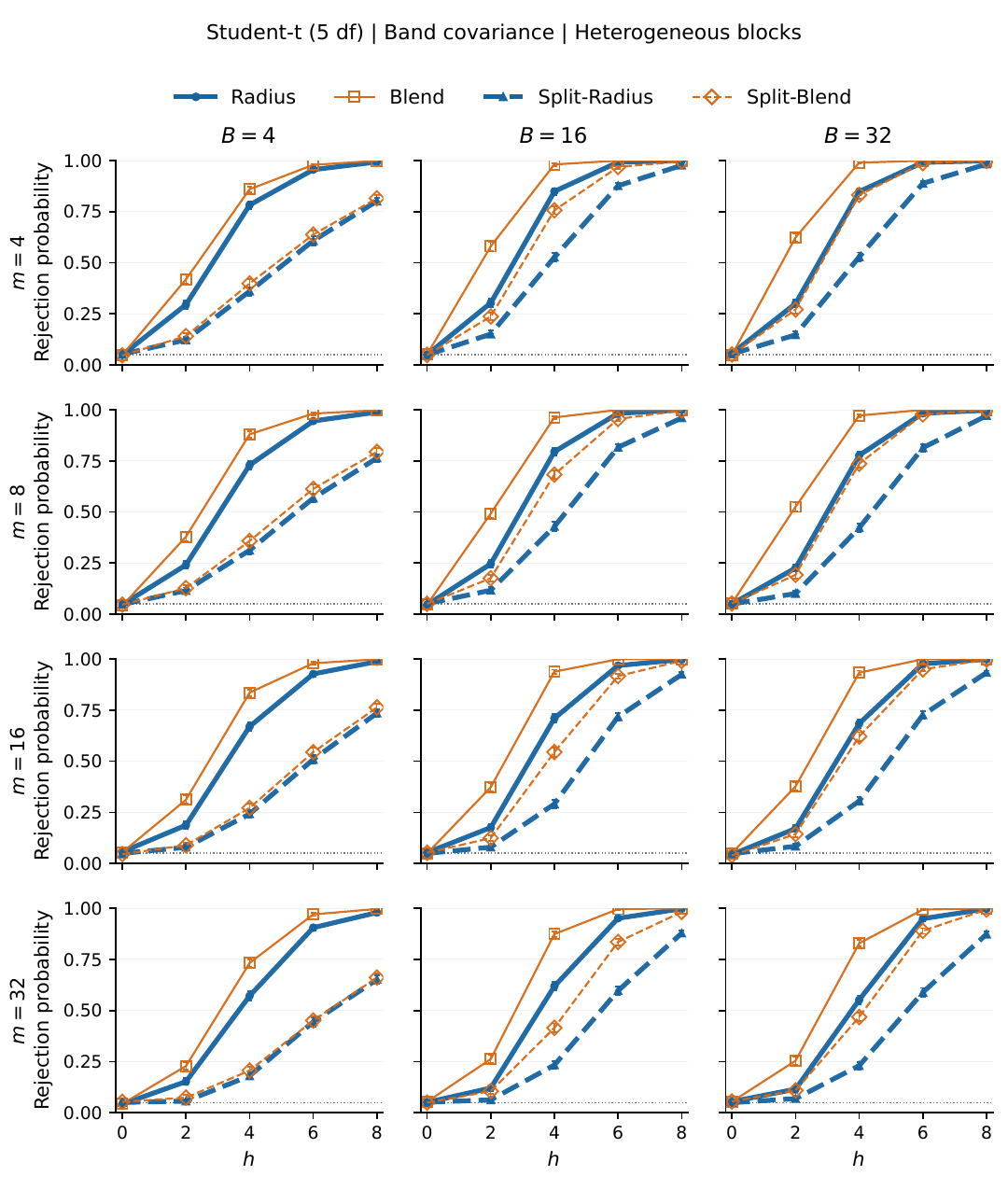}
\caption{Student-$t$ noise with band covariance and heterogeneous block scales.
The panel arrangement, four procedures, signal grid, replication count
and interval convention are the same as in
Figure~\ref{fig:factorial-baseline}. Each panel shows absolute rejection
probabilities for all four methods.}
\label{fig:factorial-heterogeneous}
\end{figure}

\begin{figure}[!htbp]
\centering
\includegraphics[width=.93\linewidth,height=.70\textheight,keepaspectratio]{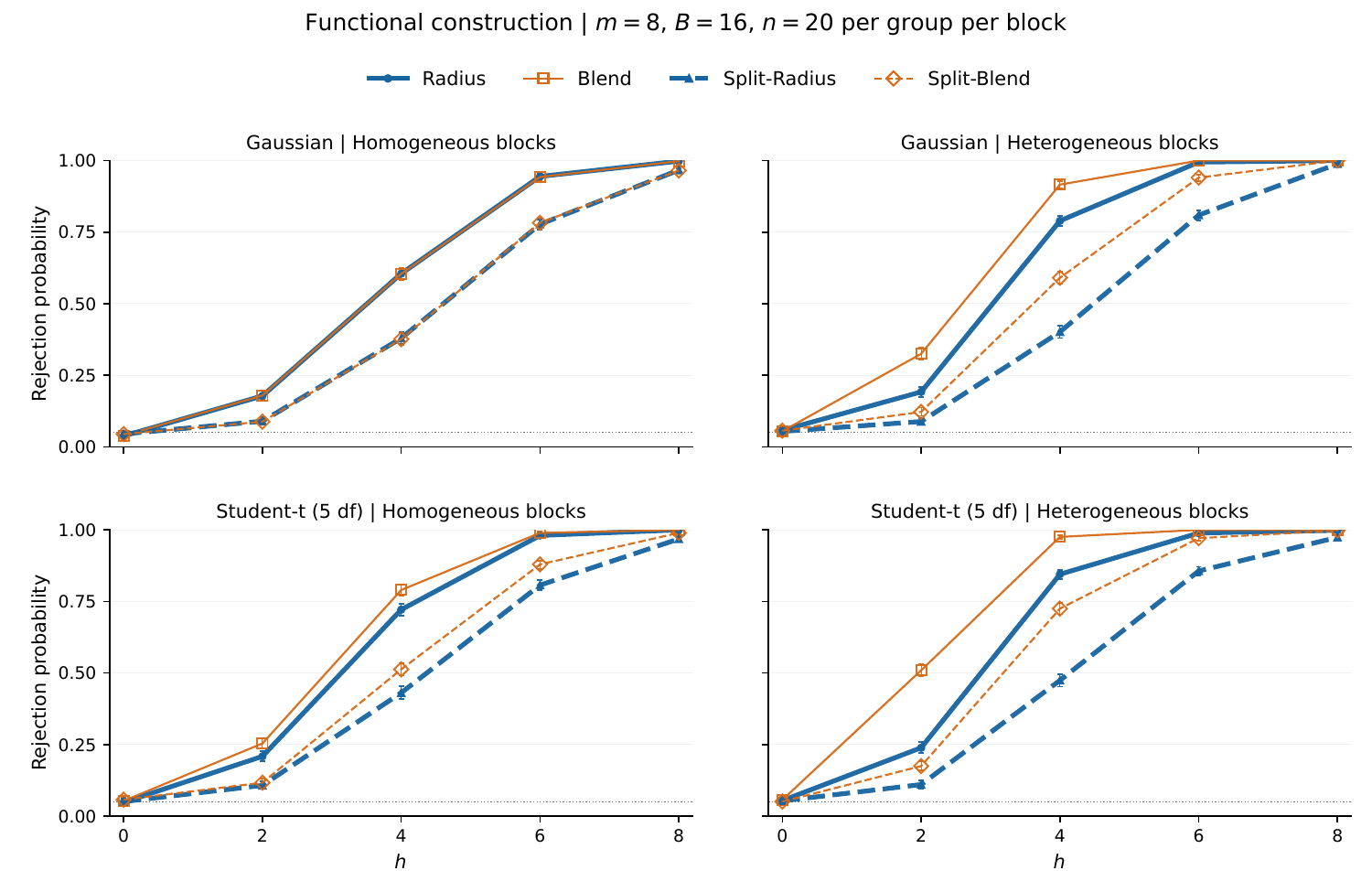}
\caption{Functional-data experiment with $m=8$, $B=16$ and 20 curves per group
per block. Rows distinguish Gaussian and Student-$t$ block scales, and
columns distinguish homogeneous and heterogeneous deterministic scales.
The pairs are formed from the generated individual curves using the
mean--residual construction. Each point uses 2,000 datasets, with
pointwise 95\% Wilson intervals.}
\label{fig:factorial-functional}
\end{figure}

\subsection{Real data analysis}
\label{subsec:ecg}

We analyse PTB-XL \citep{Wagner2020} and Chapman--Shaoxing
\citep{Zheng2020} ECG records from PhysioNet Challenge 2021
\citep{Reyna2021}, comparing patients labelled with atrial
fibrillation (AF) but not normal sinus rhythm (NSR) with NSR-only
controls. Each patient contributes one 10-second recording,
represented by four RR-interval summaries and four spectral power
ratios from leads II and V1. The feature transformations and
PTB-XL calibration-based scaling are fixed across sources.
Supplementary Section~\ref{supp:ecg} gives the preprocessing details.

We use $n=160$ patients per group and $B=8$ blocks,
with 20 patients per group in each block.
Mean differences and zero-sum residual contrasts form the pairs
in Section~\ref{sec:functional}. Radius, Blend and their splits
follow Section~\ref{subsec:functionalsimulation}, with all eight
blocks used by the full-data procedures and four training blocks
and four testing blocks used by the splits.

To examine the effect of signal dilution, we compare NSR patients with
mixtures having AF fractions $q\in\{.10,.20,.40,1\}$.
All methods share the same assignments, with no patient appearing
twice within an assignment. Figure~\ref{fig:ecg-primary} shows
null rejection frequencies close to $.05$ and higher dilution
rejection frequencies for the full-data procedures before saturation.
At $q=.10$, Radius rejects in about $44\%$ and $32\%$ of assignments
in the two sources, compared with $19\%$ and $18\%$ for Split-Radius.
Blend performs similarly to Radius. For the nonlinear ECG features, the conditional Gaussian model
underlying the matching construction is not guaranteed, so
calibration is assessed empirically through the null checks.
\begin{figure}[!htbp]
\centering
\includegraphics[width=.93\linewidth,height=.70\textheight,keepaspectratio]{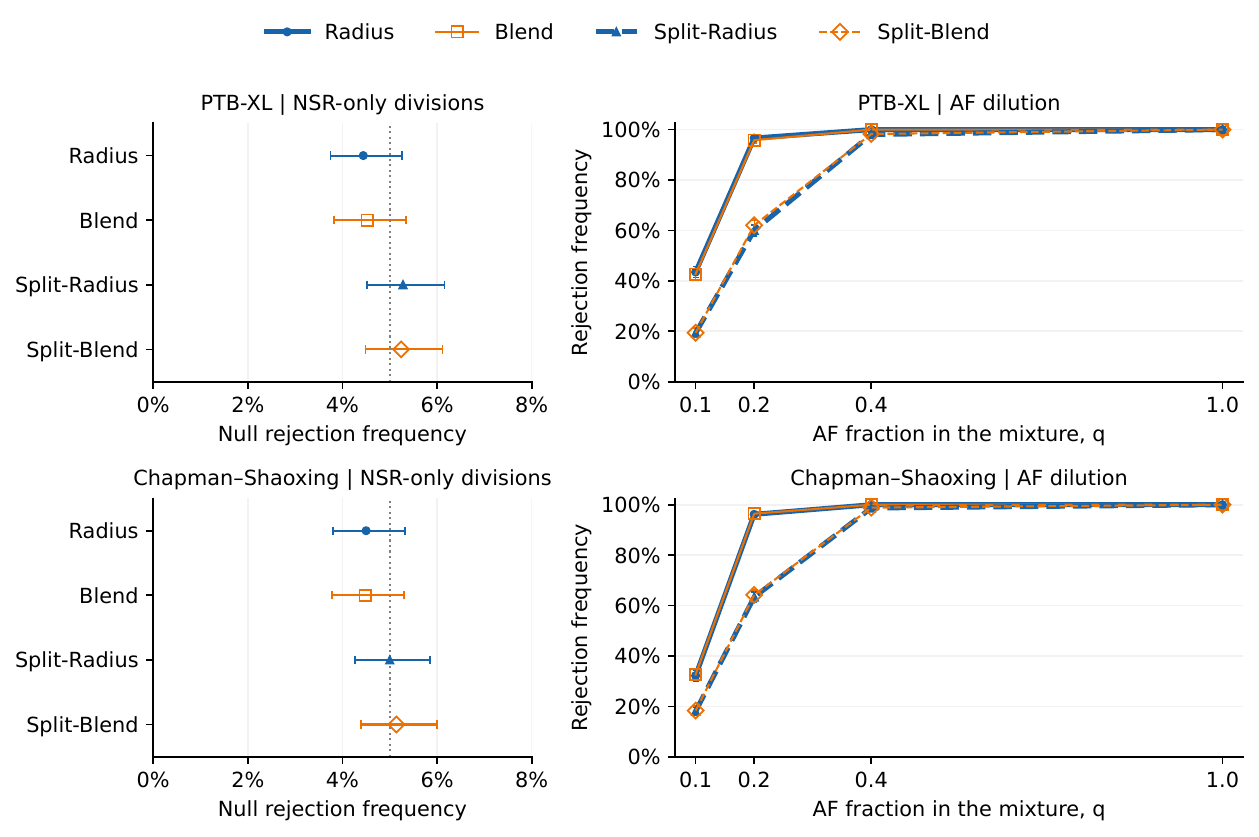}
\caption{ECG comparisons at $n=160$, $m=8$ and $B=8$, using $p<.05$.
Rows correspond to the two sources. Left, null rejection frequencies
and 99\% Wilson intervals from 5,000 NSR-only assignments.
Right, dilution rejection frequencies and 95\% Wilson intervals
from 2,000 assignments per AF fraction. Intervals condition on
the fixed patient pools.}
\label{fig:ecg-primary}
\end{figure}

\section*{Declaration of the use of generative AI and AI-assisted technologies}
During the preparation of this manuscript, the authors used ChatGPT (OpenAI)
for LaTeX formatting and assistance with submission materials. The authors
take full responsibility for the content of the manuscript.

\section*{Supplementary material}
The Supplementary Material includes proofs of the calibration, projection,
direction-recovery and power results, additional simulation results, and
preprocessing details for the electrocardiogram analysis.

\clearpage
\setcounter{section}{0}
\renewcommand{\thesection}{S\arabic{section}}
\renewcommand{\theHsection}{S.\arabic{section}}
\setcounter{equation}{0}
\renewcommand{\theequation}{S\arabic{equation}}
\renewcommand{\theHequation}{S.\arabic{equation}}
\setcounter{figure}{0}
\renewcommand{\thefigure}{S\arabic{figure}}
\renewcommand{\theHfigure}{S.\arabic{figure}}
\setcounter{table}{0}
\renewcommand{\thetable}{S\arabic{table}}
\renewcommand{\theHtable}{S.\arabic{table}}
\setcounter{theorem}{0}
\renewcommand{\thetheorem}{S\arabic{theorem}}
\renewcommand{\theHtheorem}{S.\arabic{theorem}}
\setcounter{lemma}{0}
\renewcommand{\thelemma}{S\arabic{lemma}}
\renewcommand{\theHlemma}{S.\arabic{lemma}}
\setcounter{corollary}{0}
\renewcommand{\thecorollary}{S\arabic{corollary}}
\renewcommand{\theHcorollary}{S.\arabic{corollary}}
\setcounter{proposition}{0}
\renewcommand{\theproposition}{S\arabic{proposition}}
\renewcommand{\theHproposition}{S.\arabic{proposition}}
\setcounter{remark}{0}
\renewcommand{\theremark}{S\arabic{remark}}
\renewcommand{\theHremark}{S.\arabic{remark}}
\renewcommand{\thesubsection}{\thesection.\arabic{subsection}}
\pdfbookmark[0]{Supplementary material}{supplementary-material}
\begin{center}
  {\LARGE\bfseries Supplementary material\par}
  \vspace{8pt}
  {\Large Adaptive Cauchy inference from invariant information\par}
  \vspace{8pt}
  \begin{minipage}[t]{.47\textwidth}
    \centering
    YIMING LIU\\[2pt]
    {\small School of Economics, Jinan University,}\\
    {\small Guangzhou 510632, China}\\
    {\small \href{mailto:liuyiming@jnu.edu.cn}{\texttt{liuyiming@jnu.edu.cn}}}
  \end{minipage}\hfill
  \begin{minipage}[t]{.47\textwidth}
    \centering
    SHAOCHEN WANG\\[2pt]
    {\small School of Mathematics,}\\
    {\small South China University of Technology,}\\
    {\small Guangzhou 510640, China}\\
    {\small \href{mailto:mascwang@scut.edu.cn}{\texttt{mascwang@scut.edu.cn}}}
  \end{minipage}
\end{center}
\medskip
The supplement follows the order of the main paper.
Section~\ref{proof:cal-section} proves the calibration results
and their projection extensions.
Section~\ref{proof:learning-section} gives the proofs for
representative procedures, direction recovery and power.
Section~\ref{proof:cal-functional} verifies the functional-data
construction.
Sections~\ref{supp:factorial-study} and~\ref{supp:ecg} provide
additional simulation results and ECG analysis details.
References prefixed by S refer to this supplement.

\section{Exact adaptation under rotation symmetry}
\label{proof:cal-section}

\subsection{Conditional Cauchy calibration}
\label{proof:cal-single}

We use the notation of Section~\ref{sec:calibration}. Here $q(z)=\operatorname{Im}(z)/\operatorname{Re}(z)$
when $\operatorname{Re}(z)\ne0$, and $q(z)=0$ otherwise, as in the
main paper. For $\mathbf g\ne\mathbf0$, the admissible simplex is
\[
 \Delta(\mathbf g)=
 \left\{\mathbf w\in[0,1]^m\ \middle|\
  \sum_{j=1}^m w_j=1,\quad w_j=0\text{ whenever }g_j=0\right\}.
\]
Conditional identities below hold for bounded measurable test functions.

The calibration argument separates the information in
$\mathbf M=\mathbf G\mathbf G^*$ from the common phase.
For a nonzero vector $\mathbf g$, its orbit is the set
$\{e^{i\phi}\mathbf g:\phi\in[0,2\pi)\}$ of all its common
rotations.
We choose a representative $\mathbf h(\mathbf M)$ whose
first nonzero coordinate is positive and real, so that
$\mathbf G=e^{i\Theta}\mathbf h(\mathbf M)$.
Under the rotation assumption, we then show that $\Theta$
is conditionally uniform given the invariant information.
This reduces the conditional
ratio statistic to a weighted sum of shifted tangents, to which
Lemma~\ref{proof:cal-tangent} applies.

To fix the representative, rotate the first nonzero coordinate
to the positive real axis.
For a nonzero $\mathbf g\in\mathbb C^m$, let
\[
 j(\mathbf g)=\min\{j\mid |g_j|>0\},\qquad
 \theta(\mathbf g)=\arg g_{j(\mathbf g)}\in[0,2\pi).
\]
Writing $\mathbf m=\mathbf g\mathbf g^*$, define
\begin{equation}\label{proof:cal-section-map}
 j(\mathbf m)=\min\{j\mid m_{jj}>0\},\qquad
 \mathbf h(\mathbf m)=
 \frac{\mathbf m\mathbf e_{j(\mathbf m)}}
      {\sqrt{m_{j(\mathbf m)j(\mathbf m)}}}.
\end{equation}
Here $\mathbf e_j$ is the $j$th coordinate unit vector.
The entries of $\mathbf m$ are continuous functions of $\mathbf g$.
The index is measurable since
$\{j(\mathbf m)=r\}=\{m_{rr}>0\}\cap\bigcap_{j<r}\{m_{jj}=0\}$.
Consequently, both maps in \eqref{proof:cal-section-map} are measurable.
Direct multiplication gives
\begin{equation}\label{proof:cal-orbit-factorization}
 \mathbf h(\mathbf g\mathbf g^*)=
 \frac{\overline{g_{j(\mathbf g)}}}{|g_{j(\mathbf g)}|}\mathbf g,
 \qquad
 \mathbf g=e^{i\theta(\mathbf g)}\mathbf h(\mathbf g\mathbf g^*),
 \qquad
 h_{j(\mathbf g)}=|g_{j(\mathbf g)}|>0.
\end{equation}
Moreover, $\mathbf h(\mathbf m)\mathbf h(\mathbf m)^*=\mathbf m$.
Thus two nonzero vectors have the same matrix in
\eqref{eq:gram} exactly when they differ by a common unit phase.
The displayed map supplies a measurable representative of each orbit.
At zero set $\mathbf h(\mathbf0)=\mathbf0$ and $\theta(\mathbf0)=0$.

We next establish conditional uniformity of the remaining phase.
Let $U$ collect any additional information on which we condition.
Assume $\mathbf G\ne\mathbf0$ almost surely and
$(e^{i\phi}\mathbf G,U)\overset d=(\mathbf G,U)$ for every fixed
$\phi$, so that this additional information preserves the joint
rotation symmetry. Set $\Theta=\theta(\mathbf G)$.
For bounded measurable functions $a$ and $d$, joint invariance gives
\[
 \mathbb E\{d(\mathbf M,U)a(\Theta)\}
 =\mathbb E\{d(\mathbf M,U)a((\Theta+\phi)\bmod 2\pi)\}.
\]
Integrate over $\phi$ and apply Fubini's theorem, which is applicable
because the integrand is bounded. Translation invariance of circle
measure yields
\begin{equation}\label{proof:cal-haar-single}
 \mathbb E\{d(\mathbf M,U)a(\Theta)\}
 =\mathbb E\{d(\mathbf M,U)\}
   \frac1{2\pi}\int_0^{2\pi}a(t)\,dt.
\end{equation}
This is precisely conditional uniformity of $\Theta$ given
$\sigma(\mathbf M,U)$. In particular, for every bounded measurable
$f$ on $\mathbb C^m$,
\begin{equation}\label{proof:cal-orbit-kernel}
 \mathbb E\{f(\mathbf G)\mid\mathbf M,U\}
 =\frac1{2\pi}\int_0^{2\pi}
   f\{e^{it}\mathbf h(\mathbf M)\}\,dt
 \quad\text{almost surely}.
\end{equation}
The same identity with an integrand depending measurably on
$(\mathbf M,U)$ follows first for products of bounded functions and
then by a monotone-class argument. This justifies substituting
invariantly selected directions or weights into the integral.
For nonnegative $f$, apply \eqref{proof:cal-orbit-kernel} to
$f\wedge n$ and let $n\to\infty$ by conditional monotone convergence.

The following identity converts the uniform phase into the
standard Cauchy reference distribution.

\begin{lemma}\label{proof:cal-tangent}
Let $\Theta$ be uniform on $[0,2\pi)$, let
$u_1,\ldots,u_L\in\mathbb R$, and let
$w_1,\ldots,w_L\geq0$ satisfy $\sum_{\ell=1}^Lw_\ell=1$.
Then
\[
  T=\sum_{\ell=1}^Lw_\ell\tan(\Theta+u_\ell)
\]
is well defined almost surely and has the standard Cauchy distribution.
\end{lemma}

\begin{proof}
Set $c_\ell=e^{2iu_\ell}$ and, for $|z|<1$, define
\[
 F(z)=i\sum_{\ell=1}^Lw_\ell\frac{1-c_\ell z}{1+c_\ell z}.
\]
Since $|c_\ell|=1$,
\[
  |1+c_\ell z|\geq 1-|z|>0,
  \qquad |z|<1,\quad \ell=1,\ldots,L.
\]
Hence $F$ is analytic on $\{z\in\mathbb C:|z|<1\}$, $F(0)=i$, and
\begin{equation*}
 \operatorname{Im}F(z)=
 \sum_{\ell=1}^Lw_\ell
 \operatorname{Re}\frac{1-c_\ell z}{1+c_\ell z}
 =\sum_{\ell=1}^Lw_\ell
   \frac{1-|z|^2}{|1+c_\ell z|^2}>0.
\end{equation*}
For $s\ge0$, the analytic function $A_s(z)=\exp\{isF(z)\}$
satisfies $|A_s(z)|\le1$. The analytic mean-value formula on each
circle of radius $0<r<1$ therefore gives
\[
 \frac1{2\pi}\int_0^{2\pi}A_s(re^{it})\,dt=A_s(0)=e^{-s}.
\]
Let
\[
  \mathcal E
  =\bigcup_{\ell=1}^L
    \{t\in[0,2\pi):1+c_\ell e^{it}=0\},
  \qquad \#\mathcal E\leq L.
\]
For every $t\in[0,2\pi)\setminus\mathcal E$,
\[
  \lim_{r\uparrow1}F(re^{it})
  =i\sum_{\ell=1}^L w_\ell
    \frac{1-e^{i(t+2u_\ell)}}{1+e^{i(t+2u_\ell)}}
  =\sum_{\ell=1}^L w_\ell\tan(t/2+u_\ell).
\]
Dominated convergence using $|A_s|\le1$ consequently gives
\begin{equation}\label{proof:cal-tangent-cf}
 \frac1{2\pi}\int_0^{2\pi}
 \exp\left\{is\sum_{\ell=1}^Lw_\ell\tan(t/2+u_\ell)\right\}\,dt
 =e^{-s},\qquad s\ge0.
\end{equation}
Write
$\varphi_T(s)=\mathbb E(e^{isT})$.
By the $\pi$-periodicity of the tangent function,
\[
\begin{aligned}
  \varphi_T(s)
  &=\frac1{2\pi}\int_0^{2\pi}
    \exp\left\{is\sum_{\ell=1}^L
      w_\ell\tan(\theta+u_\ell)\right\}\,d\theta\\
  &=\frac1\pi\int_0^\pi
    \exp\left\{is\sum_{\ell=1}^L
      w_\ell\tan(\theta+u_\ell)\right\}\,d\theta\\
  &=\frac1{2\pi}\int_0^{2\pi}
    \exp\left\{is\sum_{\ell=1}^L
      w_\ell\tan(t/2+u_\ell)\right\}\,dt.
\end{aligned}
\]
Equation~\eqref{proof:cal-tangent-cf} therefore gives
$\varphi_T(s)=e^{-s}$ for $s\geq0$.
Since $T$ is real-valued, for $s<0$,
\[
  \varphi_T(s)
  =\overline{\varphi_T(-s)}
  =\overline{e^s}
  =e^s.
\]
Consequently,
\[
  \varphi_T(s)=e^{-|s|},
  \qquad s\in\mathbb R.
\]
Since $\int_{\mathbb R}|\varphi_T(s)|\,ds=2<\infty$,
the Fourier inversion theorem gives a density for $T$,
\[
\begin{aligned}
  f_T(x)
  &=\frac1{2\pi}\int_{\mathbb R}
      e^{-isx}\varphi_T(s)\,ds\\
  &=\frac1{2\pi}\left\{
      \int_0^\infty e^{-(1+ix)s}\,ds
      +\int_0^\infty e^{-(1-ix)s}\,ds
    \right\}\\
  &=\frac1{2\pi}
    \left(\frac1{1+ix}+\frac1{1-ix}\right)
   =\frac1{\pi(1+x^2)},
  \qquad x\in\mathbb R.
\end{aligned}
\]
Thus $T$ has the standard Cauchy distribution.

\end{proof}

\begin{proof}[Proof of Theorem~\ref{thm:cauchy}]

Apply \eqref{proof:cal-haar-single} with $U=(V,\mathbf W)$.
The joint invariance in \eqref{eq:jointrotation} gives
\[
  \mathcal L(\Theta\mid\mathbf M,V,\mathbf W)
  =\operatorname{Unif}[0,2\pi)
  \quad\text{almost surely}.
\]
Let $\mathcal A=\{j\mid M_{jj}>0\}$. For $j\in\mathcal A$, set
\[
  r_j=\sqrt{M_{jj}},
  \qquad
  u_j=\arg h_j(\mathbf M).
\]
These quantities are $\mathbf M$-measurable, and
\eqref{proof:cal-orbit-factorization} yields
\[
  G_j=e^{i\Theta}h_j(\mathbf M)
     =r_je^{i(\Theta+u_j)},
  \qquad
  Y_j=r_j\cos(\Theta+u_j),
  \qquad
  X_j=r_j\sin(\Theta+u_j).
\]
Conditional uniformity therefore implies
\[
\begin{aligned}
  &\mathbb P\left(
    Y_j=0\text{ for some }j\in\mathcal A
    \,\middle|\,\mathbf M,V,\mathbf W
  \right)\\
  &\qquad=
  \frac1{2\pi}\int_0^{2\pi}
    \mathbf1\left\{
      \cos(t+u_j)=0\text{ for some }j\in\mathcal A
    \right\}\,dt
  =0
  \quad\text{almost surely}.
\end{aligned}
\]
Hence, simultaneously for all $j\in\mathcal A$,
\[
  q(G_j)=\frac{X_j}{Y_j}
        =\tan(\Theta+u_j)
  \quad\text{almost surely}.
\]
By the support condition,
\[
  W_j=0\quad(j\notin\mathcal A),
  \qquad W_j\geq0,
  \qquad \sum_{j\in\mathcal A}W_j=1.
\]
Consequently,
\[
  T=\sum_{j\in\mathcal A}W_j\tan(\Theta+u_j)
  \quad\text{almost surely}.
\]
Conditional on $(\mathbf M,V,\mathbf W)$, the weights and
phase shifts are fixed. Lemma~\ref{proof:cal-tangent} gives
\[
  \mathcal L(T\mid\mathbf M,V,\mathbf W)
  =\operatorname{Cauchy}(0,1)
  \quad\text{almost surely},
\]
which proves \eqref{eq:conditionalcauchy}.

For every Borel set $E\subseteq\mathbb R$ and every
$D\in\sigma(\mathbf M,V,\mathbf W)$,
\[
\begin{aligned}
  \mathbb P(T\in E,D)
  &=\mathbb E\left\{
      \mathbf1_D\mathbb P(T\in E\mid\mathbf M,V,\mathbf W)
    \right\}\\
  &=\mathbb P(D)\int_E\frac{dx}{\pi(1+x^2)}\\
  &=\mathbb P(D)\mathbb P(T\in E).
\end{aligned}
\]
Thus $T$ is independent of $(\mathbf M,V,\mathbf W)$.

Finally, suppose $\mathbf W=\mathbf w(\mathbf M,V)$ for a
measurable function $\mathbf w$. Since
\[
  (e^{i\phi}\mathbf G)(e^{i\phi}\mathbf G)^*
  =\mathbf M,
\]
the assumed invariance of $(\mathbf G,V)$ implies
\[
\begin{aligned}
  (e^{i\phi}\mathbf G,V,\mathbf W)
  &=
  \left(e^{i\phi}\mathbf G,V,
    \mathbf w\{(e^{i\phi}\mathbf G)
              (e^{i\phi}\mathbf G)^*,V\}\right)\\
  &\overset d=
  \left(\mathbf G,V,\mathbf w(\mathbf G\mathbf G^*,V)\right)
  =(\mathbf G,V,\mathbf W).
\end{aligned}
\]
Moreover,
\[
  \sigma(\mathbf M,V,\mathbf W)=\sigma(\mathbf M,V).
\]
The preceding conditional distribution therefore becomes
\[
  \mathcal L(T\mid\mathbf M,V)
  =\operatorname{Cauchy}(0,1)
  \quad\text{almost surely},
\]
proving the final assertion.
\end{proof}

\begin{proof}[Proof of Corollary~\ref{cor:projections}]
Put $\widetilde G=\widehat{\mathbf a}^{\mathsf T}\mathbf G$ and
$d=\widehat{\mathbf a}^{\mathsf T}\operatorname{Re}(\mathbf M)
\widehat{\mathbf a}$. Both $d$ and the direction are measurable in
$(\mathbf M,V)$, and $|\widetilde G|^2=d$. On the invariant event
$\{d>0\}$, \eqref{proof:cal-orbit-factorization} gives
\[
 \widetilde G=e^{i\Theta}
   \widehat{\mathbf a}^{\mathsf T}\mathbf h(\mathbf M),
 \qquad
 \left|\widehat{\mathbf a}^{\mathsf T}\mathbf h(\mathbf M)\right|^2=d>0.
\]
On $\{d>0\}$, let
$\psi=\arg\{\widehat{\mathbf a}^{\top}\mathbf h(\mathbf M)\}$.
Then
\[
  \widetilde Y+i\widetilde X
  =\sqrt d\,e^{i(\Theta+\psi)},
  \qquad
  q(\widetilde Y+i\widetilde X)
  =\tan(\Theta+\psi)
  \quad\text{almost surely}.
\]
Conditional on $(\mathbf M,V)$, the phase $\psi$ is fixed
and $\Theta$ is uniform on $[0,2\pi)$.
Lemma~\ref{proof:cal-tangent} with $L=1$ and $w_1=1$
therefore proves \eqref{eq:selectedprojection}
on $\{d>0\}$.
\end{proof}

\begin{proposition}[Finite invariantly selected projections]
\label{proof:cal-finite-projections}
Under the assumptions of Corollary~\ref{cor:projections}, let
$\widehat{\mathbf a}_1,\ldots,\widehat{\mathbf a}_L\in\mathbb R^m$
be $(\mathbf M,V)$-measurable, where $L$ is finite. Set
\[
 \widetilde G_\ell=\widehat{\mathbf a}_\ell^{\mathsf T}\mathbf G,
 \qquad
 d_\ell=\widehat{\mathbf a}_\ell^{\mathsf T}
       \operatorname{Re}(\mathbf M)\widehat{\mathbf a}_\ell,
 \qquad
 \mathcal A=\{\ell\mid d_\ell>0\}.
\]
On $\{\mathcal A\ne\varnothing\}$, choose nonnegative weights
$w_\ell$ measurable with respect to $\sigma(\mathbf M,V)$
that vanish off $\mathcal A$ and sum to one.
Then
\[
 \mathcal L\left\{\sum_{\ell=1}^Lw_\ell q(\widetilde G_\ell)
          \,\middle|\,\mathbf M,V\right\}
 =\operatorname{Cauchy}(0,1).
\]
The directions need not be orthogonal or distinct. The result also
allows an almost surely finite, invariantly selected number $L$.
\end{proposition}

\begin{proof}
The energies $d_\ell=|\widetilde G_\ell|^2$ make $\mathcal A$
invariant. Conditional on $(\mathbf M,V)$, every nonzero projected
coordinate equals $e^{i\Theta}$ times the fixed nonzero number
$\widehat{\mathbf a}_\ell^{\mathsf T}\mathbf h(\mathbf M)$.
Lemma~\ref{proof:cal-tangent} applies to their weighted tangents.
For random finite $L$, apply this argument on each measurable event
$\{L=l\}$ and combine the countably many conclusions. In particular,
conditioning remains on the original invariant, which may contain
more information than the projected Gram matrix.
\end{proof}

\subsection{Joint adaptation and evidence combination}
\label{proof:cal-joint}

\begin{proof}[Proof of Theorem~\ref{thm:joint}]
For each nonzero block use the canonical representation
$\mathbf G_b=e^{i\Theta_b}\mathbf h(\mathbf M_b)$.
The active block set is invariant since
\[
 \mathcal B_*=\{b\mid\operatorname{tr}(\mathbf M_b)>0\}.
\]
For a fixed subset $J\subseteq\{1,\ldots,B\}$, write
$E_J=\{\mathcal B_*=J\}$. For bounded $\mathcal I_B$-measurable
$D$ and bounded measurable functions $a_b$, average
\eqref{eq:product} over the $B$ rotation angles. Fubini and
circle translation invariance give
\begin{equation}\label{proof:cal-product-haar}
 \mathbb E\left\{\mathbf1_{E_J}D\prod_{b\in J}a_b(\Theta_b)\right\}
 =\mathbb E(\mathbf1_{E_J}D)
   \prod_{b\in J}\left\{\frac1{2\pi}\int_0^{2\pi}a_b(t)\,dt\right\}.
\end{equation}
Taking $a_b=\mathbf1_{[0,t_b]}$ in
\eqref{proof:cal-product-haar} gives, on $E_J$,
\[
  \mathbb P\left(
    \Theta_b\leq t_b\text{ for all }b\in J
    \,\middle|\,\mathcal I_B
  \right)
  =\prod_{b\in J}\frac{t_b}{2\pi},
  \qquad 0\leq t_b\leq2\pi,
\]
almost surely. Thus, conditional on $\mathcal I_B$, the phases
$(\Theta_b)_{b\in J}$ are independent and uniform on $[0,2\pi)$.
For $b\notin J$, we have $\mathbf G_b=\mathbf0$ on $E_J$,
and these blocks are excluded from the phase statement.

On $E_J$, conditioning on $\mathcal I_B$ fixes all weights
and phase shifts, so each $T_b$ depends only on $\Theta_b$.
Lemma~\ref{proof:cal-tangent} and the conditional independence
of $(\Theta_b)_{b\in J}$ therefore give
\[
  \mathcal L\bigl((T_b)_{b\in J}\mid\mathcal I_B\bigr)
  =\bigotimes_{b\in J}\operatorname{Cauchy}(0,1)
  \quad\text{almost surely on }E_J.
\]
For $b\in J$, the two-sided $p$-value is
\[
  p_b=1-\frac2\pi\arctan|T_b|.
\]
For every $0<u<1$, on $E_J$,
\[
\begin{aligned}
  \mathbb P(p_b\leq u\mid\mathcal I_B)
  &=\mathbb P\left(
      |T_b|\geq\tan\left\{\frac{\pi(1-u)}2\right\}
      \,\middle|\,\mathcal I_B
    \right)\\
  &=1-\frac2\pi
    \arctan\left[\tan\left\{\frac{\pi(1-u)}2\right\}\right]
   =u.
\end{aligned}
\]
Since each $p_b$ is a function of $T_b$ alone,
\[
  \mathcal L\bigl((p_b)_{b\in J}\mid\mathcal I_B\bigr)
  =\bigotimes_{b\in J}\operatorname{Unif}(0,1)
  \quad\text{almost surely on }E_J.
\]
Let $S_b=\operatorname{sgn}(\widehat{\mathbf a}_b^{\mathsf T}
\mathbf X_b)$. The assumed positive projected energy means
$\widehat{\mathbf a}_b^{\mathsf T}\mathbf h(\mathbf M_b)\ne0$.
Thus $S_b\in\{-1,1\}$ except at two conditional probability-zero
angles. A half-turn sends $\mathbf G_b$ to $-\mathbf G_b$ and
changes $S_b$ to $-S_b$, while $q(-z)=q(z)$ leaves $T_b$ and $p_b$
unchanged. Consequently, for every Borel set $A\subseteq(0,1)$,
\begin{equation}\label{proof:cal-half-turn}
 \begin{split}
 \mathbb P(p_b\in A,S_b=1\mid\mathcal I_B)
 &=\mathbb P(p_b\in A,S_b=-1\mid\mathcal I_B)\\
 &=\tfrac12\mathbb P(p_b\in A\mid\mathcal I_B)
 =\tfrac12\operatorname{Leb}(A).
 \end{split}
\end{equation}
Hence the sign is conditionally symmetric and independent of $p_b$.
For $x\ge0$, the magnitude $A_b=\Phi^{-1}(1-p_b/2)$ obeys
\[
 \mathbb P(A_b\le x\mid\mathcal I_B)=2\Phi(x)-1.
\]
Combining this half-normal magnitude with the independent symmetric
sign gives a standard normal variable. Product phase independence
also makes the pairs $(p_b,S_b)$ independent across $b$, proving
the joint assertion for $Z_b^{\mathrm{loc}}=S_bA_b$.

Let $L^2=\sum_{b\in J}\lambda_b^2$. On $\{L>0\}$, the coefficients
are fixed under conditioning, so for every $t\in\mathbb R$,
\begin{equation}\label{proof:cal-normal-cf}
 \mathbb E(e^{itN_B}\mid\mathcal I_B)
 =\prod_{b\in J}\exp\left(-\frac{t^2\lambda_b^2}{2L^2}\right)
 =e^{-t^2/2}.
\end{equation}
This proves \eqref{eq:combination}, including negative real
coefficients. If $J$ is empty or $L=0$, there is no normalized score
and the prescribed combined probability is one. In particular, for
$0<\alpha<1$ its conditional rejection probability is
$\alpha\mathbf1_{\{L>0\}}$, with $L=0$ on an empty active set.

For the projection extension, form
$\widetilde G_{b\ell}=\widehat{\mathbf a}_{b\ell}^{\mathsf T}
\mathbf G_b$ from finitely many $\mathcal I_B$-measurable directions.
Their Gram entries are
\[
 \widetilde G_{b\ell}\overline{\widetilde G_{br}}
 =\widehat{\mathbf a}_{b\ell}^{\mathsf T}
   \mathbf M_b\widehat{\mathbf a}_{br},
\]
and hence are already $\mathcal I_B$-measurable. Omit zero projected
coordinates, then omit a block if all its projected coordinates
vanish. On each resulting active set, the proof of
Proposition~\ref{proof:cal-finite-projections} applies to the original
block phase. Equations~\eqref{proof:cal-product-haar}--\eqref{proof:cal-normal-cf}
remain valid. A sign may be obtained from any invariantly selected
real projection with positive energy in the block, for example the
first active selected direction. No conditioning on the data outside
the original $\mathcal I_B$ is introduced.
\end{proof}

\begin{corollary}[Explicit score for a single projection]
\label{proof:cal-angular}
Under the assumptions of Theorem~\ref{thm:joint},
let $\widehat{\mathbf a}_b\in\mathbb R^{m_b}$ be an
$\mathcal I_B$-measurable unit vector, and define
\[
  \xi_b=\widehat{\mathbf a}_b^\top\mathbf X_b,
  \qquad
  \eta_b=\widehat{\mathbf a}_b^\top\mathbf Y_b,
  \qquad
  r_b=(\xi_b^2+\eta_b^2)^{1/2}.
\]
The signed-normal score for this projection satisfies
\begin{equation}\label{proof:cal-angular-expression}
  Z_b^{\mathrm{loc}}(\widehat{\mathbf a}_b)
  =\Phi^{-1}\left\{
      \frac12+\frac1\pi\arcsin(\xi_b/r_b)
    \right\}
  \quad\text{almost surely on }\{r_b>0\}.
\end{equation}
On the same event, its conditional distribution given
$\mathcal I_B$ is $\mathcal N(0,1)$.
\end{corollary}

\begin{proof}
On $\{r_b>0\}$, we have $\eta_b\ne0$ almost surely under
the conditional null law. The one-projection construction gives
\[
  p_b=1-\frac2\pi\arctan|\xi_b/\eta_b|,
  \qquad
  Z_b^{\mathrm{loc}}(\widehat{\mathbf a}_b)
  =\operatorname{sgn}(\xi_b)\Phi^{-1}(1-p_b/2).
\]
Since
\[
  \arctan|\xi_b/\eta_b|=\arcsin(|\xi_b|/r_b),
  \qquad
  1-p_b/2=\frac12+\frac1\pi\arcsin(|\xi_b|/r_b),
\]
the identity $\Phi^{-1}(1-u)=-\Phi^{-1}(u)$ and the
oddness of $\arcsin$ yield
\eqref{proof:cal-angular-expression}.
Equation~\eqref{proof:cal-angular-expression} is
\eqref{eq:angular}. The conditional normal law follows from
Theorem~\ref{thm:joint}.
\end{proof}

\paragraph{Hilbert-space extension.}

Let $\mathcal H$ and
$\mathcal H_{\mathbb C}$ be a nontrivial real separable Hilbert space and its complexification, respectively.  We use the complex inner product on $\mathcal H_{\mathbb C}$,
linear in its first argument and conjugate-linear in its second.
For $\mathbf G=\mathbf Y+i\mathbf X$, with
$\mathbf X,\mathbf Y\in\mathcal H$, define
\[
  \mathbf M\mathbf z
  =\langle\mathbf z,\mathbf G\rangle\mathbf G,
  \qquad \mathbf z\in\mathcal H_{\mathbb C}.
\]
Then $\mathbf M$ is a positive trace-class operator of rank
at most one, with
\[
  \operatorname{tr}(\mathbf M)
  =\|\mathbf M\|_1
  =\|\mathbf G\|_{\mathcal H_{\mathbb C}}^2,
\]
where $\|\cdot\|_1$ denotes the trace norm.
All measurability statements for trace-class operators below
refer to the Borel $\sigma$-field induced by this norm. The map from $\mathbf G$ to $\mathbf M$ is continuous since
\[
 \|\mathbf g\otimes\mathbf g-\mathbf h\otimes\mathbf h\|_1
 \le(\|\mathbf g\|+\|\mathbf h\|)\|\mathbf g-\mathbf h\|,
\]
where $(\mathbf g\otimes\mathbf h)\mathbf z=
\langle\mathbf z,\mathbf h\rangle\mathbf g$.

Fix an orthonormal basis $(\mathbf e_j)$ of $\mathcal H$,
finite or countably infinite. For $\mathbf G\ne\mathbf0$,
Parseval's identity gives
\[
  \sum_j\langle\mathbf M\mathbf e_j,\mathbf e_j\rangle
  =\sum_j|\langle\mathbf G,\mathbf e_j\rangle|^2
  =\|\mathbf G\|^2>0,
\]
which implies at least one term is positive. Define
\[
  j(\mathbf M)
  =\min\{j\mid
    \langle\mathbf M\mathbf e_j,\mathbf e_j\rangle>0\},
  \qquad
  \mathbf h(\mathbf M)
  =\frac{\mathbf M\mathbf e_{j(\mathbf M)}}
    {\sqrt{\langle\mathbf M\mathbf e_{j(\mathbf M)},
                         \mathbf e_{j(\mathbf M)}\rangle}}.
\]
For each fixed $r$, the maps
$\mathbf M\mapsto\mathbf M\mathbf e_r$ and
$\mathbf M\mapsto\langle\mathbf M\mathbf e_r,\mathbf e_r\rangle$
are continuous in the trace norm, since
\[
  |\langle(\mathbf M-\mathbf N)\mathbf e_r,\mathbf e_r\rangle|
  \leq\|(\mathbf M-\mathbf N)\mathbf e_r\|
  \leq\|\mathbf M-\mathbf N\|_1.
\]
Consequently, each set
\[
  E_r=\{j(\mathbf M)=r\}
  =\{\langle\mathbf M\mathbf e_r,\mathbf e_r\rangle>0\}
   \cap\bigcap_{j<r}
   \{\langle\mathbf M\mathbf e_j,\mathbf e_j\rangle=0\}
\]
is Borel measurable. On $E_r$,
\[
  \mathbf h(\mathbf M)
  =\frac{\mathbf M\mathbf e_r}
    {\sqrt{\langle\mathbf M\mathbf e_r,\mathbf e_r\rangle}},
\]
which is continuous on that set. Since the sets $E_r$ form
a countable partition of the nonzero operators under consideration,
both $j(\mathbf M)$ and $\mathbf h(\mathbf M)$ are measurable.
At zero set $\mathbf h(\mathbf0)=\mathbf0$ and $\Theta=0$.
For $\mathbf G\ne\mathbf0$, put $j=j(\mathbf M)$,
$g_j=\langle\mathbf G,\mathbf e_j\rangle$ and
$\Theta=\arg g_j\in[0,2\pi)$. Then
\[
 \mathbf h(\mathbf M)=\frac{\overline{g_j}}{|g_j|}\mathbf G,
 \qquad \mathbf G=e^{i\Theta}\mathbf h(\mathbf M).
\]
The phase is measurable, so
\eqref{proof:cal-haar-single}--\eqref{proof:cal-orbit-kernel}
hold unchanged for $\mathcal H_{\mathbb C}$-valued observations.

\begin{proposition}[Hilbert-space projections]
\label{proof:cal-hilbert}
Suppose $\mathbf X_b,\mathbf Y_b\in\mathcal H_b$, where each
$\mathcal H_b$ is a real separable Hilbert space, and assume joint
product-rotation invariance together with $V$. Let
$\mathcal I_B=\sigma(\mathbf M_1,\ldots,\mathbf M_B,V)$, with
$\mathbf M_b$ defined as above. In each block choose finitely many
$\mathcal I_B$-measurable vectors $\widehat{\mathbf a}_{b\ell}
\in\mathcal H_b$. Form the complex projected coordinates
\[
 \widetilde G_{b\ell}=
 \langle\mathbf Y_b,\widehat{\mathbf a}_{b\ell}\rangle
 +i\langle\mathbf X_b,\widehat{\mathbf a}_{b\ell}\rangle.
\]
All conclusions of Proposition~\ref{proof:cal-finite-projections}
and Theorem~\ref{thm:joint} hold for these projected coordinates,
with zero coordinates and zero projected blocks omitted, conditional
on the original $\mathcal I_B$.
\end{proposition}

\begin{proof}
For each selected direction, define the
$\mathcal I_B$-measurable complex number
\[
  h_{b\ell}
  =\langle\mathbf h(\mathbf M_b),
           \widehat{\mathbf a}_{b\ell}\rangle.
\]
By the canonical representation and linearity of the complex
inner product in its first argument, on
$\{\mathbf G_b\ne\mathbf0\}$,
\[
  \widetilde G_{b\ell}
  =\langle\mathbf G_b,\widehat{\mathbf a}_{b\ell}\rangle
  =e^{i\Theta_b}h_{b\ell}.
\]
In particular,
\[
  |\widetilde G_{b\ell}|^2
  =|h_{b\ell}|^2
  =\langle\mathbf M_b\widehat{\mathbf a}_{b\ell},
                      \widehat{\mathbf a}_{b\ell}\rangle,
\]
which gives zero projected energy when
$\mathbf G_b=\mathbf0$. Thus the sets
\[
  \mathcal A_b=\{\ell\mid |h_{b\ell}|>0\},
  \qquad
  J=\{b\mid\mathcal A_b\ne\varnothing\}
\]
are $\mathcal I_B$-measurable and identify the nonzero
projected coordinates and blocks.

The averaging argument in \eqref{proof:cal-product-haar}
uses only product-rotation invariance and a measurable
phase representation, both of which hold here.
Consequently, conditional on $\mathcal I_B$, the phases
$(\Theta_b)_{b\in J}$ are independent and uniform on
$[0,2\pi)$, while all $h_{b\ell}$ are fixed.

For $b\in J$, let $w_{b\ell}$ be the permitted
$\mathcal I_B$-measurable nonnegative weights, supported
on $\mathcal A_b$ and summing to one. Then
\[
  T_b
  =\sum_{\ell\in\mathcal A_b}
     w_{b\ell}q(\widetilde G_{b\ell})
  =\sum_{\ell\in\mathcal A_b}
     w_{b\ell}\tan(\Theta_b+\arg h_{b\ell})
  \quad\text{almost surely}.
\]
Lemma~\ref{proof:cal-tangent} and conditional phase
independence therefore give independent standard Cauchy
laws for $(T_b)_{b\in J}$, conditional on $\mathcal I_B$.
Choose an $\mathcal I_B$-measurable index
$\ell_b\in\mathcal A_b$, and define
\[
  S_b=\operatorname{sgn}
      \langle\mathbf X_b,\widehat{\mathbf a}_{b\ell_b}\rangle,
  \qquad
  p_b=1-\frac2\pi\arctan|T_b|,
  \qquad
  Z_b^{\mathrm{loc}}=S_b\Phi^{-1}(1-p_b/2).
\]
The positive projected energy ensures that $S_b\in\{-1,1\}$
almost surely. Replacing $\Theta_b$ by $\Theta_b+\pi$
leaves $T_b$ unchanged and reverses $S_b$.
Equation~\eqref{proof:cal-half-turn} therefore applies,
so $S_b$ is conditionally symmetric and independent of
$p_b$. Hence the scores $(Z_b^{\mathrm{loc}})_{b\in J}$
are conditionally independent standard normal variables.

By \eqref{proof:cal-normal-cf}, their combination in
\eqref{eq:combination} is conditionally standard normal
whenever $\sum_{b\in J}\lambda_b^2>0$.
Zero projected coordinates and blocks are omitted, and
$p_{\mathrm{comb}}=1$ when this sum is zero.
All conditioning remains on the original $\mathcal I_B$.
If the numbers of selected projections are random and
$\mathcal I_B$-measurable, the argument applies on each
event specifying their finite values. These events form
a countable partition, completing the proof.
\end{proof}

\section{Representative procedures and power}
\label{proof:learning-section}

The direction in \eqref{eq:directionaverage} admits a measurable
choice even at tied eigenvalues. For a nonzero positive semidefinite
matrix $\mathbf A$, write
\[
 \lambda_1=\|\mathbf A\|_{\mathrm{op}},\qquad
 \mathbf P=\lim_{n\to\infty}(\mathbf A/\lambda_1)^n,\qquad
 j_* =\min\{j\mid\|\mathbf P\mathbf e_j\|>0\},\qquad
 \mathbf a(\mathbf A)=
 \frac{\mathbf P\mathbf e_{j_*}}{\|\mathbf P\mathbf e_{j_*}\|}.
\]
The spectral decomposition gives the displayed operator-norm limit.
Each coordinate map is measurable and
$\mathbf e_{j_*}^{\top}\mathbf a(\mathbf A)>0$.
Set $\mathbf a(\mathbf0)=\mathbf e_1$.
Applying this rule to $\widehat{\mathbf A}_B$ supplies the required
invariant sign and tie convention.

\proofanchor{proof:procedures}
\begin{proof}[Proof of Corollary~\ref{cor:representative-calibration}]
Each $E_b=\operatorname{tr}(\mathbf M_b)$ is
$\mathcal I_B$-measurable. For a selected direction $\mathbf a$,
\[
 r_b^2=\mathbf a^\top\operatorname{Re}(\mathbf M_b)\mathbf a
\]
is also $\mathcal I_B$-measurable.
Thus $\{b\mid r_b>0\}$ and the estimates
$\widehat\pi_b=(2m-2)/E_b$, $b\in\mathcal B_*$, are
$\mathcal I_B$-measurable. The pooled estimate
$\widehat P_B$ in \eqref{eq:blendestimates} and
the mixing coefficient $\widehat t_B$ in
\eqref{eq:blendmixing} are also
$\mathcal I_B$-measurable. Consequently, both
$\lambda_b^{\mathrm{rad}}$ in \eqref{eq:radstat}
and $\lambda_b^{\mathrm{blend}}$ in
\eqref{eq:blendweights} are fixed conditional
on $\mathcal I_B$.

By Theorem~\ref{thm:joint}, conditional on
$\mathcal I_B$, the scores $Z_b^{\mathrm{loc}}$
with $r_b>0$ are independent standard normal variables.
All sums and products below are over $\{b\mid r_b>0\}$.
If $D=\sum_b\lambda_b^2>0$, then for every real $t$,
\[
 \mathbb E_0\left\{
   \exp\left(it\frac{\sum_b\lambda_bZ_b^{\mathrm{loc}}}{\sqrt D}\right)
   \,\middle|\,\mathcal I_B\right\}
 =\prod_b\exp\left(-\frac{t^2\lambda_b^2}{2D}\right)
 =e^{-t^2/2},
\]
which is the conditional standard normal characteristic function.
The zero-denominator convention never rejects and is invariant.
\end{proof}

To prepare for the local power calculation in
Proposition~\ref{proof:fixedpower}, we compute two angular
moments of the score in \eqref{eq:angular}.
Writing a nonzero projected pair as
$(\xi,\eta)=(r\sin\theta,r\cos\theta)$ gives
\[
  Z(\theta)
  =\Phi^{-1}\left\{
      \frac12+\frac{\arcsin(\sin\theta)}{\pi}
    \right\},
  \qquad |\sin\theta|<1.
\]
Under the uniform null phase
$\Theta\sim\operatorname{Unif}(0,2\pi)$, this score is
defined almost surely. Since
\[
  \mathbb P(\sin\Theta\le x)
  =\frac12+\frac{\arcsin x}{\pi},
  \qquad -1\le x\le1,
\]
the variable
\[
  U=\frac12+\frac{\arcsin(\sin\Theta)}{\pi}
\]
is uniform on $(0,1)$, and $Z(\Theta)=\Phi^{-1}(U)$
is standard normal.
The covariance between $Z(\Theta)$ and $\sin\Theta$
produces the constant $\beta$ in the power formulas.
Indeed, $\sin\Theta=-\cos(\pi U)$ gives
\begin{equation}\label{proof:pow-beta}
  \mathbb E\{Z(\Theta)\sin\Theta\}
  =-\int_0^1\Phi^{-1}(u)\cos(\pi u)\,du
  =\beta>0.
\end{equation}
The integral is absolutely convergent since
\[
  \mathbb E|Z(\Theta)\sin\Theta|
  \le
  \{\mathbb E Z(\Theta)^2\,
     \mathbb E\sin^2\Theta\}^{1/2}
  =\frac1{\sqrt2}.
\]
The positivity of $\beta$ follows from
$-\Phi^{-1}(u)\cos(\pi u)>0$ for
$u\in(0,1)\setminus\{1/2\}$.
Finally, $Z(\pi-\theta)=Z(\theta)$, whereas
$\cos(\pi-\theta)=-\cos\theta$.
The uniform phase law is unchanged by this reflection.
Together with $\Phi^{-1}(1-u)=-\Phi^{-1}(u)$, these
identities give
\begin{equation}\label{proof:pow-angular-moments}
  \mathbb E\{Z(\Theta)\cos\Theta\}=0,
  \qquad
  Z(\Theta+\pi)=-Z(\Theta)
  \quad\text{almost surely}.
\end{equation}
For a fixed projection under a Gaussian location model,
these identities yield the following local power comparison
of weighting rules based on the projected radius.
\begin{proposition}\label{proof:fixedpower}
Let $(\xi_b,\eta_b)$ be independent pairs with independent
$\xi_b\sim\mathcal N(\theta_B,1)$ and $\eta_b\sim\mathcal N(0,1)$,
and put $r_b=(\xi_b^2+\eta_b^2)^{1/2}$. For a nonnegative measurable
function $w$ satisfying $0<\mathbb E_0w(R)^2<\infty$, where
$R^2\sim\chi^2_2$, define
\[
 N_{B,w}=\frac{\sum_b w(r_b)Z_b^{\mathrm{loc}}}
                  {\{\sum_b w(r_b)^2\}^{1/2}}.
\]
On a zero denominator, set $N_{B,w}=0$ and the final $p$-value to one. Under the null this
statistic is conditionally standard normal given the radii on a positive
denominator. Under $\theta_B=h/\sqrt B$ for fixed $h\in\mathbb R$,
\begin{equation}\label{proof:fixedlocalpower}
 N_{B,w}\Rightarrow\mathcal N(h\kappa_w,1),
 \qquad
 \kappa_w=\beta
 \frac{\mathbb E_0\{Rw(R)\}}{\{\mathbb E_0w(R)^2\}^{1/2}}.
\end{equation}
This coefficient is maximized by $w(r)=cr$, $c>0$, up to null sets, with maximum being $\kappa_R=\sqrt2\,\beta$.
\end{proposition}

\begin{proof}[Proof of Proposition~\ref{proof:fixedpower}]
Under the null, write $(\xi_b,\eta_b)=(r_b\sin\Theta_b,
r_b\cos\Theta_b)$. The phases are independent and uniform,
and independent of the radii. By Corollary~\ref{proof:cal-angular} and independence
of the phases, conditional on the radii, the scores
$Z_b^{\mathrm{loc}}=Z(\Theta_b)$ are independent
standard normal variables. Normalizing the
coefficients $w(r_b)$ proves the conditional null assertion
on a positive denominator.
Set $\sigma_w^2=\mathbb E_0w(R)^2\in(0,\infty)$ and define
\[
\begin{gathered}
 A_B=\frac{1}{\sigma_w\sqrt B}
       \sum_{b=1}^B w(r_b)Z_b^{\mathrm{loc}},
 \qquad
 \Delta_B=\frac{1}{\sqrt B}\sum_{b=1}^B\xi_b,\\
 Q_B=\frac{1}{B\sigma_w^2}\sum_{b=1}^B w(r_b)^2.
\end{gathered}
\]
Independence of $R$ and $\Theta$ gives
\[
 \mathbb E_0\{w(R)Z(\Theta)\}=0,
 \qquad
 \mathbb E_0\{w(R)^2Z(\Theta)^2\}=\sigma_w^2.
\]
Also, Cauchy--Schwarz and $\mathbb E_0R^2=2$ imply
$\mathbb E_0|Rw(R)|\le\sqrt2\,\sigma_w<\infty$.
Thus the normalized summands in $(A_B,\Delta_B)$ have zero
means, unit variances and, by \eqref{proof:pow-beta}, covariance
\begin{equation*}
 \frac{\mathbb E_0\{w(R)Z(\Theta)R\sin\Theta\}}{\sigma_w}
 =\beta\frac{\mathbb E_0\{Rw(R)\}}{\sigma_w}
 =\kappa_w.
\end{equation*}
The multivariate central limit theorem and the law of large
numbers give, under the null,
\begin{equation}\label{proof:pow-null-joint-limit}
 (A_B,\Delta_B)\Rightarrow(A,\Delta),
 \qquad
 \begin{pmatrix}A\\\Delta\end{pmatrix}
 \sim\mathcal N_2\left(
 \mathbf0,
 \begin{pmatrix}1&\kappa_w\\\kappa_w&1\end{pmatrix}
 \right),
 \qquad Q_B\xrightarrow{\mathbb P_0}1.
\end{equation}
For $\theta_B=h/\sqrt B$, the likelihood ratio relative to
the null is
\begin{equation*}
 \Lambda_B
 =\exp\left\{\theta_B\sum_{b=1}^B\xi_b
                    -\frac{B\theta_B^2}{2}\right\}
 =\exp\{h\Delta_B-h^2/2\}.
\end{equation*}
Since $\Delta_B\sim\mathcal N(0,1)$ under the null,
$\mathbb E_0\Lambda_B=1$ and
$\mathbb E_0\Lambda_B^2=e^{h^2}$.
Write $\Lambda=\exp(h\Delta-h^2/2)$.
For fixed $t\in\mathbb R$ and $K>0$, the integrand below is
a bounded continuous function of $(A_B,\Delta_B)$, so
\[
 \mathbb E_0\{e^{itA_B}\min(\Lambda_B,K)\}
 \longrightarrow
 \mathbb E\{e^{itA}\min(\Lambda,K)\}.
\]
Removing the truncation changes the left-hand side by at most
\[
 \mathbb E_0\{\Lambda_B\mathbf1_{\{\Lambda_B>K\}}\}
 \le\frac{\mathbb E_0\Lambda_B^2}{K}
 =\frac{e^{h^2}}{K}.
\]
The same bound holds for the right-hand side because
$\mathbb E\Lambda^2=e^{h^2}$. Letting $B\to\infty$ and
then $K\to\infty$ yields
\[
 \mathbb E_{h/\sqrt B}e^{itA_B}
 =\mathbb E_0(e^{itA_B}\Lambda_B)
 \longrightarrow
 \mathbb E e^{itA+h\Delta-h^2/2}
 =\exp\{it h\kappa_w-t^2/2\},
\]
where the last equality follows from the joint Gaussian law
in \eqref{proof:pow-null-joint-limit}.
Hence $A_B\Rightarrow\mathcal N(h\kappa_w,1)$ under the
local alternative.
For every $\varepsilon>0$, Cauchy--Schwarz also gives
\[
 \mathbb P_{h/\sqrt B}(|Q_B-1|>\varepsilon)
 \le e^{h^2/2}
      \mathbb P_0(|Q_B-1|>\varepsilon)^{1/2}
 \longrightarrow0.
\]
Thus $Q_B\to1$ under the local alternative and
$\mathbb P_{h/\sqrt B}(Q_B=0)\to0$.
Applying Slutsky's theorem to $N_{B,w}=A_B/\sqrt{Q_B}$
on $\{Q_B>0\}$ proves \eqref{proof:fixedlocalpower}.

Finally, the bound $\mathbb E_0\{Rw(R)\}\le\sqrt2\,\sigma_w$
gives $\kappa_w\le\sqrt2\,\beta=\kappa_R$.
Since $\beta>0$, equality holds precisely when equality holds
in Cauchy--Schwarz, namely when $w(R)=cR$ almost surely
for some $c>0$.
\end{proof}

We next consider recovery of the signal line from the
normalized paired scatter under \eqref{eq:spherical},
with $m\ge2$ and deterministic positive scales $s_b$.
Recall that $\mathbf C_b$, $\mathbf S_b$ and
$\widehat{\mathbf A}_B$ are defined in
\eqref{eq:direction} and
\eqref{eq:directionaverage}.

\proofanchor{proof:rec-theorem}
\begin{proof}[Proof of Theorem~\ref{thm:recovery}]
Fix a block and write $t=\rho/s_b$. An orthogonal change of coordinates
sends $\mathbf v$ to $\mathbf e_1$. Dividing both vectors by $s_b$
gives unit noise variance and signal magnitude $t=\rho/s_b$.
This scaling leaves $\mathbf S_b$ unchanged because the factor
$s_b^{-2}$ cancels from its numerator and denominator. Write $\varepsilon_j$ and
$\zeta_j$ for the standard Gaussian noise coordinates in this block
and define
\[
 A=(\varepsilon_1+t)^2+\zeta_1^2,
 \qquad
 D=\sum_{j=2}^m(\varepsilon_j^2+\zeta_j^2).
\]
The variables in the two sums are independent, so $A$ and $D$ are
independent. For $u\ge0$, completing the square in the Gaussian
integral gives
\begin{align*}
 \E e^{-u(\varepsilon_1+t)^2}
 &=\frac1{\sqrt{2\pi}}
   \int_{\mathbb R}e^{-u(x+t)^2-x^2/2}\,dx\\
 &=(1+2u)^{-1/2}
   \exp\left\{-\frac{ut^2}{1+2u}\right\}.
\end{align*}
Multiplying by $\E e^{-u\zeta_1^2}$ and expanding the exponential yields
\begin{align*}
 \E e^{-uA}
 &=(1+2u)^{-1}\exp\left\{-\frac{ut^2}{1+2u}\right\}\\
 &=e^{-t^2/2}\sum_{n=0}^\infty
   \frac{(t^2/2)^n}{n!}(1+2u)^{-(n+1)},
\end{align*}
where the last factor is the Laplace transform of $\chi^2_{2n+2}$.
Uniqueness of Laplace transforms therefore permits the representation
\[
 N\sim\operatorname{Poisson}(t^2/2),
 \qquad A\mid N=n\sim\chi^2_{2n+2},
 \qquad D\sim\chi^2_{2m-2},
\]
with $D$ independent of $(A,N)$.
Let $a=n+1$ and $d=m-1$. Conditional on $N=n$,
the change of variables $x=A/(A+D)$ and $z=A+D$ has Jacobian $z$, with
their joint density being
\[
 \frac{x^{a-1}(1-x)^{d-1}z^{a+d-1}e^{-z/2}}
      {2^{a+d}\Gamma(a)\Gamma(d)},
 \qquad 0<x<1,\quad z>0.
\]
Integrating over $z$ and then over $x$ gives
\begin{equation}\label{proof:rec-beta-ratio}
 \E\left(\frac A{A+D}\,\middle|\,N=n\right)
 =\frac{\Gamma(a+d)}{\Gamma(a)\Gamma(d)}
   \int_0^1x^a(1-x)^{d-1}\,dx
 =\frac{n+1}{n+m}.
\end{equation}
For $i\ne j$, choose $k\in\{i,j\}$ with $k\ge2$.
Replacing $(\varepsilon_k,\zeta_k)$ by
$(-\varepsilon_k,-\zeta_k)$ preserves the joint noise
distribution and the denominator of $\mathbf S_b$,
but reverses the numerator of $(\mathbf S_b)_{ij}$.
Hence
\[
  \mathbb E\{(\mathbf S_b)_{ij}\}
  =-\mathbb E\{(\mathbf S_b)_{ij}\}=0.
\]
Similarly, swapping any two coordinates among $2,\ldots,m$
preserves their joint distribution and exchanges the
corresponding diagonal entries. Therefore
\[
  \mathbb E\{(\mathbf S_b)_{22}\}
  =\cdots=
  \mathbb E\{(\mathbf S_b)_{mm}\}.
\]
Since $\operatorname{tr}(\mathbf S_b)=1$ almost surely,
\begin{align*}
 \E S_{b,11}
 &=\E\frac{N+1}{N+m}
   =\frac1m+\frac{m-1}{m}\E\frac N{N+m},\\
 \E S_{b,jj}
 &=\frac{1-\E S_{b,11}}{m-1}
   =\frac1m\left(1-\E\frac N{N+m}\right),
 \qquad j\ge2.
\end{align*}
Rotating back proves \eqref{eq:populationdirection}.
Averaging over blocks gives the stated expression for
$\mathbb E\widehat{\mathbf A}_B$. Its eigenvalues along
$\mathbf v$ and on its orthogonal complement are, respectively,
\[
  \lambda_{\parallel}
  =\frac{1+(m-1)\bar g_B}{m},
  \qquad
  \lambda_{\perp}
  =\frac{1-\bar g_B}{m}.
\]
If $\rho>0$, then $t_b=\rho/s_b>0$, so
$\mathbb P(N_{t_b}\ge1)>0$ and $g_m(t_b)>0$ for every block.
Thus $\bar g_B>0$, and $\operatorname{span}(\mathbf v)$
is the unique leading eigenspace, with eigenvalue gap
$\lambda_{\parallel}-\lambda_{\perp}=\bar g_B$.
Write
\[
 \mathbf A_0=\E\widehat{\mathbf A}_B,
 \qquad \mathbf E_B=\widehat{\mathbf A}_B-\mathbf A_0,
 \qquad \mathbf P_\perp=\mathbf I_m-\mathbf v\mathbf v^{\mathsf T}.
\]
Let $\widehat\lambda$ be the largest eigenvalue of
$\widehat{\mathbf A}_B$, and set
$\lambda_\perp=(1-\bar g_B)/m$.
The variational definition of $\widehat\lambda$ gives
\[
 \widehat\lambda
 \ge \mathbf v^{\mathsf T}\widehat{\mathbf A}_B\mathbf v
 \ge \lambda_\perp+\bar g_B-\|\mathbf E_B\|_{\mathrm{op}}.
\]
Projecting its eigenvector equation onto $\mathbf v^\perp$ gives
\begin{equation*}
 (\widehat\lambda-\lambda_\perp)
     \mathbf P_\perp\widehat{\mathbf a}_B
 =\mathbf P_\perp\mathbf E_B\widehat{\mathbf a}_B.
\end{equation*}
On $\|\mathbf E_B\|_{\mathrm{op}}<\bar g_B/2$, it follows that
\[
 \ell(\widehat{\mathbf a}_B,\mathbf v)
 =\|\mathbf P_\perp\widehat{\mathbf a}_B\|^2
 \le \frac{4\|\mathbf E_B\|_{\mathrm{op}}^2}{\bar g_B^2}.
\]
On the complementary event the right-hand side is at least one,
so the same inequality holds there. Independence across blocks and
zero means of $\mathbf S_b-\E\mathbf S_b$ give
\begin{align*}
 \E\|\mathbf E_B\|_F^2
 &=\frac1{B^2}\sum_{b,c}
   \E\operatorname{tr}\{(\mathbf S_b-\E\mathbf S_b)
                         (\mathbf S_c-\E\mathbf S_c)\}\\
 &=\frac1{B^2}\sum_b\E\|\mathbf S_b-\E\mathbf S_b\|_F^2
 =\mathcal V_B.
\end{align*}
Moreover,
\[
 \E\|\mathbf S_b-\E\mathbf S_b\|_F^2
 =\E\operatorname{tr}(\mathbf S_b^2)
       -\operatorname{tr}\{(\E\mathbf S_b)^2\}
 \le\E\{\operatorname{tr}(\mathbf S_b)\}^2=1.
\]
Use $\|\mathbf E_B\|_{\mathrm{op}}\le\|\mathbf E_B\|_F$ and
$0\le\ell\le1$. Orthogonal invariance makes $\mathcal V_B$
independent of the unit vector $\mathbf v$. Taking expectations and
the supremum over $\mathbf v$ proves \eqref{eq:directionrisk}.
\end{proof}

\proofanchor{proof:rec-expansion}
\begin{proof}[Proof of Remark~\ref{rem:gapexp}]

Write  $\mu=t^2/2$. The Poisson series gives
\begin{equation*}
 g_m(t)=\frac{\mu e^{-\mu}}{m+1}
       +e^{-\mu}\sum_{n=2}^\infty\frac{\mu^n}{n!}\frac n{m+n}.
\end{equation*}
For $0\le\mu\le1$,
\[
 |\mu e^{-\mu}-\mu|\le\mu^2,
 \qquad
 \PP(N\ge2)\le\frac12\E\{N(N-1)\}=\frac{\mu^2}{2}.
\]
Consequently,
\begin{equation*}
 \left|g_m(t)-\frac{t^2}{2(m+1)}\right|
 \le\left(\frac1{m+1}+\frac12\right)\frac{t^4}{4}.
\end{equation*}
This proves \eqref{eq:gapexp}. For fixed $m$ and unit scales,
$\bar g_B/\rho_B^2\longrightarrow1/\{2(m+1)\}$ as
$\rho_B\to0$. Substitution into \eqref{eq:directionrisk} yields
\begin{equation}\label{proof:rec-weak-risk}
 \sup_{\|\mathbf v\|=1}\E_{\rho_B,\mathbf v}
 \ell(\widehat{\mathbf a}_B,\mathbf v)
 \le \min\left\{1,\frac{C_m}{B\rho_B^4}\right\}
\end{equation}
for all sufficiently large $B$. Thus $B\rho_B^4\to\infty$ gives
uniformly consistent line recovery.
\end{proof}

\paragraph{Information in the complete invariant.}
\label{proof:info-section}

To establish the necessity claim in
Remark~\ref{rem:gapexp}, we study the information
available to any estimator based on the invariant matrices.
We consider model \eqref{eq:spherical} in a fixed
dimension $m\ge2$, with unit noise scales $s_b=1$
for all blocks.  Let $Q_{B,\rho,\mathbf v}$ denote the law
of $(\mathbf M_1,\ldots,\mathbf M_B)$, and let $Q_{B,0}$
denote its null law, which does not depend on $\mathbf v$.
The estimators considered below use only these matrices
and an auxiliary variable $V$, which is constant or
independent with a parameter-free distribution.

\label{proof:info-likelihood-section}

We first derive the likelihood ratio of
$Q_{B,\rho,\mathbf v}$ relative to $Q_{B,0}$ by averaging
the full Gaussian likelihood over each block's uniform
null phase. This calculation introduces the integral
\begin{equation}\label{proof:info-bessel}
  I_0(t)
  =\frac1{2\pi}\int_0^{2\pi}e^{t\cos u}\,du
  =\sum_{j=0}^{\infty}\frac{(t^2/4)^j}{(j!)^2}.
\end{equation}
Since $(j!)^2\ge j!$, comparison with the exponential series gives
\begin{equation}\label{proof:info-bessel-bound}
  1\le I_0(t)
  \le\sum_{j=0}^{\infty}\frac{(t^2/4)^j}{j!}
  =e^{t^2/4},
  \qquad t\in\mathbb R.
\end{equation}

\begin{proposition}\label{proof:info-likelihood}
The likelihood ratio of the complete invariant is
\begin{equation}\label{proof:info-product-likelihood}
 \mathcal L_{B,\rho,\mathbf v}
 =\frac{dQ_{B,\rho,\mathbf v}}{dQ_{B,0}}
 =\prod_{b=1}^B e^{-\rho^2/2}
       I_0\left(\rho\sqrt{\mathbf v^{\mathsf T}\mathbf C_b\mathbf v}\right).
\end{equation}
For unit vectors $\mathbf v,\mathbf w$,
\begin{equation}\label{proof:info-cross-moment}
 \E_0\{\mathcal L_{B,\rho,\mathbf v}
        \mathcal L_{B,\rho,\mathbf w}\}
 =I_0(\rho^2\mathbf v^{\mathsf T}\mathbf w)^B.
\end{equation}
In particular,
\begin{equation}\label{proof:info-tv}
 d_{\mathrm{TV}}(Q_{B,\rho,\mathbf v},Q_{B,0})
 \le\delta_B,
 \qquad
 \delta_B=\frac12\sqrt{I_0(\rho^2)^B-1}
 \le\frac12\sqrt{e^{B\rho^4/4}-1}.
\end{equation}
The convention is $d_{\mathrm{TV}}(P,Q)=\sup_A|P(A)-Q(A)|$.
\end{proposition}

\begin{proof}
For one block the full Gaussian likelihood ratio is
\[
 \mathcal L^{\mathrm{full}}_{\rho,\mathbf v}
 =\exp\{\rho\mathbf v^{\mathsf T}\mathbf X-\rho^2/2\}.
\]
For every bounded measurable function $f$ of the entire complex
matrix $\mathbf M$,
\begin{align*}
 \E_{\rho,\mathbf v}f(\mathbf M)
 &=\E_0\{f(\mathbf M)\mathcal L^{\mathrm{full}}_{\rho,\mathbf v}\}\\
 &=\E_0\left[f(\mathbf M)
    \E_0\{\mathcal L^{\mathrm{full}}_{\rho,\mathbf v}\mid\mathbf M\}\right].
\end{align*}
Applying
\eqref{proof:cal-orbit-kernel} under the null to
$\mathcal L^{\mathrm{full}}_{\rho,\mathbf v}$ gives
\begin{align}
 \E_0\{\mathcal L^{\mathrm{full}}_{\rho,\mathbf v}\mid\mathbf M\}
 &=\frac{e^{-\rho^2/2}}{2\pi}
   \int_0^{2\pi}
   e^{\rho\sqrt{\mathbf v^{\mathsf T}\mathbf C\mathbf v}\sin u}\,du\notag\\
 &=e^{-\rho^2/2}I_0\left(\rho\sqrt{\mathbf v^{\mathsf T}\mathbf C\mathbf v}\right).
 \label{proof:info-one-likelihood}
\end{align}
Thus \eqref{proof:info-one-likelihood} is a likelihood ratio on the
space of complete invariant matrices, whether or not that space has
a Lebesgue density. Independence across blocks gives
\eqref{proof:info-product-likelihood}.

Set $F_{\mathbf v}(u)=\mathbf v^{\mathsf T}
(\mathbf X\cos u+\mathbf Y\sin u)$ under the null. Joint Gaussianity gives
\[
 \E_0F_{\mathbf v}(u)^2=1,
 \qquad
 \E_0\{F_{\mathbf v}(u)F_{\mathbf w}(z)\}
 =\mathbf v^{\mathsf T}\mathbf w\cos(u-z).
\]
Using the phase integral for each likelihood ratio, Tonelli's theorem
and the Gaussian exponential integral yields
\begin{align}
 \E_0\{\mathcal L_{1,\rho,\mathbf v}\mathcal L_{1,\rho,\mathbf w}\}
 &=\frac{e^{-\rho^2}}{(2\pi)^2}
   \int_0^{2\pi}\int_0^{2\pi}
   \E_0 e^{\rho F_{\mathbf v}(u)+\rho F_{\mathbf w}(z)}\,du\,dz\notag\\
 &=\frac1{(2\pi)^2}\int_0^{2\pi}\int_0^{2\pi}
   e^{\rho^2\mathbf v^{\mathsf T}\mathbf w\cos(u-z)}\,du\,dz\notag\\
 &=I_0(\rho^2\mathbf v^{\mathsf T}\mathbf w).
 \label{proof:info-one-cross}
\end{align}
Taking products proves \eqref{proof:info-cross-moment}.
Since $\E_0\mathcal L_{B,\rho,\mathbf v}=1$, Cauchy--Schwarz gives
\[
 d_{\mathrm{TV}}(Q_{B,\rho,\mathbf v},Q_{B,0})
 =\frac12\E_0|\mathcal L_{B,\rho,\mathbf v}-1|
 \le\frac12\sqrt{\E_0\mathcal L_{B,\rho,\mathbf v}^2-1}.
\]
Combining \eqref{proof:info-cross-moment} with
\eqref{proof:info-bessel-bound} proves \eqref{proof:info-tv}.
Appending independent $V$ leaves each likelihood ratio and these
bounds unchanged.
\end{proof}

\label{proof:info-risk-section}

We now use Proposition~\ref{proof:info-likelihood} to
derive lower bounds on the recovery risk of any estimator
based on the invariant information. Define
\[
  \mathcal R_B^*(\rho)
  =
  \inf_{\widetilde{\mathbf a}_B}
  \sup_{\|\mathbf v\|=1}
  \mathbb E_{\rho,\mathbf v}
  \left\{1-
    (\widetilde{\mathbf a}_B^{\mathsf T}\mathbf v)^2
  \right\},
\]
where the infimum is over all $\mathcal I_B$-measurable
unit-vector estimators, allowing independent randomization.
Thus $\mathcal R_B^*(\rho)$ is the smallest attainable
worst-case risk using the invariant information.
The benchmark $1-1/m$ is the risk of an independent
uniformly distributed direction.
Recall $\delta_B$ from \eqref{proof:info-tv}, and write
$[x]_+=\max(x,0)$.
\begin{theorem}\label{proof:info-risk}
For every $B$ and $\rho\ge0$,
\begin{equation}\label{proof:info-risk-null}
 \left[1-\frac1m-\delta_B\right]_+
 \le\mathcal R_B^*(\rho)\le1-\frac1m.
\end{equation}
For $0\le\rho\le1$, also
\begin{equation}\label{proof:info-risk-two-point}
 \mathcal R_B^*(\rho)
 \ge\frac14\exp\{-c_0B\rho^4\},
 \qquad c_0=\frac12e^{3/4}.
\end{equation}
When $B\rho_B^4\to0$, the first bound implies
$\mathcal R_B^*(\rho_B)\to1-1/m$, so no invariant estimator
can asymptotically improve on the random-direction
benchmark in worst-case risk.
The second bound shows that uniformly consistent recovery
requires $B\rho_B^4\to\infty$.
Together with Theorem~\ref{thm:recovery}, this condition
is therefore necessary and sufficient when $m$ is fixed,
$s_b=1$ and $\rho_B\to0$.
\end{theorem}

\begin{proof}
For $0\le f\le1$, Fubini's theorem gives
\[
 |\E_Pf-\E_Qf|
 =\left|\int_0^1\{P(f>t)-Q(f>t)\}\,dt\right|
 \le d_{\mathrm{TV}}(P,Q).
\]
Apply this to
$f_j=(\widetilde{\mathbf a}_B^{\mathsf T}\mathbf e_j)^2$ at
$\mathbf v=\mathbf e_j$. Since $\sum_{j=1}^mf_j=1$ pointwise,
\begin{align}
 \sup_{\|\mathbf v\|=1}\E_{\rho,\mathbf v}
       \ell(\widetilde{\mathbf a}_B,\mathbf v)
 &\ge\frac1m\sum_{j=1}^m\E_{\rho,\mathbf e_j}(1-f_j)\notag\\
 &\ge1-\frac1m\sum_{j=1}^m\E_0 f_j-\delta_B
 =1-\frac1m-\delta_B.
 \label{proof:info-coordinate-risk}
\end{align}
Since the risk is nonnegative, \eqref{proof:info-coordinate-risk}
yields the lower bound in \eqref{proof:info-risk-null}.
An independent vector
$\mathbf U=\mathbf Z/\|\mathbf Z\|$, with
$\mathbf Z\sim\mathcal N_m(\mathbf0,\mathbf I_m)$, satisfies
$\E\mathbf U\mathbf U^{\mathsf T}=\mathbf I_m/m$ and has risk
$1-1/m$ at every $\mathbf v$. This proves both bounds in
\eqref{proof:info-risk-null}.

For \eqref{proof:info-risk-two-point}, take orthogonal unit vectors
$\mathbf v,\mathbf w$. Write $p$ and $q$ for their one-block invariant
likelihoods relative to $Q_{1,0}$. By
\eqref{proof:info-one-likelihood}, $q\ge e^{-\rho^2/2}$. The
inequality $\log x\le x-1$ and \eqref{proof:info-one-cross} give
\begin{align*}
 K(Q_{1,\rho,\mathbf v},Q_{1,\rho,\mathbf w})
 &=\E_0\{p\log(p/q)\}\\
 &\le\E_0\{(p-q)^2/q\}\\
 &\le e^{\rho^2/2}\E_0(p-q)^2
 =2e^{\rho^2/2}\{I_0(\rho^2)-1\}\\
 &\le\frac12e^{3/4}\rho^4,
 \qquad 0\le\rho\le1,
\end{align*}
where the final step uses $e^x-1\le xe^x$ for $x\ge0$ and
\eqref{proof:info-bessel-bound}. Here $K(P,Q)=\E_P\log(dP/dQ)$.
The logarithm of a product likelihood is a sum, so
\[
 K(Q_{B,\rho,\mathbf v},Q_{B,\rho,\mathbf w})\le c_0B\rho^4.
\]
Let $P,Q$ be these two product laws, with the common independent
randomization adjoined when used. Let $p_B,q_B$ be their densities
relative to a common dominating measure. Their affinity satisfies
\[
 A=\int\sqrt{p_Bq_B}
   =\E_P\exp\left\{-\frac12\log\frac{p_B}{q_B}\right\}
   \ge e^{-K(P,Q)/2}
\]
by Jensen's inequality. Cauchy--Schwarz gives
\begin{align*}
 d_{\mathrm{TV}}(P,Q)
 &=\frac12\int|\sqrt{p_B}-\sqrt{q_B}|
                    (\sqrt{p_B}+\sqrt{q_B})\\
 &\le\sqrt{1-A^2}.
\end{align*}
Thus
$1-d_{\mathrm{TV}}(P,Q)\ge A^2/2\ge e^{-K(P,Q)}/2$.
For every unit $\mathbf a$,
$\ell(\mathbf a,\mathbf v)+\ell(\mathbf a,\mathbf w)\ge1$.
It follows that
\begin{align*}
 \E_P\ell(\widetilde{\mathbf a}_B,\mathbf v)
 +\E_Q\ell(\widetilde{\mathbf a}_B,\mathbf w)
 &\ge\int\min(p_B,q_B)
   \{\ell(\widetilde{\mathbf a}_B,\mathbf v)
     +\ell(\widetilde{\mathbf a}_B,\mathbf w)\}\\
 &\ge1-d_{\mathrm{TV}}(P,Q)
 \ge\frac12e^{-c_0B\rho^4}.
\end{align*}
The maximum of the two risks is at least half their sum. Taking the
infimum over estimators proves \eqref{proof:info-risk-two-point}.

If $B\rho_B^4\to0$, \eqref{proof:info-tv} makes $\delta_B\to0$,
and \eqref{proof:info-risk-null} gives the claimed limit. If
$B\rho_B^4$ fails to diverge, it has a bounded subsequence on which
\eqref{proof:info-risk-two-point} is bounded away from zero.
Uniform consistency is then impossible. Conversely, \eqref{proof:rec-weak-risk} supplies an invariant estimator with
uniformly vanishing risk when $B\rho_B^4\to\infty$.
\end{proof}

\paragraph{Conditional phase moments.}
\proofanchor{proof:pow-section}

We next derive conditional moment bounds under the alternative
for the signed-normal scores $Z_b^{\mathrm{loc}}(\mathbf a)$
defined in \eqref{eq:angular}. Lemma~\ref{proof:pow-moments}
provides the bounds used in Theorems~\ref{thm:bridge}
and~\ref{thm:blendpower}.
Under \eqref{eq:spherical}, allow deterministic
positive scales $s_b$, write $\rho=\rho_B$, and let
$\mathbf a$ be any $\mathcal I_B$-measurable unit vector.
Define
\begin{equation*}
  r_b^2=\mathbf a^\top\mathbf C_b\mathbf a,\qquad
  u_b^2=\mathbf v^\top\mathbf C_b\mathbf v,\qquad
  \gamma_b=\rho/s_b^2,\qquad
  \omega_b=\gamma_bu_b.
\end{equation*}
Here $r_b$ and $u_b$ are the paired radii along the selected
direction $\mathbf a$ and the signal direction $\mathbf v$,
respectively, with $r_b,u_b\le\sqrt{E_b}$.
The quantity $\omega_b$ bounds the phase-dependent term
$\gamma_b\mathbf v^\top\mathbf X_b$ in the log likelihood
and controls the remainders in Lemma~\ref{proof:pow-moments}.
For $r_b>0$, write $Z_b=Z_b^{\mathrm{loc}}(\mathbf a)$.
Zero projected pairs are omitted.

\begin{lemma}\label{proof:pow-moments}
Conditional on $\mathcal I_B$, the scores $Z_b$
from blocks with $r_b>0$ are independent.
For an absolute constant $K$ and
$c_j=\mathbb E|W|^j$, $W\sim\mathcal N(0,1)$, they satisfy
\begin{align}
 \left|\mathbb E_\rho(Z_b\mid\mathcal I_B)
  -\beta\gamma_b\frac{\mathbf a^\top\mathbf C_b\mathbf v}{r_b}\right|
 &\le K\omega_b^3e^{\omega_b},
 \label{proof:pow-conditional-mean-bound}\\
 |\operatorname{Var}_\rho(Z_b\mid\mathcal I_B)-1|
 &\le K\omega_b^2e^{2\omega_b},
 \label{proof:pow-conditional-variance-bound}\\
 \mathbb E_\rho\bigl(
 |Z_b-\mathbb E_\rho(Z_b\mid\mathcal I_B)|^3
 \mid\mathcal I_B\bigr)
 &\le8c_3e^{\omega_b}.
 \label{proof:pow-conditional-third-bound}
\end{align}
The bounds hold simultaneously for every selected unit direction.
\end{lemma}

\begin{proof}
Fix $\mathbf a$ and put $J(\mathbf a)=\{b\mid r_b>0\}$.
For $b\in J(\mathbf a)$, define
\[
 \mathbf H_b=\frac{\mathbf M_b\mathbf a}{r_b}
            =\mathbf p_b+i\mathbf q_b.
\]
Since $|\mathbf G_b^*\mathbf a|=r_b$, direct multiplication gives
\begin{equation}\label{proof:pow-orbit-coordinates}
 \begin{gathered}
 \mathbf H_b\mathbf H_b^*=\mathbf M_b,\qquad
 \mathbf a^\top\mathbf H_b=r_b,\qquad
 \mathbf C_b\mathbf a=r_b\mathbf p_b,\\
 \mathbf X_b(\theta)=\mathbf p_b\sin\theta+
                         \mathbf q_b\cos\theta,\qquad
 \mathbf a^\top\mathbf X_b(\theta)=r_b\sin\theta.
 \end{gathered}
\end{equation}
Equation~\eqref{eq:angular} therefore becomes
\[
 Z_b(\theta)=Z(\theta)
 =\Phi^{-1}\{1/2+\pi^{-1}\arcsin(\sin\theta)\}.
\]
With $f_b(\theta)=\mathbf v^\top\mathbf X_b(\theta)$,
\[
 u_b^2=(\mathbf v^\top\mathbf p_b)^2+
       (\mathbf v^\top\mathbf q_b)^2,
 \qquad |\gamma_bf_b(\theta)|\le\gamma_bu_b=\omega_b.
\]
All representatives $\mathbf H_b$ and the set $J(\mathbf a)$
are fixed conditional on $\mathcal I_B$.

The full Gaussian likelihood ratio and its conditional null mean are
\[
 \Lambda_B=\prod_{b=1}^B
   e^{\gamma_b\mathbf v^\top\mathbf X_b-\rho^2/(2s_b^2)},
 \qquad
 \mathbb E_0(\Lambda_B\mid\mathcal I_B)
 =\prod_{b=1}^B e^{-\rho^2/(2s_b^2)}I_0(\omega_b),
\]
by \eqref{proof:cal-product-haar} and \eqref{proof:info-bessel}.
For bounded measurable $F$,
\[
 \mathbb E_\rho(F\mid\mathcal I_B)
 =\frac{\mathbb E_0(F\Lambda_B\mid\mathcal I_B)}
        {\mathbb E_0(\Lambda_B\mid\mathcal I_B)}.
\]
After integrating out phases from blocks outside $J(\mathbf a)$,
\begin{equation}\label{proof:pow-conditional-density}
 \mathbb P_\rho\{(\Theta_b)_{b\in J(\mathbf a)}
                      \in d\boldsymbol\theta\mid\mathcal I_B\}
 =\prod_{b\in J(\mathbf a)}
   \frac{e^{\gamma_b f_b(\theta_b)}}{I_0(\omega_b)}
   \frac{d\theta_b}{2\pi}.
\end{equation}
The product in \eqref{proof:pow-conditional-density} gives conditional
independence of $Z_b=Z(\Theta_b)$, $b\in J(\mathbf a)$.

Write $\langle g\rangle_0=(2\pi)^{-1}\int_0^{2\pi}g(\theta)\,d\theta$.
Equations \eqref{proof:pow-beta}, \eqref{proof:pow-angular-moments}
and \eqref{proof:pow-orbit-coordinates} imply
\begin{equation}\label{proof:pow-score-signal}
 \langle r_bZ_b\mathbf X_b\rangle_0=\beta\mathbf C_b\mathbf a,
 \qquad
 \langle Z_b f_b\rangle_0=
 \beta\frac{\mathbf a^\top\mathbf C_b\mathbf v}{r_b}.
\end{equation}
Moreover, $\langle|Z_b|^j\rangle_0=c_j$, $c_2=1$, and
\[
 \begin{gathered}
 Z_b(\theta+\pi)=-Z_b(\theta),\qquad
 f_b(\theta+\pi)=-f_b(\theta),\\
 \langle Z_b\rangle_0=\langle Z_bf_b^2\rangle_0
 =\langle Z_b^2f_b\rangle_0=\langle f_b\rangle_0=0.
 \end{gathered}
\]
Using
$|e^x-1-x|\le x^2e^{|x|}/2$ and
$|e^x-1-x-x^2/2|\le |x|^3e^{|x|}/6$ gives
\begin{align*}
 |\langle Z_be^{\gamma_bf_b}\rangle_0
      -\gamma_b\langle Z_bf_b\rangle_0|
 &\le c_1\omega_b^3e^{\omega_b}/6,
 \\
 0\le I_0(\omega_b)-1
 &=\langle e^{\gamma_bf_b}-1-\gamma_bf_b\rangle_0
 \le \omega_b^2e^{\omega_b}/2,
 \\
 |\langle Z_b^2e^{\gamma_bf_b}\rangle_0-1|
 &\le \omega_b^2e^{\omega_b}/2.
\end{align*}
By \eqref{proof:pow-conditional-density}, for $j=1,2,3$,
\begin{equation*}
 \mathbb E_\rho(Z_b^j\mid\mathcal I_B)
 =\frac{\langle Z_b^je^{\gamma_bf_b}\rangle_0}{I_0(\omega_b)},
 \qquad
 \mathbb E_\rho(|Z_b|^j\mid\mathcal I_B)\le c_je^{\omega_b}.
\end{equation*}
In particular, with $\bar z_b=\mathbb E_\rho(Z_b\mid\mathcal I_B)$,
\begin{align}
 |\bar z_b-\gamma_b\langle Z_bf_b\rangle_0|
 &\le c_1\omega_b^3e^{\omega_b}/6
       +c_1\omega_b\{I_0(\omega_b)-1\}
 \le K\omega_b^3e^{\omega_b},
 \label{proof:pow-tilted-mean-error}\\
 |\bar z_b|
 &=\frac{|\langle Z_b(e^{\gamma_bf_b}-1)\rangle_0|}{I_0(\omega_b)}
 \le c_1\omega_be^{\omega_b},\notag\\
 |\operatorname{Var}_\rho(Z_b\mid\mathcal I_B)-1|
 &\le \omega_b^2e^{\omega_b}+\bar z_b^2
 \le (1+c_1^2)\omega_b^2e^{2\omega_b}.
 \label{proof:pow-tilted-variance-error}
\end{align}
Equations \eqref{proof:pow-score-signal} and
\eqref{proof:pow-tilted-mean-error} prove
\eqref{proof:pow-conditional-mean-bound}, while
\eqref{proof:pow-tilted-variance-error} proves
\eqref{proof:pow-conditional-variance-bound}.
Finally, $|x-y|^3\le4(|x|^3+|y|^3)$ and conditional Jensen's inequality give
\[
 \mathbb E_\rho(|Z_b-\bar z_b|^3\mid\mathcal I_B)
 \le8\mathbb E_\rho(|Z_b|^3\mid\mathcal I_B)
 \le8c_3e^{\omega_b},
\]
which proves \eqref{proof:pow-conditional-third-bound}.
The constant $K$ is independent of $\mathbf a$ and the invariant matrices.
\end{proof}

We now apply Lemma~\ref{proof:pow-moments} with $s_b=1$
and coefficients $r_b$ to prove the Radius approximation in
Theorem~\ref{thm:bridge}.

\proofanchor{proof:pow-bridge}
\begin{proof}[Proof of Theorem~\ref{thm:bridge}]
Here $m\ge2$ is fixed, $s_b=1$ for every block, and
$0\le\rho=\rho_B\to0$.
For an $\mathcal I_B$-measurable unit direction $\mathbf a$, define
\begin{equation*}
 \mathbf Q_B=\frac1B\sum_b\mathbf C_b,
 \qquad D_B=\sum_b r_b^2
          =B\mathbf a^\top\mathbf Q_B\mathbf a.
\end{equation*}
For $D_B>0$, write the Radius statistic and its conditional moments as
\[
\begin{gathered}
 N_B^{\mathrm{rad}}(\mathbf a)
 =\frac{\sum_{b\mid r_b>0}r_bZ_b}{\sqrt{D_B}},\\
 \mu_B(\mathbf a)=\mathbb E_\rho
       \{N_B^{\mathrm{rad}}(\mathbf a)\mid\mathcal I_B\},
 \qquad
 \sigma_B^2(\mathbf a)=\operatorname{Var}_\rho
       \{N_B^{\mathrm{rad}}(\mathbf a)\mid\mathcal I_B\}.
\end{gathered}
\]
For $\mu_j=\rho v_j$, the Gaussian coordinate moments give
\[
 \operatorname{Var}(X_{bj}X_{bk})=
 \begin{cases}
  1+\mu_j^2+\mu_k^2,&j\ne k,\\
  2+4\mu_j^2,&j=k,
 \end{cases}
 \qquad
 \operatorname{Var}(Y_{bj}Y_{bk})=1+\mathbf1_{\{j=k\}}.
\]
Independence of the blocks and of $\mathbf X_b$ and $\mathbf Y_b$
therefore yields
\begin{equation*}
\begin{aligned}
 \mathbb E_\rho\mathbf Q_B
 &=2\mathbf I_m+\rho^2\mathbf v\mathbf v^\top,\\
 \mathbb E_\rho\|\mathbf Q_B-\mathbb E_\rho\mathbf Q_B\|_F^2
 &=\frac1B\sum_{j,k}
   \{\operatorname{Var}(X_{bj}X_{bk})+
     \operatorname{Var}(Y_{bj}Y_{bk})\}\\
 &=\frac{2(m+1)(m+\rho^2)}{B}.
\end{aligned}
\end{equation*}
Thus
\[
 \|\mathbf Q_B-2\mathbf I_m\|_{\mathrm{op}}
 \le\|\mathbf Q_B-\mathbb E_\rho\mathbf Q_B\|_F+\rho^2
 =O_P(B^{-1/2}+\rho^2).
\]
Define
\[
 \mathcal E_B=
 \{\|\mathbf Q_B-2\mathbf I_m\|_{\mathrm{op}}\le1\}.
\]
Then $\mathbb P_\rho(\mathcal E_B)\to1$, and on $\mathcal E_B$,
$1\le\mathbf a^\top\mathbf Q_B\mathbf a\le3$ and $D_B\ge B$
for every unit direction. In addition, the Gaussian matrix
\[
 \mathbf F_B=[\mathbf X_1,\mathbf Y_1,\ldots,
              \mathbf X_B,\mathbf Y_B]\in\mathbb R^{m\times2B}
\]
has rank $m$ almost surely when $2B\ge m$. Since
$B\mathbf Q_B=\mathbf F_B\mathbf F_B^\top$,
\begin{equation}\label{proof:pow-positive-denominator}
 \mathbb P_\rho\left\{
   D_B=B\mathbf a^\top\mathbf Q_B\mathbf a
       =\|\mathbf F_B^\top\mathbf a\|^2>0
   \text{ for all }\|\mathbf a\|=1\right\}=1
 \qquad(2B\ge m).
\end{equation}
To control the remainder terms in Lemma~\ref{proof:pow-moments}
uniformly over $\mathbf a$, use $r_b,u_b\le\sqrt{E_b}$
and define, for fixed $p,q\ge0$,
\begin{equation}\label{proof:pow-energy-envelope}
 K_{p,q,B}=\frac1B\sum_b E_b^p e^{q\rho\sqrt{E_b}}.
\end{equation}
Let $W_b=(\|\boldsymbol\varepsilon_b\|^2+
\|\boldsymbol\zeta_b\|^2)^{1/2}$. Then $W_b^2\sim\chi^2_{2m}$
and $\sqrt{E_b}\le W_b+\rho\le W_b+1$ for $0\le\rho\le1$.
Consequently,
\[
 \mathbb E_\rho K_{p,q,B}
 \le\mathbb E\{(W_b+1)^{2p}e^{q(W_b+1)}\}<\infty.
\]
The expectation is finite because the density of $W_b$ is
proportional to $w^{2m-1}e^{-w^2/2}$ for $w>0$.
The bound does not depend on $B$ or $\rho\in[0,1]$,
so Markov's inequality gives $K_{p,q,B}=O_P(1)$.

Since $r_b$ and $D_B$ are $\mathcal I_B$-measurable,
\[
 \mu_B(\mathbf a)=\frac1{\sqrt{D_B}}
   \sum_{b\mid r_b>0}r_b\mathbb E_\rho(Z_b\mid\mathcal I_B).
\]
Here $\gamma_b=\rho$ and $\omega_b=\rho u_b$.
For $r_b=0$, $\mathbf C_b\mathbf a=\mathbf0$, so the matrix
sum in the leading mean may include every block.
Applying \eqref{proof:pow-conditional-mean-bound},
$\sum_b\mathbf C_b=B\mathbf Q_B$ and $r_b,u_b\le\sqrt{E_b}$
gives, on $\mathcal E_B$,
\begin{equation}\label{proof:pow-leading-mean}
\begin{aligned}
 &\left|\mu_B(\mathbf a)-\beta\sqrt B\rho\,
   \frac{\mathbf a^\top\mathbf Q_B\mathbf v}
        {\sqrt{\mathbf a^\top\mathbf Q_B\mathbf a}}\right|\\
 &\qquad\le\frac{K\rho^3}{\sqrt{D_B}}
       \sum_b r_bu_b^3e^{\rho u_b}\\
 &\qquad\le\frac{K\rho^3}{\sqrt B}
       \sum_b E_b^2e^{\rho\sqrt{E_b}}
       =K\sqrt B\rho^3K_{2,1,B}.
\end{aligned}
\end{equation}
Conditional independence in Lemma~\ref{proof:pow-moments} gives
\[
 \sigma_B^2(\mathbf a)
 =\frac1{D_B}\sum_{b\mid r_b>0}r_b^2
       \operatorname{Var}_\rho(Z_b\mid\mathcal I_B).
\]
Using $\sum_b r_b^2=D_B$ and
\eqref{proof:pow-conditional-variance-bound}, on $\mathcal E_B$,
\begin{equation}\label{proof:pow-aggregate-variance}
\begin{aligned}
 |\sigma_B^2(\mathbf a)-1|
 &\le\frac{K\rho^2}{D_B}
        \sum_b r_b^2u_b^2e^{2\rho u_b}\\
 &\le\frac{K\rho^2}{B}
        \sum_b E_b^2e^{2\rho\sqrt{E_b}}
   =K\rho^2K_{2,2,B}=O_P(\rho^2).
\end{aligned}
\end{equation}
Equation~\eqref{proof:pow-aggregate-variance} proves
\eqref{eq:bridgevar}.

To simplify the leading mean, set
$d=\mathbf a^\top\mathbf Q_B\mathbf a$ and
$c=\mathbf a^\top\mathbf v$.
On $\mathcal E_B$, $d\ge1$, $|c|\le1$ and
$|d-2|\le\|\mathbf Q_B-2\mathbf I_m\|_{\mathrm{op}}$.
Therefore
\begin{equation}\label{proof:pow-alignment-ratio}
\begin{aligned}
 \left|\frac{\mathbf a^\top\mathbf Q_B\mathbf v}{\sqrt d}
                    -\sqrt2c\right|
 &\le\frac{|\mathbf a^\top(\mathbf Q_B-2\mathbf I_m)
                         \mathbf v|}{\sqrt d}
       +|c|\left|\frac2{\sqrt d}-\sqrt2\right|\\
 &\le2\|\mathbf Q_B-2\mathbf I_m\|_{\mathrm{op}}.
\end{aligned}
\end{equation}
Combining \eqref{proof:pow-alignment-ratio} with
\eqref{proof:pow-leading-mean} and $\kappa_R=\sqrt2\beta$ yields
\begin{equation}\label{proof:pow-final-mean}
\begin{aligned}
 &|\mu_B(\mathbf a)-\kappa_R\sqrt B\rho\,
                       \mathbf a^\top\mathbf v|\\
 &\qquad\le2\beta\sqrt B\rho
                 \|\mathbf Q_B-2\mathbf I_m\|_{\mathrm{op}}
            +K\sqrt B\rho^3K_{2,1,B}\\
 &\qquad=O_P(\sqrt B\rho^3+\rho),
\end{aligned}
\end{equation}
which proves \eqref{eq:bridgemean}.
The bounds in \eqref{proof:pow-final-mean} and
\eqref{proof:pow-aggregate-variance} depend only on
$\mathbf Q_B$ and the energies $E_b$, so they hold
simultaneously for all selected unit directions.

For the conditional normal approximation, define
\[
 \mathcal G_B=\mathcal E_B\cap
       \{K\rho^2K_{2,2,B}\le1/2\},
\]
where $K$ is the constant in \eqref{proof:pow-aggregate-variance}.
Then $\mathbb P_\rho(\mathcal G_B)\to1$, and
$1/2\le\sigma_B(\mathbf a)\le2$ for every unit direction
on $\mathcal G_B$.
Write $\bar z_b=\mathbb E_\rho(Z_b\mid\mathcal I_B)$.
Apply the Berry--Esseen inequality
\citep[Theorem~0.1]{supp-Vershynin2026} to the conditionally
independent centered variables
$r_b(Z_b-\bar z_b)/\sqrt{D_B}$, $r_b>0$.
Their total conditional variance is $\sigma_B^2(\mathbf a)$.
By \eqref{proof:pow-conditional-third-bound}, on $\mathcal G_B$,
\begin{equation}\label{proof:pow-berry-esseen}
\begin{aligned}
 &\sup_x\left|
   \mathbb P_\rho\{N_B^{\mathrm{rad}}(\mathbf a)\le x
                   \mid\mathcal I_B\}
   -\Phi\left(\frac{x-\mu_B(\mathbf a)}{\sigma_B(\mathbf a)}\right)
   \right|\\
 &\qquad\le
   \frac{C}{D_B^{3/2}\sigma_B(\mathbf a)^3}
   \sum_{b\mid r_b>0}r_b^3
       \mathbb E_\rho(|Z_b-\bar z_b|^3\mid\mathcal I_B)\\
 &\qquad\le
   \frac{8Cc_3}{D_B^{3/2}\sigma_B(\mathbf a)^3}
       \sum_b E_b^{3/2}e^{\rho\sqrt{E_b}}
   \le C'B^{-1/2}K_{3/2,1,B}.
\end{aligned}
\end{equation}
For $1/2\le\sigma\le2$, differentiation gives
$\partial_\sigma\Phi(y/\sigma)=-(y/\sigma^2)\varphi(y/\sigma)$,
where $\varphi$ is the standard normal density. Since
$\sup_z|z|\varphi(z)<\infty$, for every $\mu\in\mathbb R$,
\begin{equation}\label{proof:pow-normal-scale}
 \sup_x\left|
   \Phi\left(\frac{x-\mu}{\sigma}\right)-\Phi(x-\mu)\right|
 \le C|\sigma-1|\le C|\sigma^2-1|.
\end{equation}
Define the conditional cdf error by
\[
 d_B(\mathbf a)
 =\sup_x\left|
   \mathbb P_\rho\{N_B^{\mathrm{rad}}(\mathbf a)\le x
                   \mid\mathcal I_B\}
    -\Phi\{x-\mu_B(\mathbf a)\}\right|.
\]
Combining \eqref{proof:pow-berry-esseen},
\eqref{proof:pow-normal-scale} and
\eqref{proof:pow-aggregate-variance}, on $\mathcal G_B$,
\begin{equation}\label{proof:pow-final-cdf}
 d_B(\mathbf a)
 \le C\{B^{-1/2}K_{3/2,1,B}+\rho^2K_{2,2,B}\}.
\end{equation}
Because $\mathbb P_\rho(\mathcal G_B)\to1$ and
$K_{3/2,1,B},K_{2,2,B}=O_P(1)$, it follows that
$d_B(\mathbf a)=O_P(B^{-1/2}+\rho^2)$ uniformly over
the direction rules, as asserted in \eqref{eq:bridgecdf}.

Since $0\le d_B(\mathbf a)\le1$,
\eqref{proof:pow-final-cdf} also implies
$\mathbb E_\rho d_B(\mathbf a)=o(1)$ uniformly over the
direction rules. Applying the conditional cdf bound to the
two rejection tails, using continuity of $\Phi$, gives
\begin{equation}\label{proof:pow-tail-approximation}
 \left|\mathbb P_\rho\{|N_B^{\mathrm{rad}}(\mathbf a)|>z_\alpha\}
       -\mathbb E_\rho H_\alpha\{\mu_B(\mathbf a)\}\right|
 \le2\mathbb E_\rho d_B(\mathbf a)=o(1),
\end{equation}
where $H_\alpha(t)=1-\Phi(z_\alpha-t)+\Phi(-z_\alpha-t)$
is the normal rejection probability used in \eqref{eq:powerbridge}.
If $\sqrt B\rho_B^3\to0$, then \eqref{proof:pow-final-mean}
makes the mean approximation error $o_P(1)$.
Since $0\le H_\alpha\le1$ and
$|H_\alpha'(t)|\le2/\sqrt{2\pi}$,
\begin{equation}\label{proof:pow-power-shift}
\begin{aligned}
 &\left|\mathbb E_\rho H_\alpha\{\mu_B(\mathbf a)\}
   -\mathbb E_\rho H_\alpha
       (\kappa_R\sqrt B\rho_B\mathbf a^\top\mathbf v)\right|\\
 &\qquad\le\mathbb E_\rho\min\left\{1,
    \frac2{\sqrt{2\pi}}
    |\mu_B(\mathbf a)-\kappa_R\sqrt B\rho_B
                         \mathbf a^\top\mathbf v|\right\}
   =o(1).
\end{aligned}
\end{equation}
The bounds in \eqref{proof:pow-tail-approximation} and
\eqref{proof:pow-power-shift} are uniform over the direction
rules. Taking $\mathbf a=\widehat{\mathbf a}_B$ gives
$|\mathbf a^\top\mathbf v|=\sqrt{1-L_B}$.
The identity $H_\alpha(t)=H_\alpha(-t)$ then yields
\eqref{eq:powerbridge}.
\end{proof}

\paragraph{Detection and direction recovery.}
\proofanchor{proof:det-section}

Equation~\eqref{eq:powerbridge} expresses power through
the random alignment $1-L_B$.
For the normalized-scatter direction, the following lemma
identifies its null distribution.
This distribution is then transferred to local alternatives
to derive the explicit power limits in
Corollary~\ref{cor:localpower} and
Remark~\ref{rem:blendlocalpower}.

\begin{lemma}\label{proof:det-null-line}
Under \eqref{eq:spherical} with deterministic positive scales
and $B\ge m$, $\widehat{\mathbf A}_B$ has distinct eigenvalues
almost surely. Under the null, for every fixed unit vector $\mathbf v$,
\begin{equation}\label{proof:det-null-beta}
 (\widehat{\mathbf a}_B^\top\mathbf v)^2
 \sim\operatorname{Beta}\left(\frac12,\frac{m-1}{2}\right).
\end{equation}
\end{lemma}

\begin{proof}
We first show that $\widehat{\mathbf A}_B$ has distinct
eigenvalues almost surely. To remove the denominators
in \eqref{eq:directionaverage}, define
\[
  \mathbf H_B
  =B\left(\prod_{b=1}^B E_b\right)\widehat{\mathbf A}_B
  =\sum_{b=1}^B\mathbf C_b\prod_{c\ne b}E_c.
\]
Since $E_b>0$ almost surely, $\mathbf H_B$ has a repeated
eigenvalue exactly when $\widehat{\mathbf A}_B$ does.
Moreover, every entry of $\mathbf H_B$ is a polynomial
in the observation coordinates.
Let $D(\mathbf H_B)$ be the discriminant of
$\det(x\mathbf I_m-\mathbf H_B)$. The quantity $D(\mathbf H_B)$
is a polynomial in the $2mB$ observation coordinates, and
\[
 \{\widehat{\mathbf A}_B\text{ has a repeated eigenvalue}\}
 =\{D(\mathbf H_B)=0\}
 \quad\text{on }\{E_b>0\text{ for all }b\}.
\]
For $0<\epsilon_2<\cdots<\epsilon_m<1/4$, choose
\[
 \begin{aligned}
 \mathbf X_b&=\mathbf e_1,\quad \mathbf Y_b=\mathbf0,
 &&b=1\text{ or }b>m,\\
 \mathbf X_j&=\sqrt{1-\epsilon_j}\,\mathbf e_1,\quad
 \mathbf Y_j=\sqrt{\epsilon_j}\,\mathbf e_j,
 &&2\le j\le m.
 \end{aligned}
\]
At these coordinates, $E_b=1$ for all $b$ and
\[
 \mathbf H_B=\operatorname{diag}
 \left(B-\sum_{j=2}^m\epsilon_j,\epsilon_2,\ldots,\epsilon_m\right),
 \qquad D(\mathbf H_B)\ne0.
\]
Thus $D(\mathbf H_B)$ is a nonzero polynomial, whose zero
set has Lebesgue measure zero. The joint Gaussian law has a density,
so $\mathbb P_{\rho,\mathbf v}\{D(\mathbf H_B)=0\}=0$.

Under the null, $s_b^2$ cancels from $\mathbf S_b$ in
\eqref{eq:direction}. For every orthogonal matrix $\mathbf O$,
\(
 \widehat{\mathbf A}_B\overset d=
 \mathbf O\widehat{\mathbf A}_B\mathbf O^\top.
\)
The leading eigenvalue is simple almost surely, so the leading
projection matrix $\widehat{\mathbf P}_B=
\widehat{\mathbf a}_B\widehat{\mathbf a}_B^\top$ satisfies
\(
 \widehat{\mathbf P}_B\overset d=
 \mathbf O\widehat{\mathbf P}_B\mathbf O^\top.
\)
Take $\mathbf O$ uniform on the orthogonal group and independent of
$\widehat{\mathbf P}_B$. For each fixed unit vector $\mathbf u$,
$\mathbf O\mathbf u\overset d=\mathbf Z/\|\mathbf Z\|$, where
$\mathbf Z\sim\mathcal N_m(\mathbf0,\mathbf I_m)$. Hence
\[
 \widehat{\mathbf P}_B\overset d=
 \frac{\mathbf Z\mathbf Z^\top}{\|\mathbf Z\|^2},
 \qquad
 (\widehat{\mathbf a}_B^\top\mathbf v)^2
 \overset d=\frac{Z_1^2}{Z_1^2+\sum_{j=2}^mZ_j^2}.
\]
Since $Z_1^2\sim\chi^2_1$ is independent of
$\sum_{j=2}^mZ_j^2\sim\chi^2_{m-1}$, the change of variables
used to obtain \eqref{proof:rec-beta-ratio}, with shape parameters
$1/2$ and $(m-1)/2$, gives \eqref{proof:det-null-beta}.
\end{proof}

We combine \eqref{proof:det-null-beta} with the invariant likelihood
bound \eqref{proof:info-tv} and the conditional mean and variance
in Theorem~\ref{thm:bridge} to obtain the three detection regimes.

\proofanchor{proof:det-corollary}
\label{proof:det-local}\label{proof:det-consistency}
\begin{proof}[Proof of Corollary~\ref{cor:localpower}]
Note that $s_b=1$ and  $m\ge2$ is fixed. Let $\rho_B\to0$.
Write $\mathcal A_B=\{|N_B^{\mathrm{rad}}|>z_\alpha\}$.

\emph{(i)} Suppose $B\rho_B^2\to0$. The full-data likelihood satisfies
\begin{equation*}
 \Lambda_B^{\mathrm{full}}
 =\exp\left\{\rho_B\sum_b\mathbf v^\top\mathbf X_b-B\rho_B^2/2\right\},
 \qquad
 \mathbb E_0\{(\Lambda_B^{\mathrm{full}})^2\}=e^{B\rho_B^2}.
\end{equation*}
For $B\ge m$, \eqref{proof:pow-positive-denominator} gives
$D_B>0$ almost surely.
By Theorem~\ref{thm:joint}, $\mathbb P_0(\mathcal A_B)=\alpha$.
Therefore
\[
 \left|\mathbb P_{\rho_B,\mathbf v}(\mathcal A_B)-\alpha\right|
 \le\tfrac12\mathbb E_0|\Lambda_B^{\mathrm{full}}-1|
 \le\tfrac12\sqrt{e^{B\rho_B^2}-1}\longrightarrow0.
\]

\emph{(ii)} Suppose $\rho_B=h/\sqrt B$. Equation~\eqref{proof:info-tv} gives
\[
 d_{\mathrm{TV}}(Q_{B,\rho_B,\mathbf v},Q_{B,0})
 \le\tfrac12\sqrt{e^{h^4/(4B)}-1}\longrightarrow0.
\]
Because $V$ is independent and parameter-free, the total variation
distance is unchanged upon adjoining $V$. Applying the measurable
map $(\mathbf M_1,\ldots,\mathbf M_B,V)\mapsto1-L_B$ and
\eqref{proof:det-null-beta} yields
\begin{equation}\label{proof:det-alignment-transfer}
 d_{\mathrm{TV}}\left\{
  \mathcal L_{\rho_B,\mathbf v}(1-L_B),
  \operatorname{Beta}\left(\tfrac12,\tfrac{m-1}{2}\right)\right\}
 \le\delta_B\longrightarrow0.
\end{equation}
Since $\sqrt B\rho_B^3=h^3/B\to0$, \eqref{eq:powerbridge}
and \eqref{proof:det-alignment-transfer} give
\[
 \mathbb P_{\rho_B,\mathbf v}(\mathcal A_B)
 =\mathbb E_{\rho_B,\mathbf v}
   H_\alpha\{\kappa_Rh\sqrt{1-L_B}\}+o(1)
 \longrightarrow\mathbb E H_\alpha(\kappa_Rh\sqrt U),
 \qquad U\sim\operatorname{Beta}\left(\tfrac12,\tfrac{m-1}{2}\right).
\]

\emph{(iii)} Suppose $\nu_B=\sqrt B\rho_B\to\infty$. Define
\[
 \mu_B=\mathbb E_{\rho_B,\mathbf v}(N_B^{\mathrm{rad}}\mid\mathcal I_B),
 \qquad
 \sigma_B^2=\operatorname{Var}_{\rho_B,\mathbf v}
 (N_B^{\mathrm{rad}}\mid\mathcal I_B).
\]
Equation~\eqref{eq:bridgevar} gives $\sigma_B^2=1+o_P(1)$.
For $K>0$, conditional Chebyshev's inequality gives
\begin{equation}\label{proof:det-tight}
 \mathbb P_{\rho_B,\mathbf v}(|N_B^{\mathrm{rad}}-\mu_B|>K)
 \le\mathbb P_{\rho_B,\mathbf v}(\sigma_B^2>2)+2/K^2,
 \qquad N_B^{\mathrm{rad}}-\mu_B=O_P(1).
\end{equation}

If $B\rho_B^4\to\infty$, \eqref{eq:directionrisk} and
\eqref{eq:gapexp} imply
$|\widehat{\mathbf a}_B^\top\mathbf v|\xrightarrow{P}1$.
Dividing \eqref{eq:bridgemean} by $\nu_B$ gives
\begin{equation}\label{proof:det-relative-mean}
 \frac{\mu_B}{\nu_B}
 =\kappa_R\widehat{\mathbf a}_B^\top\mathbf v
       +O_P(\rho_B^2+B^{-1/2}),
 \qquad \frac{|\mu_B|}{\nu_B}\xrightarrow{P}\kappa_R>0.
\end{equation}

If $\sup_B B\rho_B^4=M<\infty$, then
\eqref{proof:info-cross-moment} and \eqref{proof:info-bessel-bound} give
$\mathbb E_0\mathcal L_{B,\rho_B,\mathbf v}^2\le e^{M/4}$.
For $K>0$, the event
$E_{B,K}=\{\nu_B|\widehat{\mathbf a}_B^\top\mathbf v|\le K\}$
belongs to $\mathcal I_B$. By Cauchy--Schwarz and
\eqref{proof:det-null-beta},
\begin{equation}\label{proof:det-small-alignment}
\begin{aligned}
 \mathbb P_{\rho_B,\mathbf v}(E_{B,K})
 &=\mathbb E_0(\mathcal L_{B,\rho_B,\mathbf v}\mathbf1_{E_{B,K}})\\
 &\le e^{M/8}
    \left\{\mathbb P\left(U\le K^2/\nu_B^2\right)\right\}^{1/2}
 \longrightarrow0.
\end{aligned}
\end{equation}
Here $\mathbb P(U=0)=0$ and
$\sqrt B\rho_B^3\le\sqrt M\rho_B\to0$.
Consequently, \eqref{eq:bridgemean} gives
\[
 \mu_B=\kappa_R\nu_B\widehat{\mathbf a}_B^\top\mathbf v+o_P(1),
 \qquad |\mu_B|\xrightarrow{P}\infty.
\]
Together, \eqref{proof:det-relative-mean} and
\eqref{proof:det-small-alignment} imply $|\mu_B|\xrightarrow{P}\infty$
for any sequence with $\nu_B\to\infty$. Indeed, every subsequence
has a further subsequence on which $B\rho_B^4$ is bounded or tends
to infinity. For each fixed $K>0$, \eqref{proof:det-tight} now gives
\[
 \mathbb P_{\rho_B,\mathbf v}(\mathcal A_B^c)
 \le\mathbb P_{\rho_B,\mathbf v}(|\mu_B|\le z_\alpha+K)
     +\mathbb P_{\rho_B,\mathbf v}(|N_B^{\mathrm{rad}}-\mu_B|>K),
 \qquad
 \limsup_{B\to\infty}\mathbb P_{\rho_B,\mathbf v}(\mathcal A_B^c)
 \le2/K^2.
\]
Letting $K\to\infty$ proves $\mathbb P_{\rho_B,\mathbf v}(\mathcal A_B)\to1$.

If also $B\rho_B^4\to0$, then $0\le L_B\le1$,
\eqref{proof:info-tv}, and $\mathbb E_0L_B=1-1/m$ from
\eqref{proof:det-null-beta} yield
\begin{equation*}
 \left|\mathbb E_{\rho_B,\mathbf v}L_B-
       \left(1-\frac1m\right)\right|
 \le\delta_B\longrightarrow0.
\end{equation*}
\end{proof}

\paragraph{Unequal scales and Blend.}
\proofanchor{proof:blend-section}

Theorem~\ref{thm:blendpower} extends the Radius calculation to
unequal scales. We first analyze coefficients
$r_b(p+q\widehat\pi_b)$, then optimize their leading power coefficient
and substitute the estimates in \eqref{eq:blendweights}.
Throughout the proof, $m\ge3$ is fixed,
$0<s_-\le s_b\le s_+<\infty$ and $\rho=\rho_B\to0$.
Write
\[
 d=2m,\qquad \tau=(d-2)/d,\qquad
 \widehat\pi_b=(d-2)/E_b,\qquad
 \epsilon_B=B^{-1/2}+\rho^2.
\]
The quantities $J_{+,B}$, $J_{-,B}$ and $\mathcal H_B$ are defined in
\eqref{eq:blendpopulation}.

\begin{proof}[Proof of Theorem~\ref{thm:blendpower}]
\label{proof:blend-matrices}
Recall $E_b$ and $\mathbf C_b$ from
\eqref{eq:pairedenergy} and \eqref{eq:direction}.
Under the null in \eqref{eq:spherical}, spherical
symmetry implies that
\[
  \mathbf U_b
  =\frac{1}{\sqrt{E_b}}
    \begin{pmatrix}\mathbf X_b\\ \mathbf Y_b\end{pmatrix}
\]
is uniform on the unit sphere in $\mathbb R^d$ and
independent of $E_b$.
Since $\mathbb E_0(\mathbf U_b\mathbf U_b^\top)
=\mathbf I_d/d$, we obtain
\begin{equation*}
 \begin{gathered}
 E_b/s_b^2\sim\chi^2_d,\qquad
 \mathbb E_0(\mathbf C_b\mid E_b)=\frac{2E_b}{d}\mathbf I_m,\\
 \mathbb E_0 E_b^{-q}
 =\frac{\int_0^\infty x^{d/2-q-1}e^{-x/(2s_b^2)}\,dx}
        {(2s_b^2)^{d/2}\Gamma(d/2)}
 =(2s_b^2)^{-q}\frac{\Gamma(d/2-q)}{\Gamma(d/2)},
 \qquad q<d/2.
 \end{gathered}
\end{equation*}
In particular,
\begin{align}
 \mathbb E_0\widehat\pi_b&=s_b^{-2},
 &\mathbb E_0\mathbf C_b&=2s_b^2\mathbf I_m,
 \label{proof:blend-nullmoments}\\
 \mathbb E_0(\widehat\pi_b\mathbf C_b)&=2\tau\mathbf I_m,
 &\mathbb E_0(\widehat\pi_b^2\mathbf C_b)
 &=\frac{2(d-2)^2}{d}\mathbb E_0(E_b^{-1})\mathbf I_m
 =2\tau s_b^{-2}\mathbf I_m.
 \notag
\end{align}
The second moments needed for averaging are finite because $d>4$,
\[
 \mathbb E_0\widehat\pi_b^2
 =s_b^{-4}\frac{d-2}{d-4},\qquad
 \mathbb E_0\|\widehat\pi_b^2\mathbf C_b\|_F^2
 \le(d-2)^4\mathbb E_0E_b^{-2}<\infty.
\]

For coefficients $r_b(p+q\widehat\pi_b)$, the conditional
mean expansion in \eqref{proof:pow-conditional-mean-bound}
leads to weighted sums of $\mathbf C_b/s_b^2$.
The squared normalizing denominator,
$\sum_b r_b^2(p+q\widehat\pi_b)^2$, involves weighted sums
of $\mathbf C_b$, since $r_b^2=\mathbf a^\top\mathbf C_b\mathbf a$.
We therefore define the following five matrix averages.
\begin{align}
 \mathbf A_{0B}&=\frac1B\sum_b\mathbf C_b/s_b^2,
 &\mathbf A_{1B}&=\frac1B\sum_b\widehat\pi_b\mathbf C_b/s_b^2,
 \label{proof:blend-five-matrices}\\
 \mathbf D_{0B}&=\frac1B\sum_b\mathbf C_b,
 &\mathbf D_{1B}&=\frac1B\sum_b\widehat\pi_b\mathbf C_b,
 &\mathbf D_{2B}&=\frac1B\sum_b\widehat\pi_b^2\mathbf C_b.
 \notag
\end{align}
We establish, in operator norm,
\begin{align}
 \mathbf A_{0B}&=2\mathbf I_m+O_P(\epsilon_B),
 &\mathbf A_{1B}&=2\tau J_{-,B}\mathbf I_m+O_P(\epsilon_B),
 \label{proof:blend-A-control}\\
 \mathbf D_{0B}&=2J_{+,B}\mathbf I_m+O_P(\epsilon_B),
 &\mathbf D_{1B}&=2\tau\mathbf I_m+O_P(\epsilon_B),
 \notag\\
 \mathbf D_{2B}&=2\tau J_{-,B}\mathbf I_m+O_P(\epsilon_B).
 \label{proof:blend-D-control}
\end{align}
For block $b$, use an orthonormal basis of $\mathbb R^m$
whose first vector is $\mathbf v$.
Let $\mathbf W$ collect the coordinates of $\mathbf X_b/s_b$
followed by those of $\mathbf Y_b/s_b$ in this basis.
Under the null,
\[
  \mathbf W\sim\mathcal N_d(\mathbf0,\mathbf I_d),
  \qquad W_1=\mathbf v^\top\mathbf X_b/s_b.
\]
Under \eqref{eq:spherical}, the likelihood ratio
relative to the null is
\[
  \exp\{t_bW_1-t_b^2/2\},
  \qquad t_b=\rho/s_b.
\]
Each matrix summand $\mathbf g_b$ in
\eqref{proof:blend-five-matrices}, as well as
$(E_b/d)\mathbf I_m$ and $\widehat\pi_b\mathbf I_m$,
satisfies $\mathbf g_b(-\mathbf W)=\mathbf g_b(\mathbf W)$,
because both $\mathbf C_b$ and $E_b$ are unchanged
when $\mathbf W$ is replaced by $-\mathbf W$.
Since $\mathbf W\overset d=-\mathbf W$ under the null,
\[
\begin{aligned}
  \mathbb E_\rho\mathbf g_b
  &=e^{-t_b^2/2}\mathbb E_0
      \{\mathbf g_b e^{t_bW_1}\}\\
  &=\frac{e^{-t_b^2/2}}{2}\mathbb E_0
      \{\mathbf g_b(e^{t_bW_1}+e^{-t_bW_1})\}\\
  &=\mathbb E_0
      \{\mathbf g_b e^{-t_b^2/2}\cosh(t_bW_1)\}.
\end{aligned}
\]
Subtracting $\mathbb E_0\mathbf g_b$ gives
\begin{equation}\label{proof:blend-even-tilt}
 \mathbb E_\rho \mathbf g_b-\mathbb E_0\mathbf g_b
 =\mathbb E_0[\mathbf g_b\{e^{-t_b^2/2}\cosh(t_bW_1)-1\}].
\end{equation}
For $|t_b|\le1$, Taylor's formula gives
\[
 |e^{-t_b^2/2}\cosh(t_bW_1)-1|
 \le Ct_b^2(1+W_1^2)e^{|W_1|},
 \qquad
 \|\mathbf g_b\|_F\le C(1+\|\mathbf W\|^2+\|\mathbf W\|^{-2}).
\]
In radial coordinates,
\[
\begin{aligned}
 \mathbb E_0[\|\mathbf g_b\|_F(1+W_1^2)e^{|W_1|}]
 &\le C\int_0^\infty(1+w^2+w^{-2})(1+w^2)
                  w^{d-1}e^{-w^2/2+w}\,dw<\infty,\\
 \mathbb E_\rho\|\mathbf g_b\|_F^2
 &\le C\int_0^\infty(1+w^4+w^{-4})
                  w^{d-1}e^{-w^2/2+w}\,dw<\infty.
\end{aligned}
\]
The constants are uniform over $b$ and sufficiently small $\rho$.
Equation~\eqref{proof:blend-even-tilt} and block independence yield
\begin{equation}\label{proof:blend-average-bound}
 \|\mathbb E_\rho \mathbf g_b-\mathbb E_0\mathbf g_b\|_F\le C\rho^2,
 \qquad
 \mathbb E_\rho\left\|
   \frac1B\sum_b(\mathbf g_b-\mathbb E_\rho \mathbf g_b)\right\|_F^2\le C/B.
\end{equation}
Combining \eqref{proof:blend-average-bound} with
\eqref{proof:blend-nullmoments} and
$\|\cdot\|_{\mathrm{op}}\le\|\cdot\|_F$ proves
\eqref{proof:blend-A-control}--\eqref{proof:blend-D-control}.

The factor $\widehat\pi_b=(d-2)/E_b$ in the coefficients
introduces inverse powers of $E_b$ into the weighted
score moments.
To apply the conditional moment bounds in
Lemma~\ref{proof:pow-moments}, we therefore extend
\eqref{proof:pow-energy-envelope} to
$p>-m$ and $q\ge0$.
Under the null representation, $E_b=s_b^2\|\mathbf W\|^2$ and
the alternative likelihood is at most $e^{\|\mathbf W\|}$ for
$0\le\rho\le\min(1,s_-)$. Consequently,
\begin{equation}\label{proof:blend-energy-bound}
 \sup_b\mathbb E_\rho\{E_b^p e^{q\rho\sqrt{E_b}}\}
 \le C\int_0^\infty w^{d-1+2p}
            e^{-w^2/2+(1+qs_+)w}\,dw<\infty,
 \qquad K_{p,q,B}=O_P(1).
\end{equation}

Let $\mathbf a$ be an $\mathcal I_B$-measurable unit vector.
For $r_b>0$, Lemma~\ref{proof:pow-moments} and
$\gamma_b\le\rho/s_-^2$ give
\begin{align}
 \bar z_b=\mathbb E_\rho(Z_b\mid\mathcal I_B)
 &=\beta\frac{\rho}{s_b^2}
       \frac{\mathbf a^\top\mathbf C_b\mathbf v}{r_b}+R_b,
 & |R_b|&\le C\rho^3u_b^3e^{C\rho u_b},
 \label{proof:blend-phase-mean}\\
 |\operatorname{Var}_\rho(Z_b\mid\mathcal I_B)-1|
 &\le C\rho^2u_b^2e^{C\rho u_b},
 \label{proof:blend-phase-var}\\
 \mathbb E_\rho(|Z_b-\bar z_b|^3\mid\mathcal I_B)
 &\le8c_3e^{C\rho u_b}.
 \label{proof:blend-phase-third}
\end{align}
If $r_b=0$, then $\mathbf C_b\mathbf a=\mathbf0$ and the block
contributes zero to all weighted sums below.

\label{proof:blend-uniform}
For $p,q\ge0$ with $p+q>0$, define
\begin{equation*}
 \lambda_b(p,q)=r_b(p+q\widehat\pi_b),\qquad
 F_B(p,q)=\frac{p+q\tau J_{-,B}}
 {\{p^2J_{+,B}+2pq\tau+q^2\tau J_{-,B}\}^{1/2}}.
\end{equation*}
Write
\[
\begin{gathered}
 D_B(p,q)=\sum_b\lambda_b(p,q)^2,\qquad
 N_B(p,q)=D_B(p,q)^{-1/2}\sum_{b\mid r_b>0}\lambda_b(p,q)Z_b,\\
 \mu_B(p,q)=\mathbb E_\rho\{N_B(p,q)\mid\mathcal I_B\},\qquad
 \sigma_B^2(p,q)=\operatorname{Var}_\rho\{N_B(p,q)\mid\mathcal I_B\}.
\end{gathered}
\]
Set $N_B(p,q)=0$ when $D_B(p,q)=0$.
Both $\mathbf a$ and $(p,q)$ may be $\mathcal I_B$-measurable.
Fix a compact set $\mathcal K\subset[0,\infty)^2$ excluding the origin.
All bounds below are uniform over $\|\mathbf a\|=1$ and
$(p,q)\in\mathcal K$.

Define
\[
 \mathbf A_B(p,q)=p\mathbf A_{0B}+q\mathbf A_{1B},\qquad
 \mathbf D_B(p,q)=p^2\mathbf D_{0B}+2pq\mathbf D_{1B}+q^2\mathbf D_{2B}.
\]
Then $D_B(p,q)=B\mathbf a^\top\mathbf D_B(p,q)\mathbf a$.
For $\mathbf u=(p,q)^\top$, write
\begin{equation}\label{proof:blend-G}
 \mathbf V_B=\begin{pmatrix}J_{+,B}&\tau\\\tau&\tau J_{-,B}\end{pmatrix},
 \qquad v_B(p,q)=\mathbf u^\top\mathbf V_B\mathbf u,
 \qquad \det\mathbf V_B=\tau(\mathcal H_B-\tau)\ge\tau(1-\tau)>0.
\end{equation}
The scale bounds give $\operatorname{tr}\mathbf V_B\le C$ and
$\lambda_{\min}(\mathbf V_B)\ge\det(\mathbf V_B)/
\operatorname{tr}(\mathbf V_B)\ge c>0$.
Since $\inf_{\mathcal K}(p^2+q^2)>0$,
\eqref{proof:blend-D-control} implies, with probability tending to one,
\begin{equation}\label{proof:blend-denominator}
 cB\le D_B(p,q)\le CB
 \qquad\text{for all }\|\mathbf a\|=1,\ (p,q)\in\mathcal K.
\end{equation}

By \eqref{proof:blend-phase-mean}, the leading conditional mean is
\begin{equation*}
\begin{aligned}
 \mu_B^{(0)}(p,q)
 &=\frac{\beta\rho}{\sqrt{D_B(p,q)}}
   \sum_b\frac{p+q\widehat\pi_b}{s_b^2}
                 \mathbf a^\top\mathbf C_b\mathbf v\\
 &=\beta\rho\sqrt B\,
   \frac{\mathbf a^\top\mathbf A_B(p,q)\mathbf v}
        {\{\mathbf a^\top\mathbf D_B(p,q)\mathbf a\}^{1/2}}.
\end{aligned}
\end{equation*}
For $(p,q)\in\mathcal K$,
$\lambda_b(p,q)\le C(E_b^{1/2}+E_b^{-1/2})$ and $u_b\le\sqrt{E_b}$.
On the event in \eqref{proof:blend-denominator},
\[
\begin{aligned}
 |\mu_B(p,q)-\mu_B^{(0)}(p,q)|
 &\le\frac{C\rho^3}{\sqrt{D_B(p,q)}}
             \sum_b\lambda_b(p,q)u_b^3e^{C\rho u_b}\\
 &\le C\sqrt B\rho^3\{K_{2,C,B}+K_{1,C,B}\}
 =O_P(\sqrt B\rho^3).
\end{aligned}
\]
Writing $n_B(p,q)=p+q\tau J_{-,B}$,
\eqref{proof:blend-A-control}--\eqref{proof:blend-D-control} give
\[
 \mathbf a^\top\mathbf A_B(p,q)\mathbf v
 =2n_B(p,q)\mathbf a^\top\mathbf v+O_P(\epsilon_B),\qquad
 \mathbf a^\top\mathbf D_B(p,q)\mathbf a
 =2v_B(p,q)+O_P(\epsilon_B).
\]
Because $v_B(p,q)\ge c>0$ and $n_B(p,q)\le C$ on $\mathcal K$,
\begin{equation}\label{proof:blend-general-mean}
 \mu_B(p,q)=\kappa_R\sqrt B\rho F_B(p,q)\mathbf a^\top\mathbf v
                 +O_P(\sqrt B\rho^3+\rho).
\end{equation}

Conditional independence from Lemma~\ref{proof:pow-moments} and
\eqref{proof:blend-phase-var} yield
\begin{equation}\label{proof:blend-general-var}
\begin{aligned}
 |\sigma_B^2(p,q)-1|
 &\le\frac{C\rho^2}{D_B(p,q)}
     \sum_b\lambda_b(p,q)^2u_b^2e^{C\rho u_b}\\
 &\le C\rho^2\{K_{2,C,B}+K_{1,C,B}+K_{0,C,B}\}
 =O_P(\rho^2).
\end{aligned}
\end{equation}
For the conditional normal approximation, apply the Berry--Esseen
inequality to $\lambda_b(p,q)(Z_b-\bar z_b)/\sqrt{D_B(p,q)}$.
Equations \eqref{proof:blend-phase-third}, \eqref{proof:blend-denominator}
and \eqref{proof:blend-general-var} give, on events of probability tending to one,
\[
\begin{aligned}
 &\sup_x\left|\mathbb P_\rho\{N_B(p,q)\le x\mid\mathcal I_B\}
       -\Phi\left(\frac{x-\mu_B(p,q)}{\sigma_B(p,q)}\right)\right|\\
 &\quad\le\frac{C\sum_b\lambda_b(p,q)^3e^{C\rho u_b}}
                    {D_B(p,q)^{3/2}\sigma_B(p,q)^3}
 \le CB^{-1/2}\{K_{3/2,C,B}+K_{-3/2,C,B}\}.
\end{aligned}
\]
Equation~\eqref{proof:blend-energy-bound} gives
$K_{3/2,C,B},K_{-3/2,C,B}=O_P(1)$.
Using the normal scale bound \eqref{proof:pow-normal-scale},
\begin{equation}\label{proof:blend-general-cdf}
 d_B(p,q)=\sup_x\left|
 \mathbb P_\rho\{N_B(p,q)\le x\mid\mathcal I_B\}
       -\Phi\{x-\mu_B(p,q)\}\right|
 =O_P(B^{-1/2}+\rho^2).
\end{equation}
For $\sqrt B\rho^3\to0$, let
$\eta_B(p,q)=\kappa_R\sqrt B\rho F_B(p,q)\mathbf a^\top\mathbf v$.
Equations \eqref{proof:blend-general-mean} and
\eqref{proof:blend-general-cdf} imply
\[
\begin{aligned}
 &\left|\mathbb P_\rho\{|N_B(p,q)|>z_\alpha\}
           -\mathbb E_\rho H_\alpha\{\eta_B(p,q)\}\right|\\
 &\quad\le2\mathbb E_\rho d_B(p,q)
   +\mathbb E_\rho\min\left\{1,
       \frac2{\sqrt{2\pi}}|\mu_B(p,q)-\eta_B(p,q)|\right\}
 =o(1).
\end{aligned}
\]
The expectation bound uses $0\le d_B\le1$, $0\le H_\alpha\le1$
and $\sup_t|H_\alpha'(t)|\le2/\sqrt{2\pi}$.
Since $H_\alpha$ is even,
\begin{equation}\label{proof:blend-family-power}
 \mathbb P_\rho\{|N_B(p,q)|>z_\alpha\}
 =\mathbb E_\rho H_\alpha\left(
  \kappa_R\sqrt B\rho F_B(p,q)|\mathbf a^\top\mathbf v|\right)+o(1).
\end{equation}

\proofanchor{proof:blend-optimization}
For $t\ge0$, $H_\alpha'(t)=\varphi(z_\alpha-t)-\varphi(z_\alpha+t)\ge0$.
Thus, for a fixed direction rule, maximizing $F_B(p,q)$ maximizes
the leading power in \eqref{proof:blend-family-power}.
With $\mathbf b_B=(1,\tau J_{-,B})^\top$, Cauchy--Schwarz gives
\[
 F_B(p,q)^2
 =\frac{(\mathbf u^\top\mathbf b_B)^2}{\mathbf u^\top\mathbf V_B\mathbf u}
 \le\mathbf b_B^\top\mathbf V_B^{-1}\mathbf b_B,
 \qquad
 \mathbf V_B^{-1}\mathbf b_B
 =\frac1{\mathcal H_B-\tau}
   \begin{pmatrix}(1-\tau)J_{-,B}\\\mathcal H_B-1\end{pmatrix}.
\]
Equality holds when $\mathbf u$ is a positive multiple of
$\mathbf V_B^{-1}\mathbf b_B$, whose entries are nonnegative.
Set $P_B=J_{+,B}^{-1}$, $p=(1-t)P_B$ and $q=t$. Matching the component
ratios gives
\begin{equation}\label{proof:blend-mixing-optimum}
 \frac{tJ_{+,B}}{1-t}
 =\frac{\mathcal H_B-1}{(1-\tau)J_{-,B}},
 \qquad
 t_B^*=\frac{\mathcal H_B-1}{(2-\tau)\mathcal H_B-1}.
\end{equation}
Moreover,
\begin{equation*}
\begin{aligned}
 F_B\{(1-t_B^*)P_B,t_B^*\}^2
 &=\mathbf b_B^\top\mathbf V_B^{-1}\mathbf b_B\\
 &=J_{-,B}\frac{\tau\mathcal H_B+1-2\tau}{\mathcal H_B-\tau}
 =J_{-,B}\left\{\tau+\frac{(1-\tau)^2}{\mathcal H_B-\tau}\right\}
 =J_B^{\mathrm{blend}}.
\end{aligned}
\end{equation*}

It remains to replace $(P_B,t_B^*)$ by the estimates in
\eqref{eq:blendestimates}--\eqref{eq:blendmixing}.
Apply \eqref{proof:blend-average-bound} to
$(E_b/d)\mathbf I_m$ and $\widehat\pi_b\mathbf I_m$.
Since $E_b>0$ almost surely for every $b$,
\begin{equation}\label{proof:blend-estimation-rate}
 \widehat J_{+,B}-J_{+,B}=O_P(\epsilon_B),\qquad
 \widehat J_{-,B}-J_{-,B}=O_P(\epsilon_B).
\end{equation}
The inequalities
\[
 \begin{gathered}
 s_-^2\le J_{+,B}\le s_+^2,\qquad
 s_+^{-2}\le J_{-,B}\le s_-^{-2},\\
 |\max(1,x)-\max(1,y)|\le|x-y|,\\
 \frac{d}{dH}\frac{H-1}{(2-\tau)H-1}
 =\frac{1-\tau}{\{(2-\tau)H-1\}^2}\le\frac1{1-\tau}
 \quad(H\ge1)
 \end{gathered}
\]
imply
\begin{equation*}
 \widehat P_B-P_B=O_P(\epsilon_B),\qquad
 \widehat{\mathcal H}_B-\mathcal H_B=O_P(\epsilon_B),\qquad
 \widehat t_B-t_B^*=O_P(\epsilon_B).
\end{equation*}
Define $\widehat p_B=(1-\widehat t_B)\widehat P_B$,
$\widehat q_B=\widehat t_B$ and the event
$\mathcal F_B=\{s_-^2/2\le\widehat J_{+,B}\le2s_+^2\}$.
By \eqref{proof:blend-estimation-rate}, $\mathbb P_\rho(\mathcal F_B)\to1$.
Because $0\le\widehat t_B<1/(2-\tau)$, on $\mathcal F_B$,
\begin{equation*}
 (\widehat p_B,\widehat q_B)\in\mathcal K
 =\left[\frac{1-\tau}{2(2-\tau)s_+^2},\frac2{s_-^2}\right]
    \mathbin{\times}\left[0,\frac1{2-\tau}\right].
\end{equation*}
The population pair $((1-t_B^*)P_B,t_B^*)$ also belongs to
$\mathcal K$. Equation~\eqref{proof:blend-G} bounds the denominator
of $F_B$ away from zero on $\mathcal K$, so
\[
 \sup_{(p,q)\in\mathcal K}\|\nabla F_B(p,q)\|\le C,
 \qquad
 F_B(\widehat p_B,\widehat q_B)
 =\sqrt{J_B^{\mathrm{blend}}}+O_P(\epsilon_B).
\]
Apply the uniform bounds \eqref{proof:blend-general-mean}--\eqref{proof:blend-general-cdf}
on $\mathcal F_B$, whose complement has probability $o(1)$. We obtain
\begin{equation*}
 \mathbb E_\rho(N_B^{\mathrm{blend}}\mid\mathcal I_B)
 =\kappa_R\sqrt B\rho\sqrt{J_B^{\mathrm{blend}}}\,
                     \mathbf a^\top\mathbf v
       +O_P(\sqrt B\rho^3+\rho).
\end{equation*}
For $\sqrt B\rho^3\to0$, the change in the expected power function is at most
\[
 \mathbb E_\rho\min\left\{1,
 C\sqrt B\rho\left|F_B(\widehat p_B,\widehat q_B)
                         -\sqrt{J_B^{\mathrm{blend}}}\right|\right\}
 =o(1),
\]
since $\sqrt B\rho\epsilon_B=\rho+\sqrt B\rho^3\to0$.
\eqref{proof:blend-family-power} and
$|\mathbf a^\top\mathbf v|=\sqrt{1-L_B}$ now prove
\eqref{eq:blendpower}.

For Radius, $F_B(1,0)^2=J_{+,B}^{-1}=J_B^{\mathrm{rad}}$
in \eqref{proof:blend-family-power}. For the known-scale coefficients
$\lambda_b=r_b/s_b^2$, define
\[
 \mathbf K_B=\frac1B\sum_b\frac{\mathbf C_b}{s_b^4},\qquad
 D_B^{\mathrm{known}}=B\mathbf a^\top\mathbf K_B\mathbf a.
\]
Let $N_B^{\mathrm{known}}=(D_B^{\mathrm{known}})^{-1/2}
\sum_{b\mid r_b>0}(r_b/s_b^2)Z_b$, and let
$\mu_B^{\mathrm{known}}$ and $(\sigma_B^{\mathrm{known}})^2$
be its conditional mean and variance given $\mathcal I_B$.
Gaussian moments and block independence give
\[
 \mathbb E_\rho\mathbf K_B
 =2J_{-,B}\mathbf I_m+
   \rho^2\left(\frac1B\sum_bs_b^{-4}\right)\mathbf v\mathbf v^\top,
 \qquad
 \mathbb E_\rho\|\mathbf K_B-\mathbb E_\rho\mathbf K_B\|_F^2\le C/B.
\]
Thus $\mathbf K_B=2J_{-,B}\mathbf I_m+O_P(\epsilon_B)$ and
$cB\le D_B^{\mathrm{known}}\le CB$ with probability tending to one.
Equation~\eqref{proof:blend-phase-mean} yields
\[
\begin{aligned}
 \mu_B^{\mathrm{known}}
 &=\beta\sqrt B\rho
   \frac{\mathbf a^\top\mathbf K_B\mathbf v}
        {\sqrt{\mathbf a^\top\mathbf K_B\mathbf a}}
       +O_P(\sqrt B\rho^3K_{2,C,B})\\
 &=\kappa_R\sqrt B\rho\sqrt{J_{-,B}}\,
     \mathbf a^\top\mathbf v+O_P(\sqrt B\rho^3+\rho).
\end{aligned}
\]
Equations \eqref{proof:blend-phase-var}--\eqref{proof:blend-phase-third}
likewise give the explicit variance and third-moment bounds
\[
\begin{aligned}
 |(\sigma_B^{\mathrm{known}})^2-1|
 &\le\frac{C\rho^2}{B}\sum_bE_b^2e^{C\rho\sqrt{E_b}}
 =O_P(\rho^2),\\
 \frac{\sum_{b\mid r_b>0}(r_b/s_b^2)^3
          \mathbb E_\rho(|Z_b-\bar z_b|^3\mid\mathcal I_B)}
 {(D_B^{\mathrm{known}})^{3/2}(\sigma_B^{\mathrm{known}})^3}
 &\le\frac C{B^{3/2}}\sum_bE_b^{3/2}e^{C\rho\sqrt{E_b}}
 =O_P(B^{-1/2}).
\end{aligned}
\]
By the Berry--Esseen inequality and \eqref{proof:pow-normal-scale},
$d_B^{\mathrm{known}}=O_P(B^{-1/2}+\rho^2)$, where
$d_B^{\mathrm{known}}=\sup_x|\mathbb P_\rho(N_B^{\mathrm{known}}
\le x\mid\mathcal I_B)-\Phi(x-\mu_B^{\mathrm{known}})|$.
For $\eta_B^{\mathrm{known}}=
\kappa_R\sqrt B\rho\sqrt{J_{-,B}}\mathbf a^\top\mathbf v$,
\[
 \left|\mathbb P_\rho(|N_B^{\mathrm{known}}|>z_\alpha)
     -\mathbb E_\rho H_\alpha(\eta_B^{\mathrm{known}})\right|
 \le2\mathbb E_\rho d_B^{\mathrm{known}}
   +\mathbb E_\rho\min\{1,C|\mu_B^{\mathrm{known}}-
                               \eta_B^{\mathrm{known}}|\}=o(1).
\]
Finally, direct subtraction gives
\begin{align*}
 J_B^{\mathrm{blend}}-J_B^{\mathrm{rad}}
 &=J_{-,B}\frac{\tau(\mathcal H_B-1)^2}
                  {\mathcal H_B(\mathcal H_B-\tau)}\ge0,
 \\
 J_{-,B}-J_B^{\mathrm{blend}}
 &=J_{-,B}\frac{(1-\tau)(\mathcal H_B-1)}
                  {\mathcal H_B-\tau}\ge0,
\end{align*}
which proves \eqref{eq:blendcomparison} and the known-scale assertion.
\end{proof}

Theorem~\ref{thm:blendpower} allows any invariant direction.
For the normalized-scatter choice, the following proof evaluates
its alignment distribution when $\rho_B=h/\sqrt B$ and obtains
the explicit limits in Remark~\ref{rem:blendlocalpower}.

\proofanchor{proof:blend-local}
\begin{proof}[Proof of Remark~\ref{rem:blendlocalpower}]
Let $\rho_B=h/\sqrt B$. Apply \eqref{proof:info-one-likelihood} and
\eqref{proof:info-one-cross} to
$(\mathbf X_b/s_b,\mathbf Y_b/s_b)$ with signal $\rho_B/s_b$.
The map $\mathbf M_b\mapsto s_b^{-2}\mathbf M_b$ is invertible, so
the likelihood ratio of the original invariant matrices is
\[
 \mathcal L_B^{(s)}
 =\prod_b e^{-\rho_B^2/(2s_b^2)}
    I_0\left(\frac{\rho_B}{s_b^2}
        \sqrt{\mathbf v^\top\mathbf C_b\mathbf v}\right),
 \qquad
 \mathbb E_0\{(\mathcal L_B^{(s)})^2\}
 =\prod_b I_0(\rho_B^2/s_b^2).
\]
Write $Q_{B,\rho_B,\mathbf v}^{(s)}$ for the invariant law.
Cauchy--Schwarz and \eqref{proof:info-bessel-bound} imply
\begin{equation*}
\begin{aligned}
 \delta_B^{(s)}
 &=d_{\mathrm{TV}}(Q_{B,\rho_B,\mathbf v}^{(s)},Q_{B,0}^{(s)})\\
 &\le\frac12\sqrt{
    \exp\left(\frac{\rho_B^4}{4}\sum_b s_b^{-4}\right)-1}
 \le\frac12\sqrt{\exp\{h^4/(4Bs_-^4)\}-1}\longrightarrow0.
\end{aligned}
\end{equation*}
The distance is unchanged upon adjoining independent parameter-free $V$.
By the measurable-map inequality for total variation and
\eqref{proof:det-null-beta},
\[
 d_{\mathrm{TV}}\left\{
 \mathcal L_{\rho_B,\mathbf v}(1-L_B),
 \operatorname{Beta}\left(\tfrac12,\tfrac{m-1}{2}\right)\right\}
 \le\delta_B^{(s)}.
\]
Let $U\sim\operatorname{Beta}\{1/2,(m-1)/2\}$ and
$k\in\{\mathrm{rad},\mathrm{blend}\}$. Since $J_B^k\to j_k>0$,
\[
\begin{aligned}
 &\left|\mathbb E_{\rho_B,\mathbf v}
 H_\alpha(\kappa_Rh\sqrt{J_B^k}\sqrt{1-L_B})
   -\mathbb E H_\alpha(\kappa_Rh\sqrt{j_kU})\right|\\
 &\quad\le\delta_B^{(s)}
   +\frac{2\kappa_Rh}{\sqrt{2\pi}}
                |\sqrt{J_B^k}-\sqrt{j_k}|\longrightarrow0.
\end{aligned}
\]
Equation~\eqref{eq:blendpower} for Blend and
\eqref{proof:blend-family-power} with $(p,q)=(1,0)$ for Radius give
\[
 \mathbb P_{\rho_B,\mathbf v}(|N_B^k|>z_\alpha)
 \longrightarrow\mathbb E H_\alpha(\kappa_Rh\sqrt{j_kU}).
\]
Finally, for $0<\alpha<1$ and $t>0$,
\[
 H_\alpha'(t)=\varphi(z_\alpha-t)-\varphi(z_\alpha+t)>0.
\]
Thus \eqref{eq:blendcomparison} and $\mathbb P(0<U<1)=1$
give the limiting power ordering, with strict inequality when
$h>0$ and $j_{\mathrm{blend}}>j_{\mathrm{rad}}$.
\end{proof}

\section{Application to heterogeneous functional data}
\label{proof:cal-functional}

The proof of Proposition~\ref{prop:functionalcalibration} requires two steps.
We verify the matched Gaussian law in \eqref{eq:functionalmatching},
then derive product-rotation symmetry and apply
Theorem~\ref{thm:joint}.

\begin{proof}[Proof of Proposition~\ref{prop:functionalcalibration}]
As in Section~\ref{sec:functional}, all sums omit terms
with $q_{sh}=0$. For each remaining $(s,h,b)$, define the feature summaries
\[
 \mathbf A_{shb}=\frac1{n_{shb}}\sum_i
                      \boldsymbol\Phi\boldsymbol\varepsilon_{shbi},
 \qquad
 \mathbf B_{shb}=\frac1{\sqrt{n_{shb}}}\sum_i
                   u_{shbi}\boldsymbol\Phi\boldsymbol\varepsilon_{shbi}.
\]
Because $\boldsymbol\Phi$ is fixed, bounded and linear,
\eqref{eq:functionalerror} makes all feature summaries jointly
Gaussian conditional on $\Lambda$. Fix $(s,h,b)$ and abbreviate
$n=n_{shb}$, $u_i=u_{shbi}$ and
$\mathbf K^\Phi=\boldsymbol\Phi\mathbf K_{shb}(\Lambda)\boldsymbol\Phi^*$.
Conditional independence of the errors and
\eqref{eq:residualcoeff} give
\begin{equation}\label{proof:cal-functional-covariances}
 \begin{aligned}
 \operatorname{Cov}(\mathbf A_{shb}\mid\Lambda)&=\mathbf K^\Phi/n,\\
 \operatorname{Cov}(\mathbf B_{shb}\mid\Lambda)
   &=\frac{\mathbf K^\Phi}{n}\sum_i u_i^2=\mathbf K^\Phi/n,\\
 \operatorname{Cov}(\mathbf A_{shb},\mathbf B_{shb}\mid\Lambda)
   &=\frac{\mathbf K^\Phi}{n\sqrt n}\sum_i u_i=\mathbf0.
 \end{aligned}
\end{equation}
The zero-sum condition in \eqref{eq:residualcoeff} gives
\[
 \boldsymbol\Phi\overline{\mathbf Z}_{shb}=\boldsymbol\Phi\boldsymbol\mu_{sh}+\mathbf A_{shb},
 \qquad
 \frac1{\sqrt{n_{shb}}}\sum_i u_{shbi}\boldsymbol\Phi\mathbf Z_{shbi}
 =\mathbf B_{shb}.
\]
Substituting into \eqref{eq:functionalsummaries} and
\eqref{eq:functionalpairs} yields
\[
 \mathbf X_b=\boldsymbol\delta+
                 \sum_{s,h}c_sq_{sh}\mathbf A_{shb},
 \qquad
 \mathbf Y_b=\sum_{s,h}c_sq_{sh}\mathbf B_{shb}.
\]
Conditional independence across $(s,h,b)$ and
\eqref{proof:cal-functional-covariances} therefore imply
\begin{equation}\label{proof:cal-functional-joint-law}
 \begin{pmatrix}\mathbf X_b\\\mathbf Y_b\end{pmatrix}\Bigm|\Lambda
 \sim\mathcal N_{2m}\left\{
   \begin{pmatrix}\boldsymbol\delta\\\mathbf0\end{pmatrix},
   \begin{pmatrix}\boldsymbol\Sigma_b(\Lambda)&\mathbf0\\
                   \mathbf0&\boldsymbol\Sigma_b(\Lambda)\end{pmatrix}
 \right\},
 \qquad
 \boldsymbol\Sigma_b(\Lambda)=\boldsymbol\Phi
  \left\{\sum_{s,h}\frac{c_s^2q_{sh}^2}{n_{shb}}
                  \mathbf K_{shb}(\Lambda)\right\}\boldsymbol\Phi^*.
\end{equation}
Equation~\eqref{proof:cal-functional-joint-law} proves
\eqref{eq:functionalmatching}--\eqref{eq:functionalcov}.
The stacked pairs are conditionally independent across $b$.

Under $H_0^{(\Phi)}$, $\boldsymbol\delta=\mathbf0$.
For $\mathbf t_b,\mathbf s_b\in\mathbb R^m$,
\begin{equation}\label{proof:cal-functional-characteristic}
 \mathbb E_0\left[\exp\left\{i\sum_b
 (\mathbf t_b^\top\mathbf X_b+\mathbf s_b^\top\mathbf Y_b)\right\}
 \middle|\Lambda\right]
 =\exp\left[-\frac12\sum_b
 \{\mathbf t_b^\top\boldsymbol\Sigma_b\mathbf t_b+
   \mathbf s_b^\top\boldsymbol\Sigma_b\mathbf s_b\}\right],
\end{equation}
where $\boldsymbol\Sigma_b=\boldsymbol\Sigma_b(\Lambda)$.
A rotation $\mathbf G_b\mapsto e^{i\phi_b}\mathbf G_b$ replaces
$(\mathbf t_b,\mathbf s_b)$ by
\[
 \mathbf t'_b=\mathbf t_b\cos\phi_b-\mathbf s_b\sin\phi_b,
 \qquad
 \mathbf s'_b=\mathbf t_b\sin\phi_b+\mathbf s_b\cos\phi_b.
\]
For every $b$,
\[
 (\mathbf t'_b)^\top\boldsymbol\Sigma_b\mathbf t'_b+
 (\mathbf s'_b)^\top\boldsymbol\Sigma_b\mathbf s'_b
 =\mathbf t_b^\top\boldsymbol\Sigma_b\mathbf t_b+
  \mathbf s_b^\top\boldsymbol\Sigma_b\mathbf s_b.
\]
Hence the conditional characteristic function
\eqref{proof:cal-functional-characteristic} is unchanged.
Independence of $V$ from the observations and $\Lambda$ gives
\begin{equation}\label{proof:cal-functional-rotation}
 \mathcal L_0(e^{i\phi_1}\mathbf G_1,\ldots,e^{i\phi_B}\mathbf G_B,V
                 \mid\Lambda)
 =\mathcal L_0(\mathbf G_1,\ldots,\mathbf G_B,V\mid\Lambda).
\end{equation}
Integrating \eqref{proof:cal-functional-rotation} over $\Lambda$
proves \eqref{eq:product} unconditionally.

Set $D=\sum_{b=1}^B\lambda_b^2$, with $\lambda_b=0$ for omitted
blocks. The coefficients and $D$ are $\mathcal I_B$-measurable.
Applying Theorem~\ref{thm:joint} to
\eqref{proof:cal-functional-rotation}, with $\Lambda$ included in
the conditioning information, gives on $\{D>0\}$
\[
 \mathcal L_0(N_B\mid\mathcal I_B,\Lambda)=\mathcal N(0,1).
\]
With $p_{\mathrm{comb}}=1$ on $\{D=0\}$, for $0<u<1$,
\[
\begin{aligned}
 \mathbb P_0(p_{\mathrm{comb}}\le u\mid\mathcal I_B,\Lambda)
 &=u\mathbf1_{\{D>0\}},\\
 \mathbb P_0(p_{\mathrm{comb}}\le u\mid\mathcal I_B)
 &=\mathbb E_0\{u\mathbf1_{\{D>0\}}\mid\mathcal I_B\}
 =u\mathbf1_{\{D>0\}}\le u.
\end{aligned}
\]
Corollary~\ref{cor:representative-calibration} verifies the
required measurability for Radius and Blend.
\end{proof}

\section{Supplementary simulation results}
\label{supp:factorial-study}

The simulation settings are specified in
Section~\ref{subsec:functionalsimulation}.
For the split procedures, define the training and testing sets by
\[
  \mathcal T=\{1,\ldots,B/2\},
  \qquad \mathcal E=\{B/2+1,\ldots,B\},
\]
and let
\[
  \mathcal J_B
  =\sigma\{(\mathbf G_b)_{b\in\mathcal T},
           (\mathbf M_b)_{b\in\mathcal E},V\}.
\]
Under the null, rotations of the testing blocks preserve
the joint law while leaving the training observations unchanged.
The selected direction and testing coefficients are
$\mathcal J_B$-measurable.
Theorem~\ref{thm:joint} therefore gives conditionally
independent standard normal scores for the nonzero projected
pairs. Their normalized combination is also standard normal
conditional on $\mathcal J_B$, whenever its denominator is positive.

Table~\ref{supp:tab-factorial-null} summarizes null rejection
frequencies across the factorial settings. All four procedures
remain close to the nominal level of $.05$ across the noise
distributions, covariance structures and block scales considered.
Figures~\ref{supp:fig-factorial-gaussian-identity-heterogeneous}--
\ref{supp:fig-factorial-contaminated-band-heterogeneous}
complete the power comparisons in the main paper.
Rows correspond to $m$ and columns to $B$, as in the main paper,
with pointwise
95\% Wilson intervals.
The full-data procedures generally outperform their splits,
particularly with few blocks, while Blend benefits from
block-scale heterogeneity. Splitting can be more powerful
in some settings with larger $B$.

In order to examine how power depends on the position of the signal
relative to the noise covariance, we also consider the case of $\mathbf v=\mathbf e_1$ or $\mathbf e_8$, placing the
signal in a coordinate with smaller or larger noise variance,
respectively. Here $\mathbf e_j$ is the $j$th coordinate unit
vector. Block variances are either all one or alternate between
$1/4$ and $7/4$.  For simplicity, we set $m=8$, $B=16$ and
\(
  \boldsymbol\Sigma
  =\operatorname{diag}\left(\tfrac12\mathbf I_4,
                            \tfrac32\mathbf I_4\right).
\)
Independently within and across blocks, generate
\[
  \mathbf X_b\sim\mathcal N_8
  \left(\frac{h}{\sqrt B}\mathbf v,s_b^2\boldsymbol\Sigma\right),
  \qquad
  \mathbf Y_b\sim\mathcal N_8
  \left(\mathbf0,s_b^2\boldsymbol\Sigma\right).
\]
We use the signal grid and replication count in Section~\ref{subsec:functionalsimulation}.

Figure~\ref{supp:fig-factorial-direction} shows that signal
location can reverse the power ordering. With homogeneous
blocks, splitting is more powerful for the low-variance
signal, whereas the full-data procedures are more powerful
for the high-variance signal. Blend also improves on Radius
under heterogeneous block scales. Moving the signal changes
both its relation to the learned direction and its noise
variance, so these comparisons reflect their combined effects.

\begin{table}[!htbp]
\centering
\caption{Ranges of null rejection frequencies across all dimensions
and block counts within each factorial configuration. Each range
contains twelve settings, each with 2,000 repetitions. H and V denote
homogeneous and heterogeneous block scales.}
\label{supp:tab-factorial-null}
\scriptsize
\setlength{\tabcolsep}{2.6pt}
\begin{tabular}{lllrrrr}
\toprule
Noise & Covariance & Scales & Radius & Blend & Split-Radius & Split-Blend\\
\midrule
Gaussian & identity & H & 0.0415--0.0585 & 0.0410--0.0590 & 0.0380--0.0585 & 0.0395--0.0585\\
Gaussian & identity & V & 0.0420--0.0630 & 0.0455--0.0585 & 0.0420--0.0605 & 0.0425--0.0635\\
Gaussian & diagonal & H & 0.0450--0.0610 & 0.0440--0.0590 & 0.0435--0.0590 & 0.0435--0.0585\\
Gaussian & diagonal & V & 0.0480--0.0595 & 0.0430--0.0610 & 0.0435--0.0595 & 0.0445--0.0560\\
Gaussian & band & H & 0.0435--0.0575 & 0.0435--0.0560 & 0.0430--0.0590 & 0.0410--0.0580\\
Gaussian & band & V & 0.0435--0.0570 & 0.0445--0.0570 & 0.0415--0.0540 & 0.0440--0.0585\\
Student-$t$ & identity & H & 0.0420--0.0600 & 0.0395--0.0585 & 0.0440--0.0620 & 0.0470--0.0575\\
Student-$t$ & identity & V & 0.0460--0.0595 & 0.0410--0.0530 & 0.0440--0.0565 & 0.0410--0.0540\\
Student-$t$ & diagonal & H & 0.0425--0.0570 & 0.0465--0.0565 & 0.0440--0.0570 & 0.0415--0.0580\\
Student-$t$ & diagonal & V & 0.0420--0.0550 & 0.0465--0.0545 & 0.0430--0.0510 & 0.0380--0.0550\\
Student-$t$ & band & H & 0.0445--0.0590 & 0.0355--0.0585 & 0.0425--0.0565 & 0.0430--0.0570\\
Student-$t$ & band & V & 0.0415--0.0565 & 0.0420--0.0540 & 0.0475--0.0535 & 0.0390--0.0550\\
Contaminated Gaussian & identity & H & 0.0425--0.0585 & 0.0420--0.0550 & 0.0425--0.0590 & 0.0415--0.0585\\
Contaminated Gaussian & identity & V & 0.0445--0.0565 & 0.0435--0.0585 & 0.0420--0.0605 & 0.0370--0.0650\\
Contaminated Gaussian & diagonal & H & 0.0470--0.0600 & 0.0445--0.0585 & 0.0470--0.0580 & 0.0420--0.0550\\
Contaminated Gaussian & diagonal & V & 0.0435--0.0635 & 0.0415--0.0620 & 0.0390--0.0590 & 0.0365--0.0555\\
Contaminated Gaussian & band & H & 0.0435--0.0585 & 0.0410--0.0550 & 0.0415--0.0575 & 0.0410--0.0545\\
Contaminated Gaussian & band & V & 0.0380--0.0535 & 0.0395--0.0580 & 0.0470--0.0625 & 0.0385--0.0595\\
\bottomrule
\end{tabular}
\end{table}

\begin{figure}[p]
\centering
\includegraphics[width=\linewidth,height=.70\textheight,keepaspectratio]{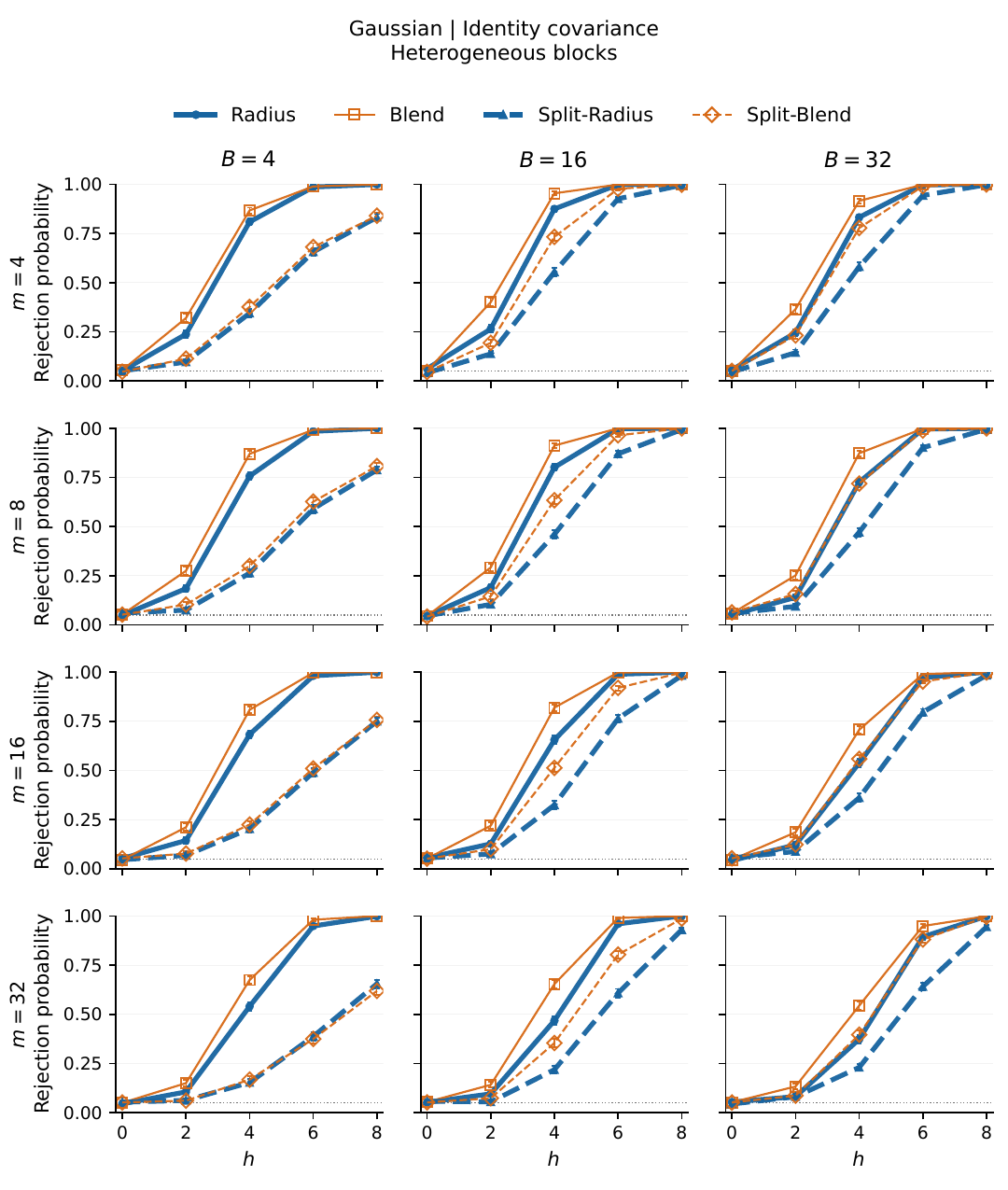}
\caption{Gaussian noise, identity covariance and heterogeneous block scales.}
\label{supp:fig-factorial-gaussian-identity-heterogeneous}
\end{figure}

\begin{figure}[p]
\centering
\includegraphics[width=\linewidth,height=.70\textheight,keepaspectratio]{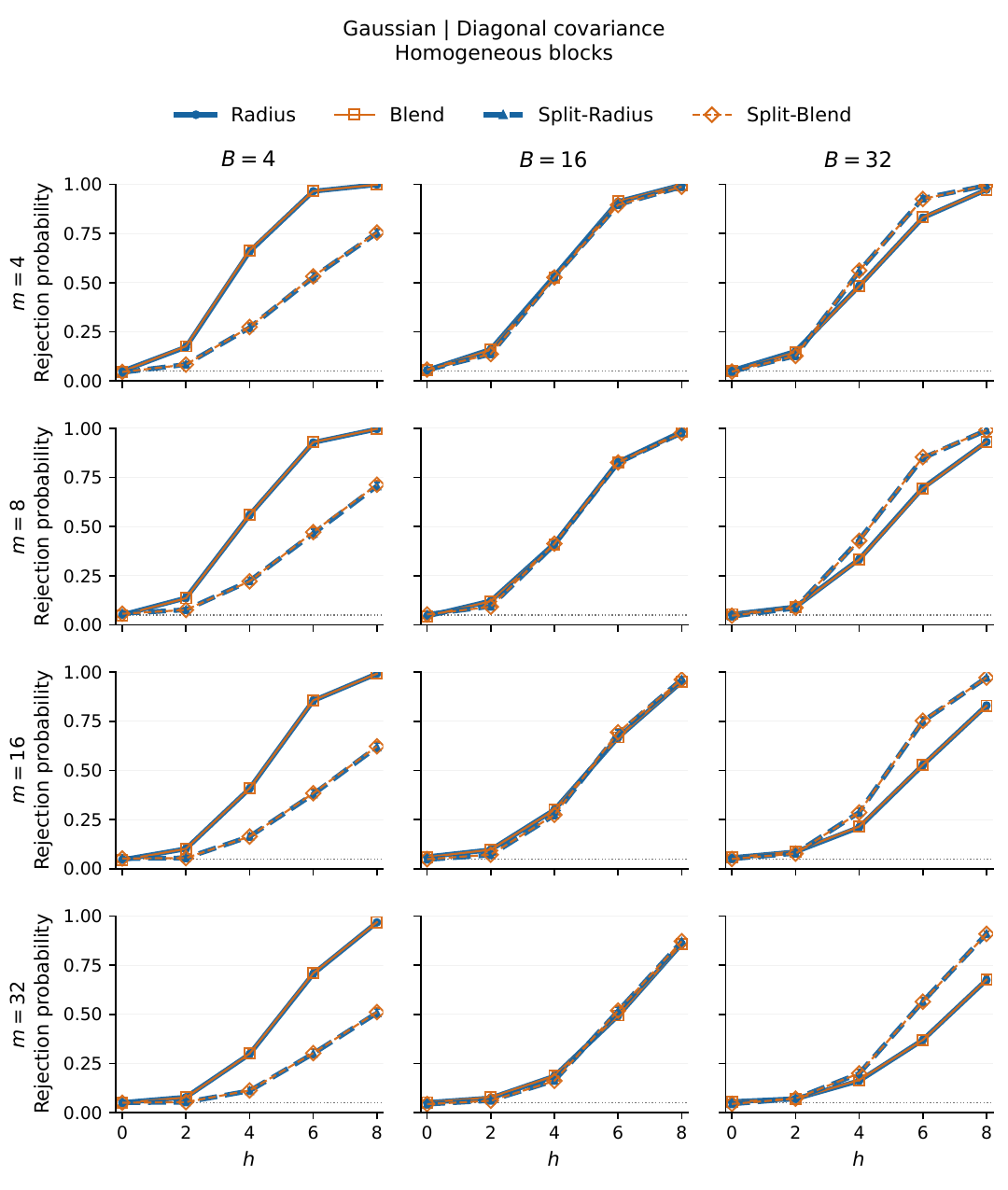}
\caption{Gaussian noise, diagonal covariance and homogeneous block scales.}
\label{supp:fig-factorial-gaussian-diagonal-homogeneous}
\end{figure}

\begin{figure}[p]
\centering
\includegraphics[width=\linewidth,height=.70\textheight,keepaspectratio]{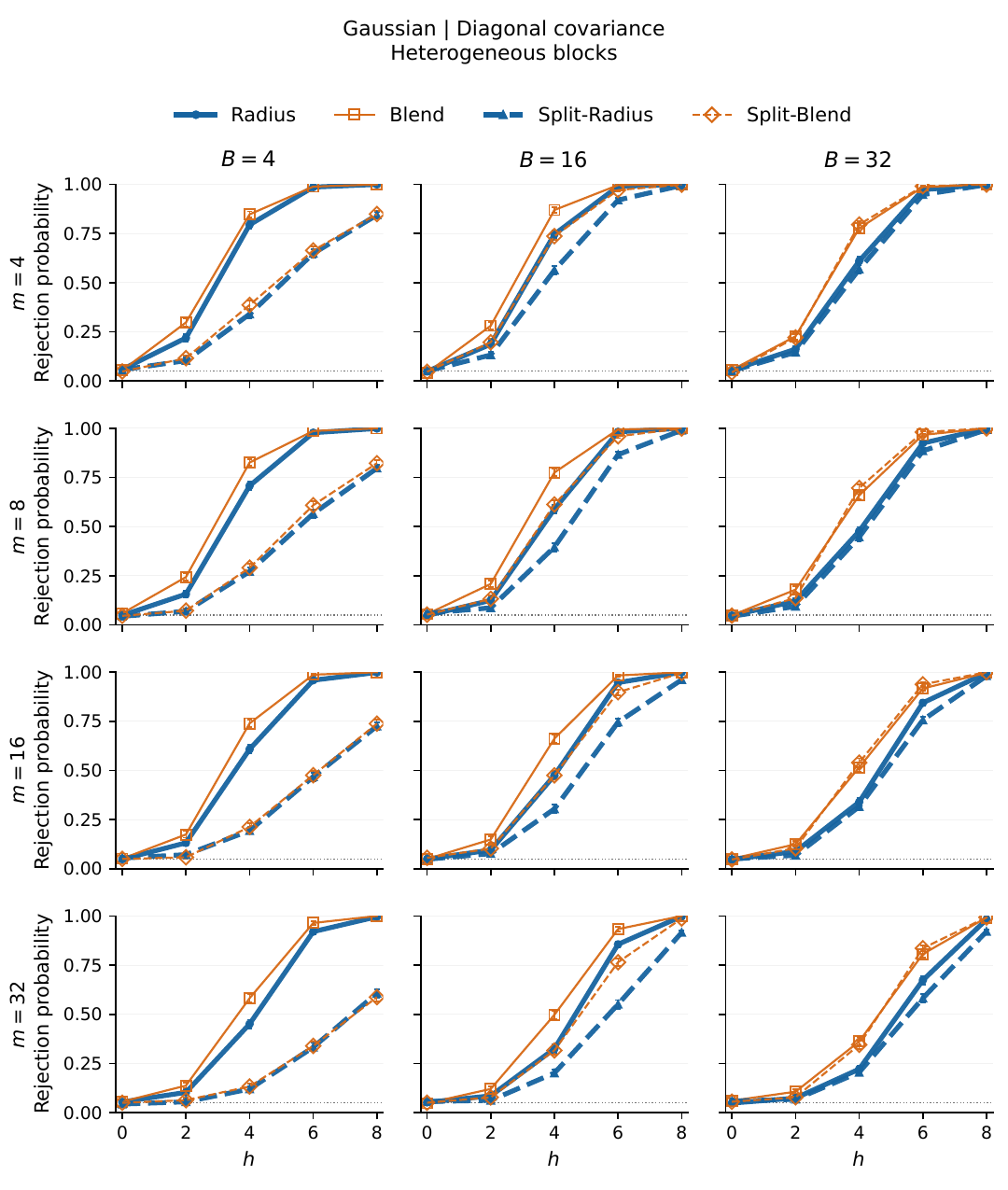}
\caption{Gaussian noise, diagonal covariance and heterogeneous block scales.}
\label{supp:fig-factorial-gaussian-diagonal-heterogeneous}
\end{figure}

\begin{figure}[p]
\centering
\includegraphics[width=\linewidth,height=.70\textheight,keepaspectratio]{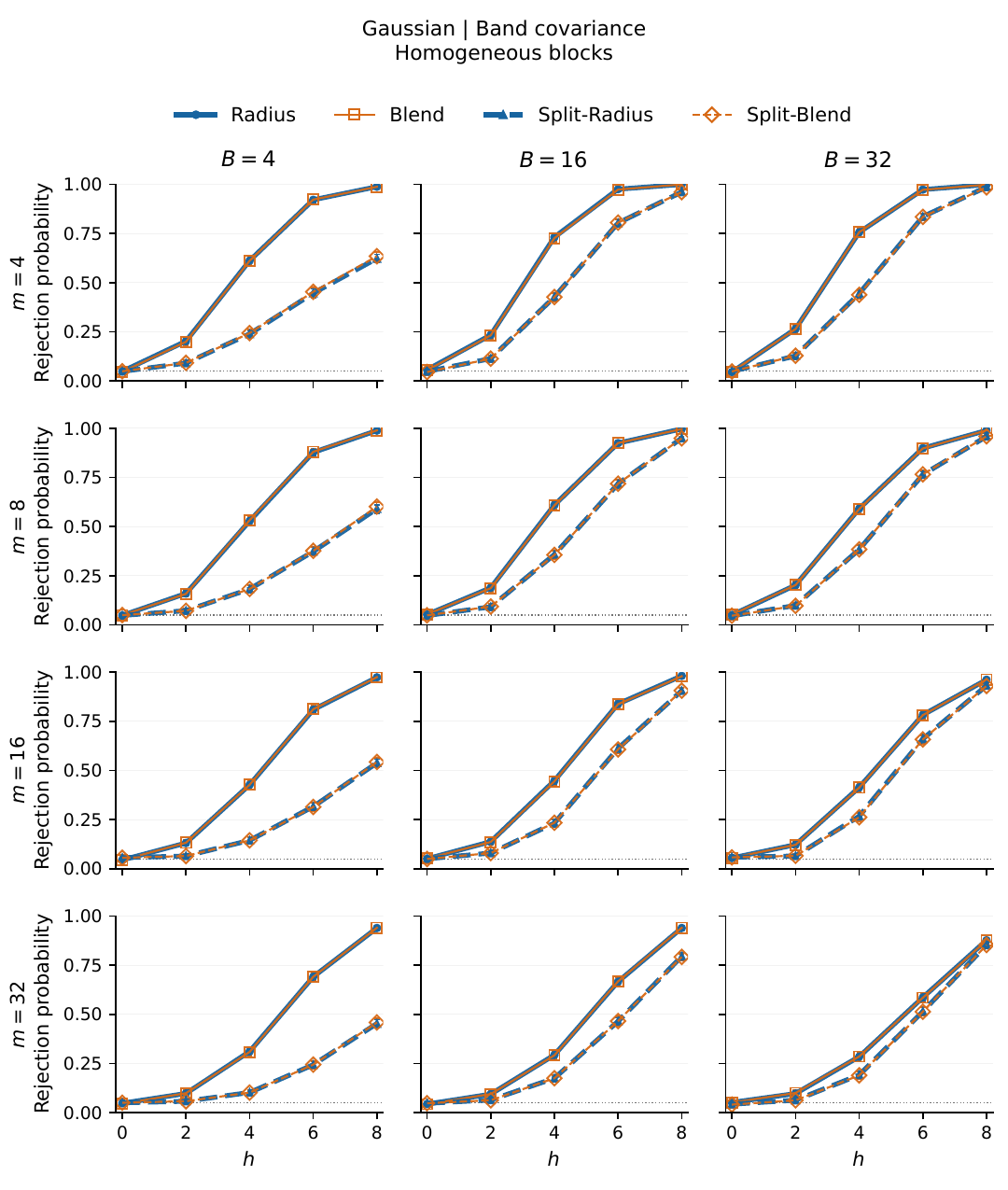}
\caption{Gaussian noise, band covariance and homogeneous block scales.}
\label{supp:fig-factorial-gaussian-band-homogeneous}
\end{figure}

\begin{figure}[p]
\centering
\includegraphics[width=\linewidth,height=.70\textheight,keepaspectratio]{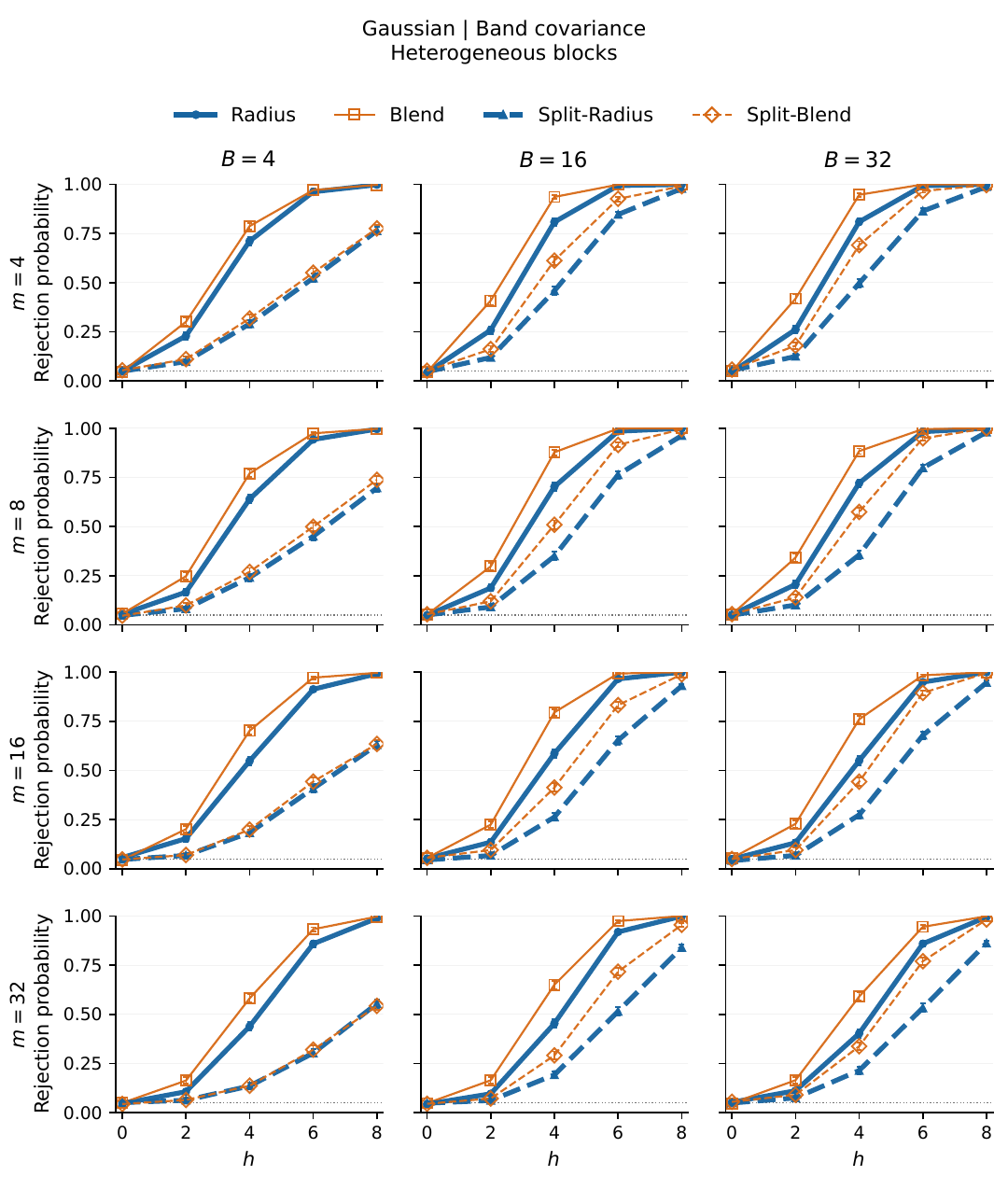}
\caption{Gaussian noise, band covariance and heterogeneous block scales.}
\label{supp:fig-factorial-gaussian-band-heterogeneous}
\end{figure}

\begin{figure}[p]
\centering
\includegraphics[width=\linewidth,height=.70\textheight,keepaspectratio]{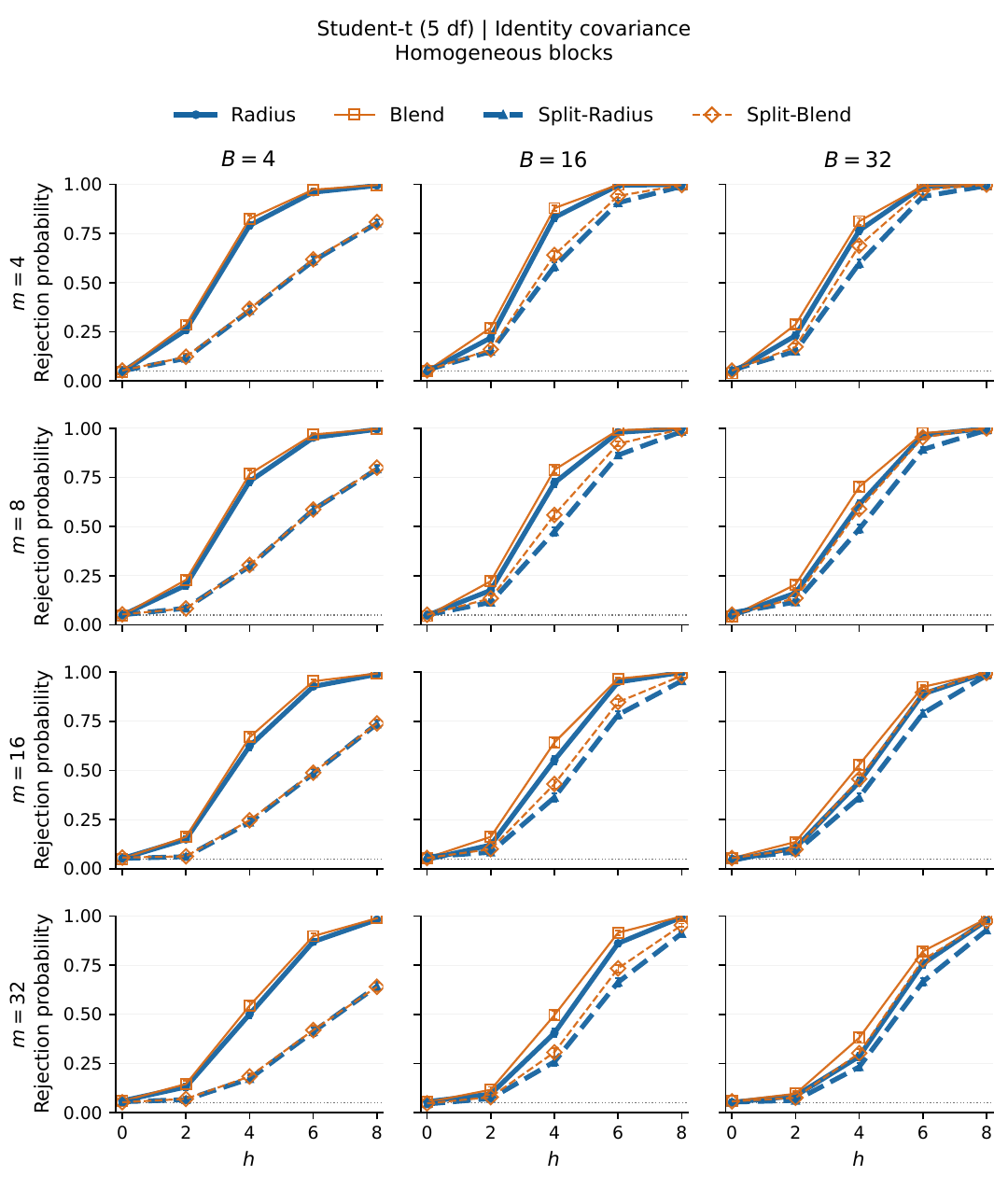}
\caption{Student-$t$ noise, identity covariance and homogeneous block scales.}
\label{supp:fig-factorial-student5-identity-homogeneous}
\end{figure}

\begin{figure}[p]
\centering
\includegraphics[width=\linewidth,height=.70\textheight,keepaspectratio]{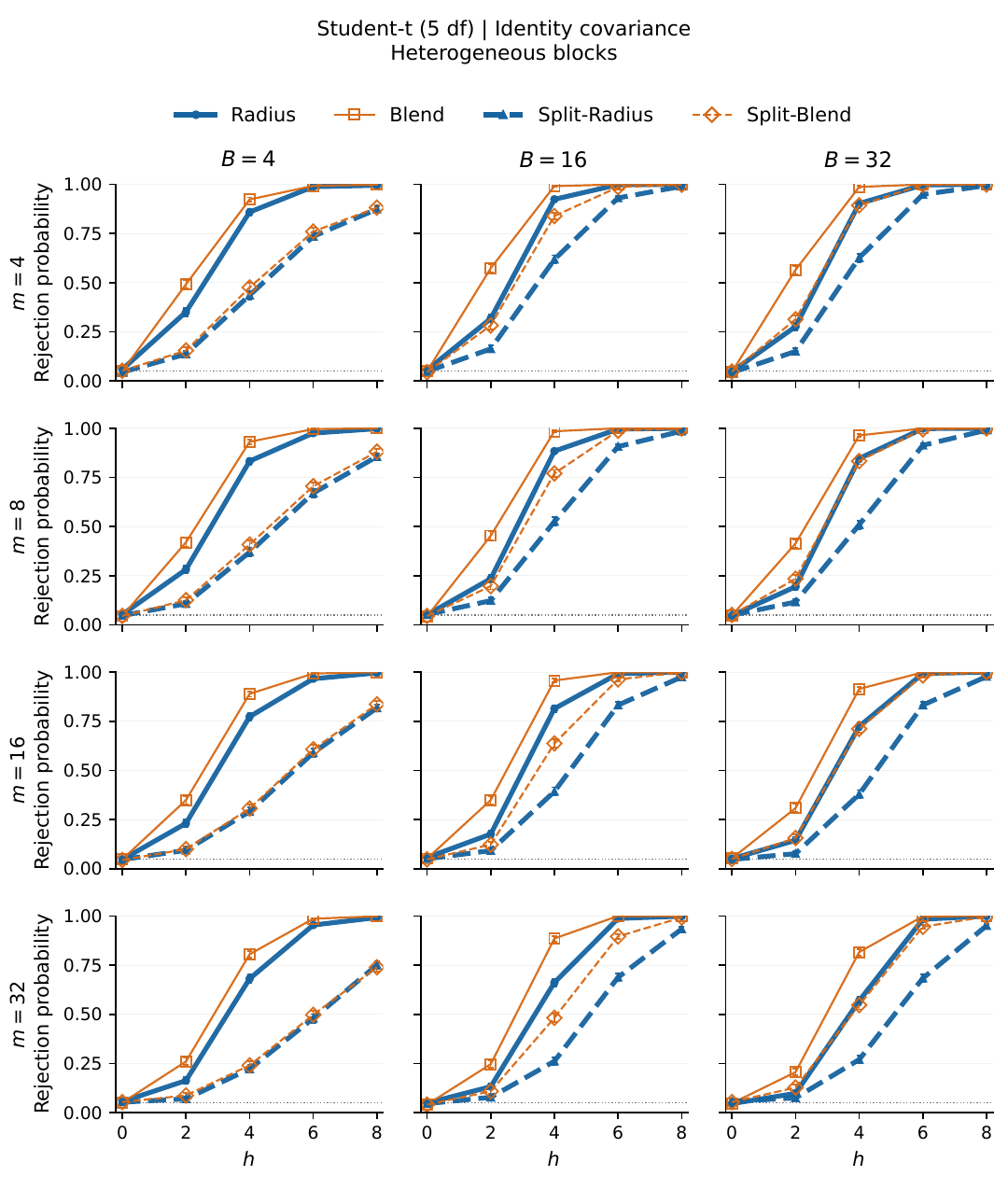}
\caption{Student-$t$ noise, identity covariance and heterogeneous block scales.}
\label{supp:fig-factorial-student5-identity-heterogeneous}
\end{figure}

\begin{figure}[p]
\centering
\includegraphics[width=\linewidth,height=.70\textheight,keepaspectratio]{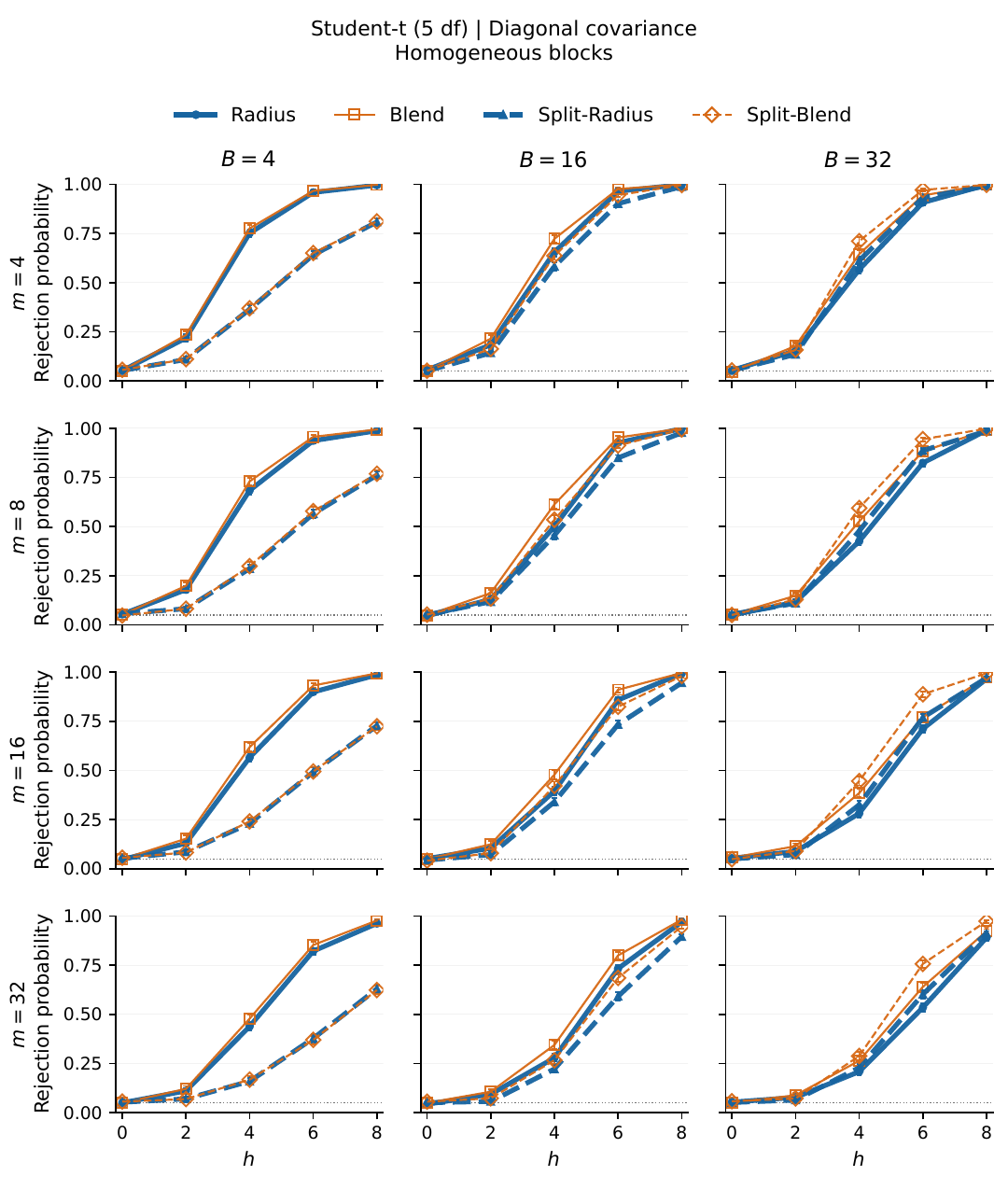}
\caption{Student-$t$ noise, diagonal covariance and homogeneous block scales.}
\label{supp:fig-factorial-student5-diagonal-homogeneous}
\end{figure}

\begin{figure}[p]
\centering
\includegraphics[width=\linewidth,height=.70\textheight,keepaspectratio]{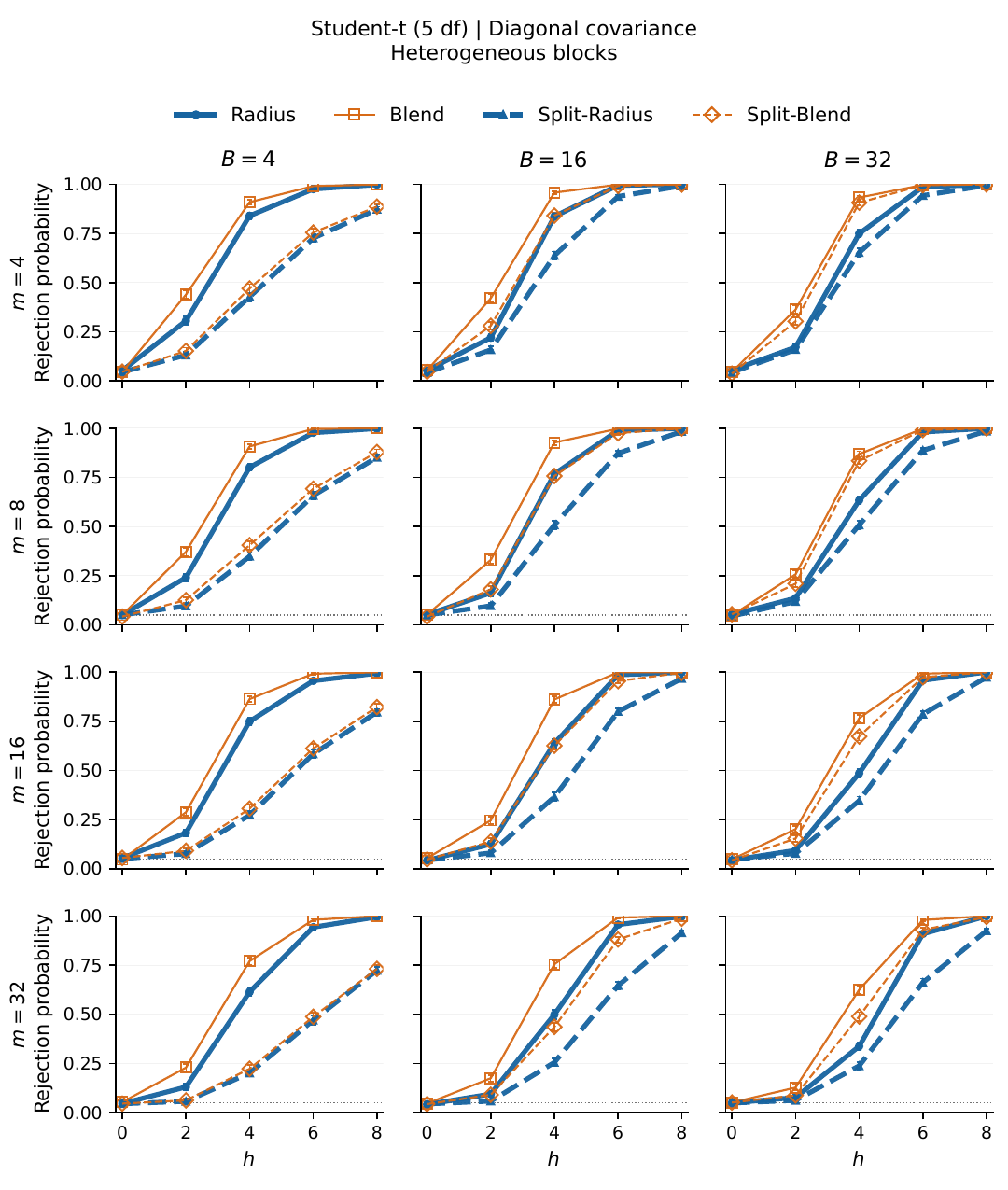}
\caption{Student-$t$ noise, diagonal covariance and heterogeneous block scales.}
\label{supp:fig-factorial-student5-diagonal-heterogeneous}
\end{figure}

\begin{figure}[p]
\centering
\includegraphics[width=\linewidth,height=.70\textheight,keepaspectratio]{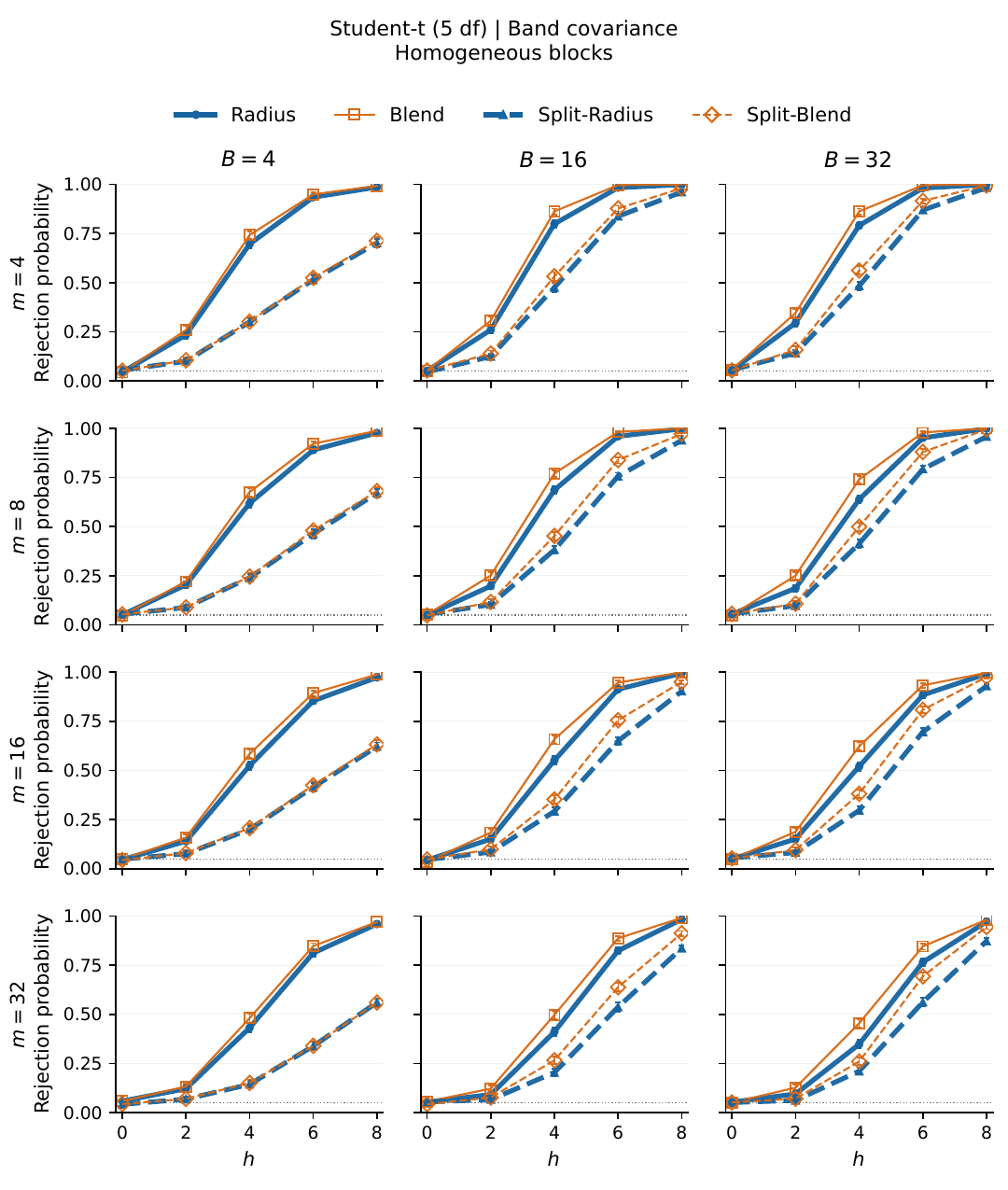}
\caption{Student-$t$ noise, band covariance and homogeneous block scales.}
\label{supp:fig-factorial-student5-band-homogeneous}
\end{figure}

\begin{figure}[p]
\centering
\includegraphics[width=\linewidth,height=.70\textheight,keepaspectratio]{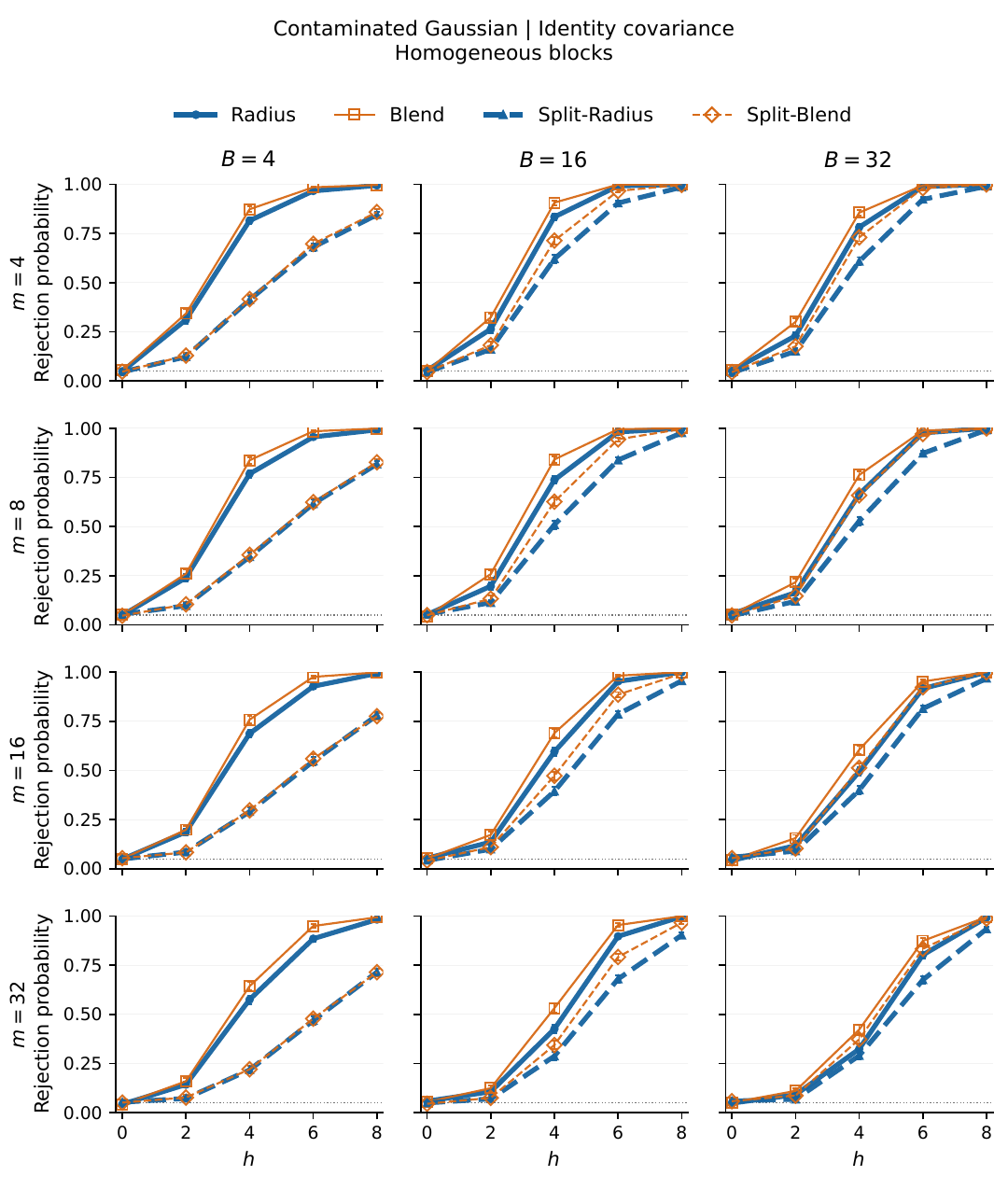}
\caption{Contaminated Gaussian noise, identity covariance and homogeneous block scales.}
\label{supp:fig-factorial-contaminated-identity-homogeneous}
\end{figure}

\begin{figure}[p]
\centering
\includegraphics[width=\linewidth,height=.70\textheight,keepaspectratio]{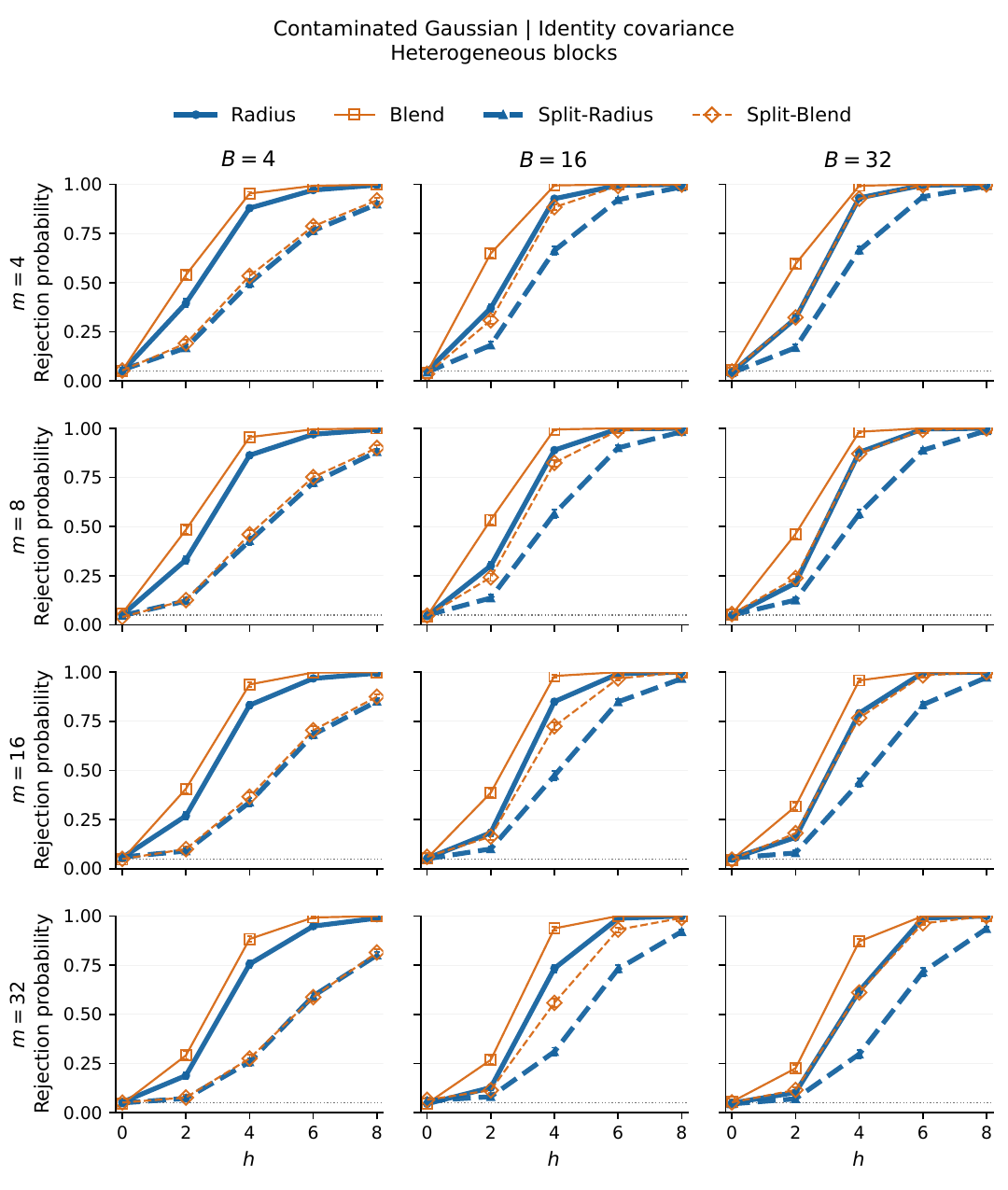}
\caption{Contaminated Gaussian noise, identity covariance and heterogeneous block scales.}
\label{supp:fig-factorial-contaminated-identity-heterogeneous}
\end{figure}

\begin{figure}[p]
\centering
\includegraphics[width=\linewidth,height=.70\textheight,keepaspectratio]{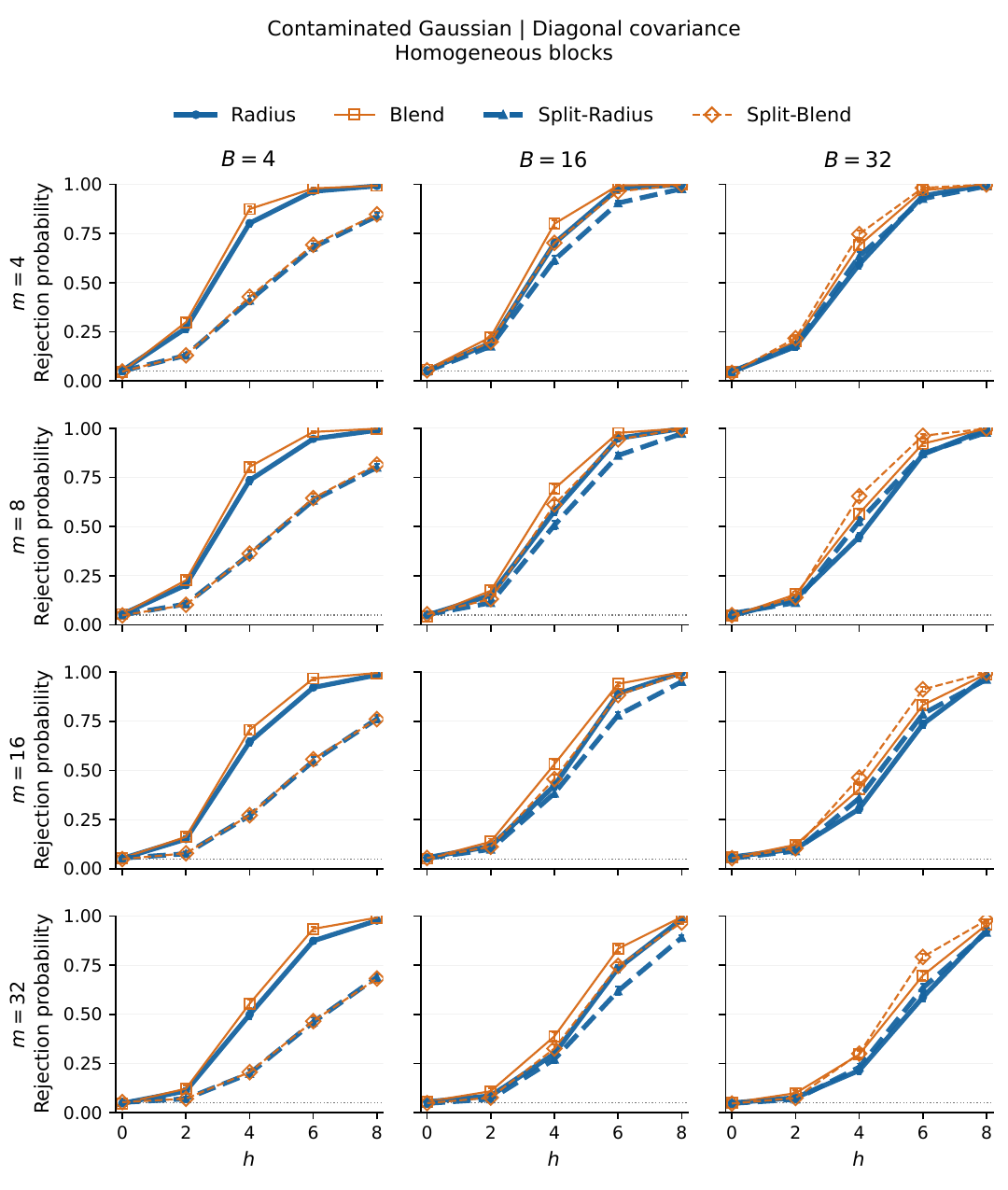}
\caption{Contaminated Gaussian noise, diagonal covariance and homogeneous block scales.}
\label{supp:fig-factorial-contaminated-diagonal-homogeneous}
\end{figure}

\begin{figure}[p]
\centering
\includegraphics[width=\linewidth,height=.70\textheight,keepaspectratio]{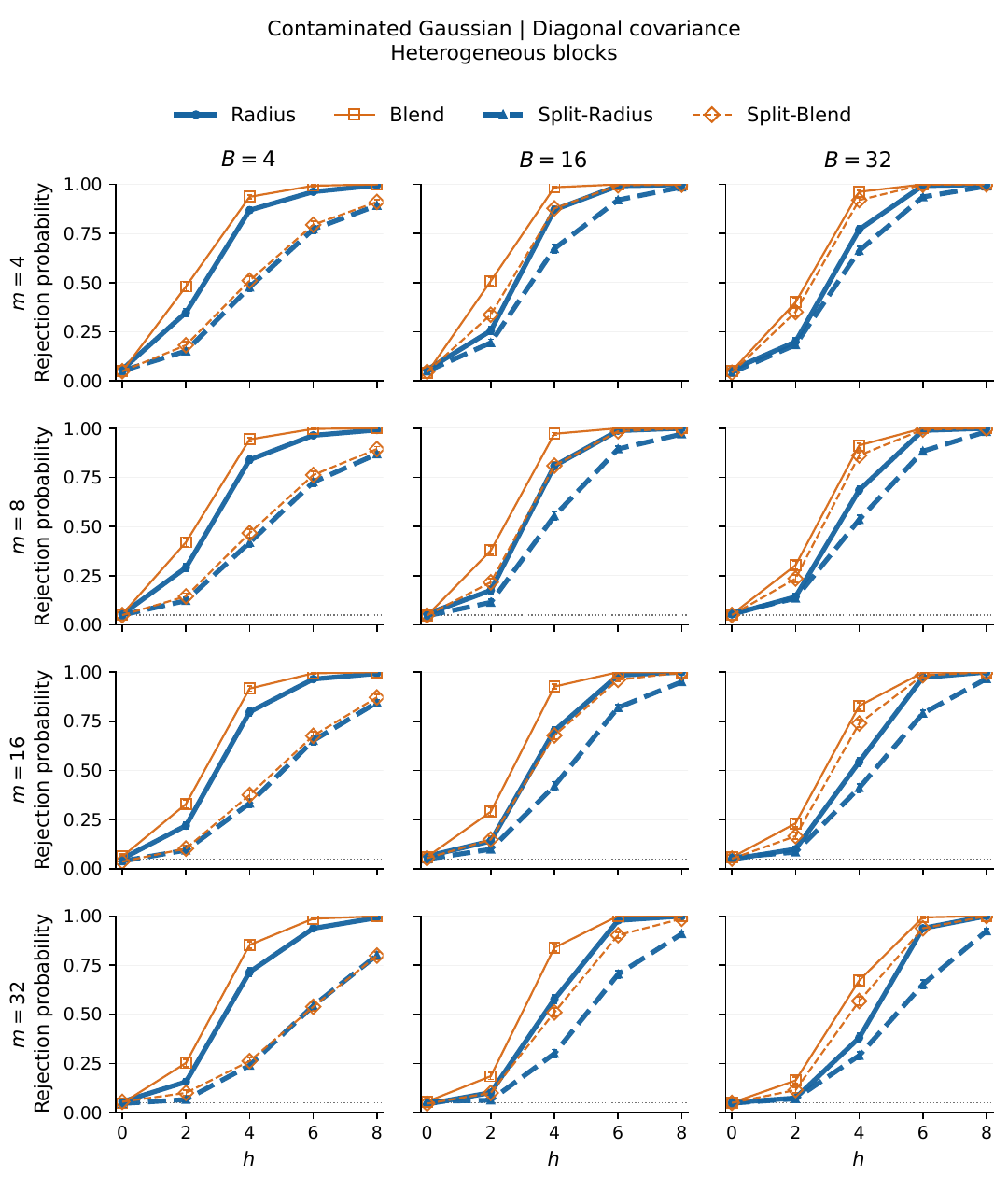}
\caption{Contaminated Gaussian noise, diagonal covariance and heterogeneous block scales.}
\label{supp:fig-factorial-contaminated-diagonal-heterogeneous}
\end{figure}

\begin{figure}[p]
\centering
\includegraphics[width=\linewidth,height=.70\textheight,keepaspectratio]{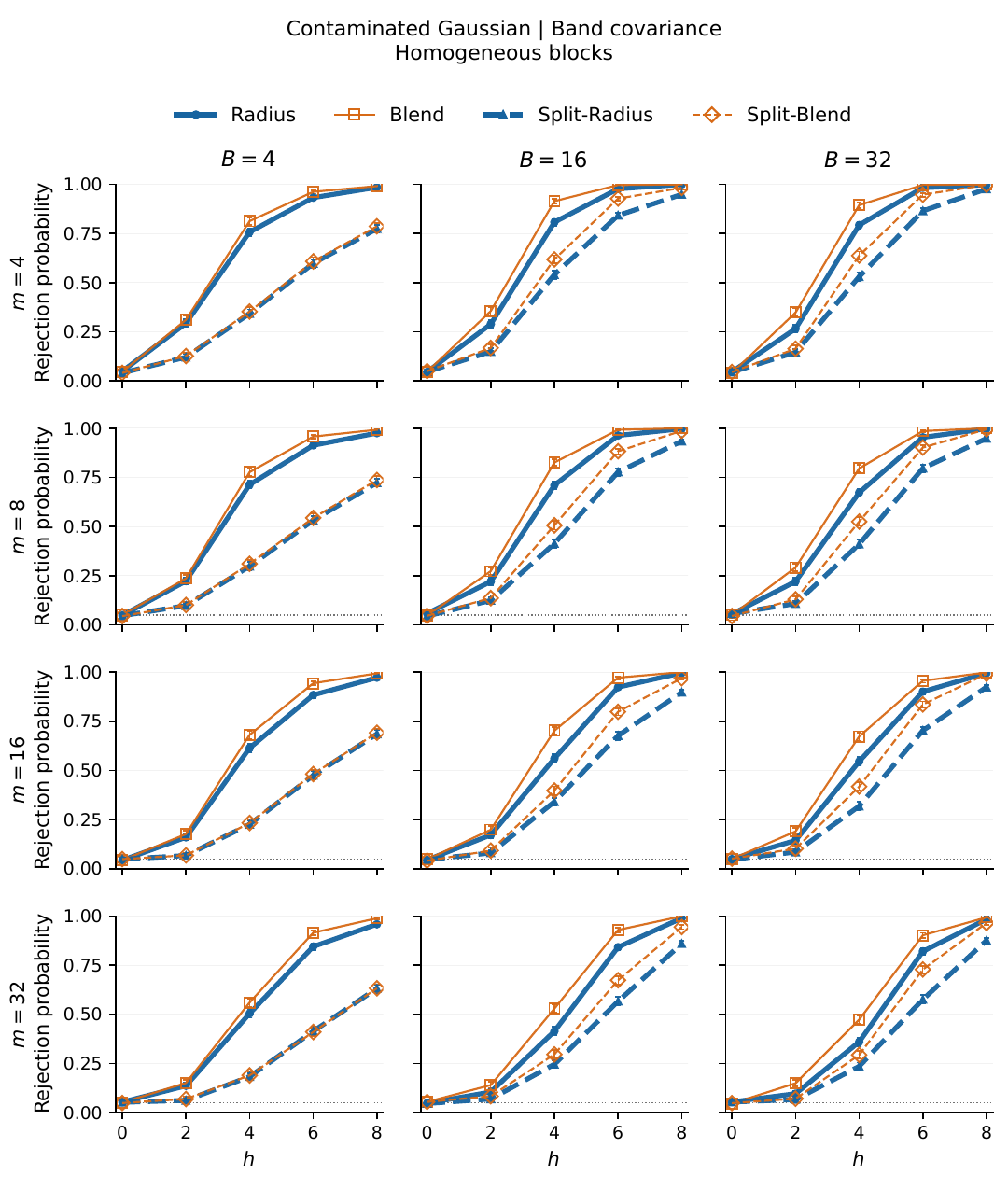}
\caption{Contaminated Gaussian noise, band covariance and homogeneous block scales.}
\label{supp:fig-factorial-contaminated-band-homogeneous}
\end{figure}

\begin{figure}[p]
\centering
\includegraphics[width=\linewidth,height=.70\textheight,keepaspectratio]{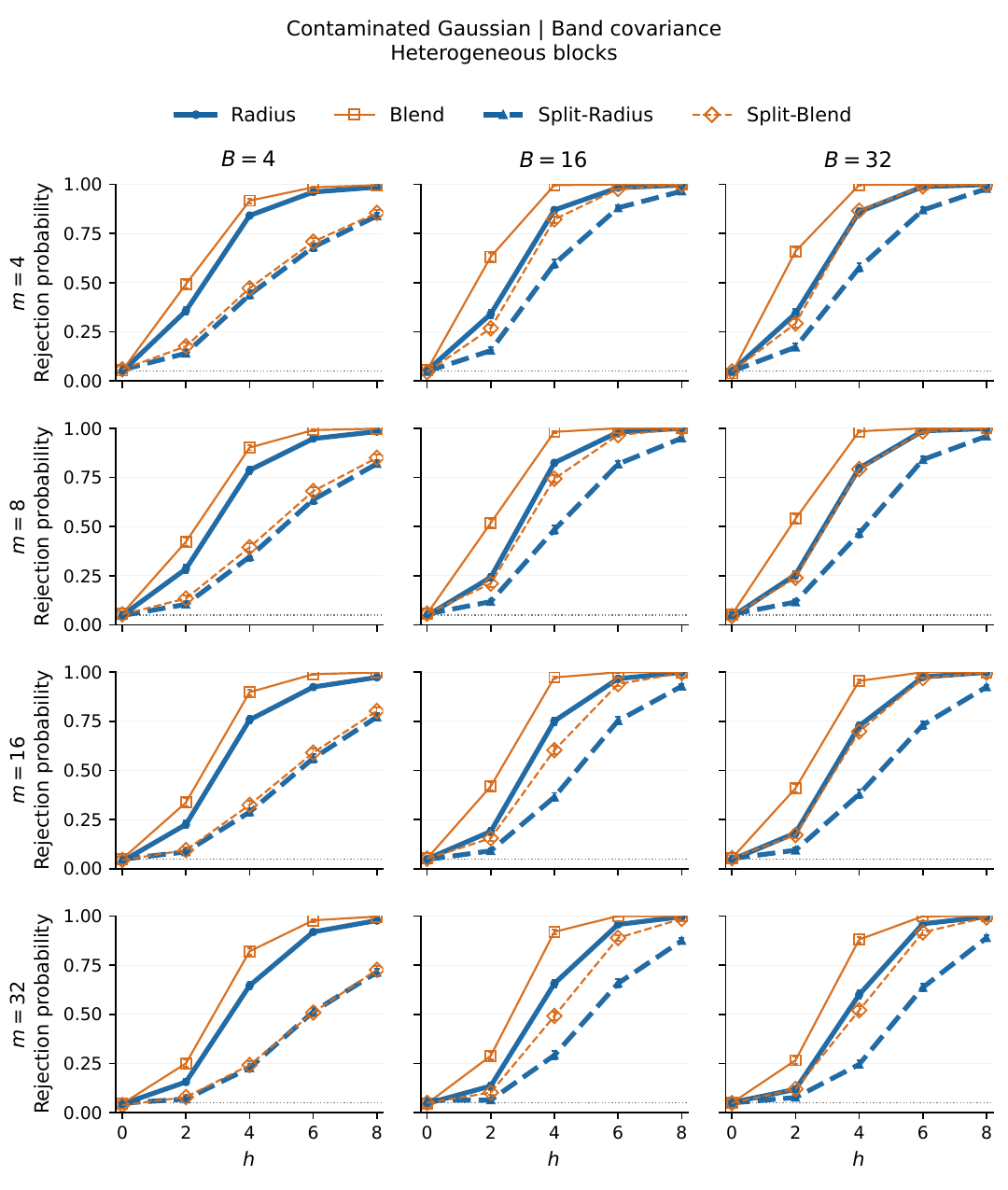}
\caption{Contaminated Gaussian noise, band covariance and heterogeneous block scales.}
\label{supp:fig-factorial-contaminated-band-heterogeneous}
\end{figure}

\begin{figure}[!htbp]
\centering
\includegraphics[width=\linewidth,height=.70\textheight,keepaspectratio]{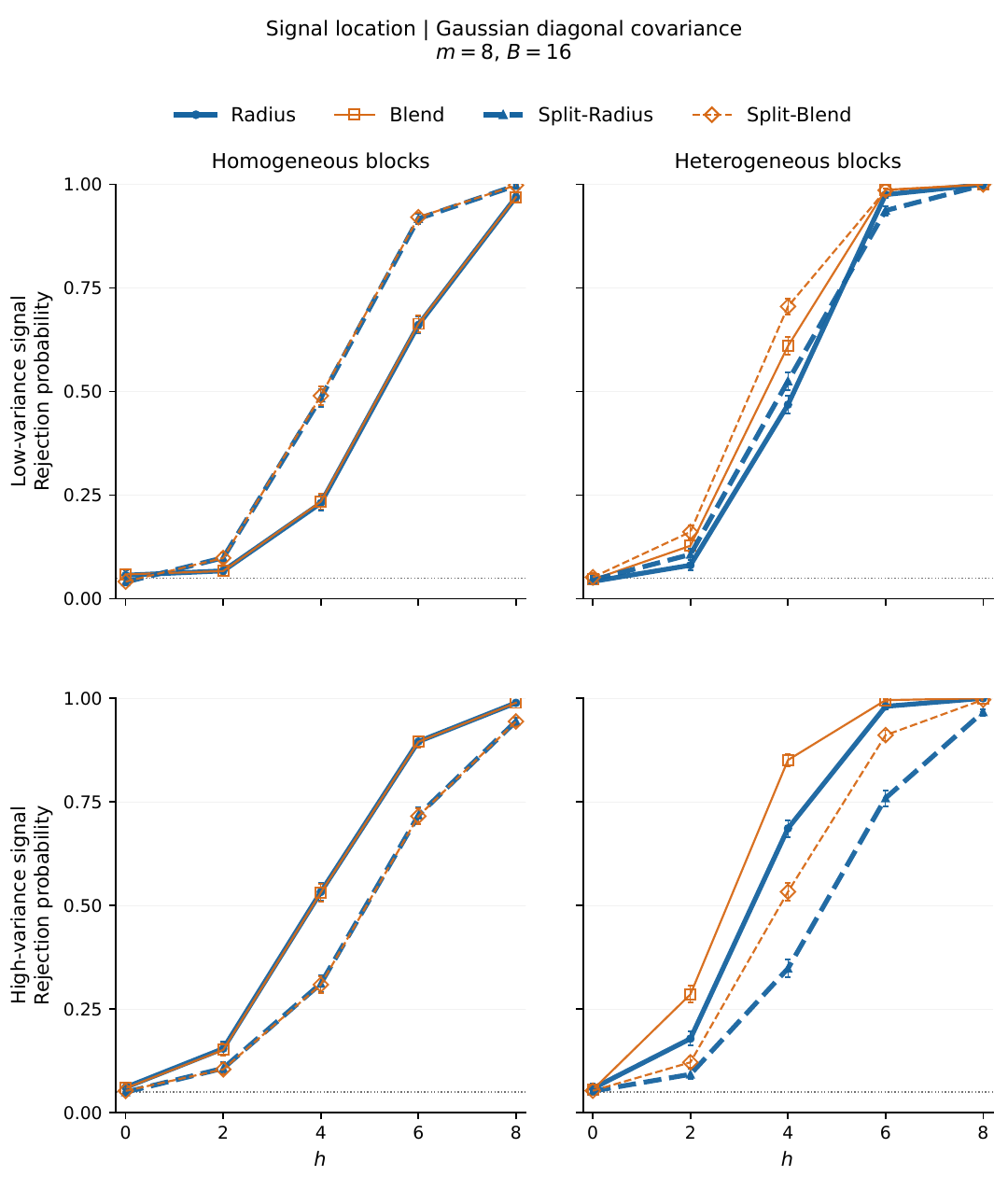}
\caption{Signal in a low-variance or high-variance coordinate, with Gaussian diagonal covariance, $m=8$ and $B=16$. Rows distinguish signals in low-variance and high-variance coordinates, and columns distinguish homogeneous and heterogeneous block scales. Every point uses 2,000 repetitions with pointwise 95\% Wilson intervals.}
\label{supp:fig-factorial-direction}
\end{figure}

\FloatBarrier
\section{Additional details for the ECG analysis}
\label{supp:ecg}

We use one 10-second recording per patient from PTB-XL and
Chapman--Shaoxing \citep{supp-Wagner2020,supp-Zheng2020,supp-Reyna2021}, with the
diagnostic groups defined in Section~\ref{subsec:ecg}.
The complete-feature pools contain 191 AF and 598 NSR patients
from PTB-XL, and 397 AF and 798 NSR patients from Chapman--Shaoxing.

Each patient has eight features. Four summarize the time intervals
between consecutive heartbeats, called RR intervals and measured
in seconds, using their
median, standard deviation, root mean square of successive
differences and coefficient of variation. The other four measure
the fractions of signal power in the 3--8 Hz and 8--20 Hz bands,
relative to the 0.5--40 Hz band, in two recording channels called
leads II and V1. We log-transform the median RR interval directly,
add $.002$ before logging each of the other three RR summaries,
and add $10^{-6}$ before logging each spectral ratio.
The transformed features are centred by their medians and scaled
by their interquartile ranges divided by 1.349, using 200 separate
PTB-XL NSR patients. This transformation is fixed for both sources.

Each group of 160 patients is randomly divided into eight blocks
of $k=20$. Write $\mathbf Z_{gbi}\in\mathbb R^8$ for the transformed
features of patient $i$ in block $b$, with $g=0$ denoting NSR
controls and $g=1$ the comparison group. We form
\begin{equation*}
 \mathbf X_b=\frac1k\sum_{i=1}^k
      (\mathbf Z_{1bi}-\mathbf Z_{0bi}),
 \qquad
 \mathbf Y_b=\frac1k\sum_{i=1}^k(-1)^{i-1}
      (\mathbf Z_{1bi}-\mathbf Z_{0bi}).
\end{equation*}
The first summary measures the group mean difference. The
alternating signs in the second remove the group means, giving
the residual summary used in Section~\ref{sec:functional}.

The null check uses 5,000 assignments of two disjoint NSR groups.
The dilution comparison uses 2,000 assignments per
$q\in\{.10,.20,.40,1\}$, with $160q$ AF patients and $160(1-q)$
NSR patients in the comparison group. All four methods share the
assignments and use $p<.05$. Patients appear only once within an
assignment but may reappear across assignments. The intervals in
Figure~\ref{fig:ecg-primary} describe Monte Carlo uncertainty
conditional on these fixed pools.

\end{document}